\documentclass[a4paper]{article}
\usepackage{amsmath}
\usepackage{amsthm}
\usepackage{mathrsfs}
\usepackage{amssymb}
\usepackage{amsmath}
\usepackage{mathcomp}
\usepackage{authblk}
\usepackage{url}
\usepackage{soul}
\usepackage[absolute,overlay]{textpos}
\usepackage{bbm}
\usepackage{bm}
\usepackage{authblk}
\usepackage{setspace}
\usepackage[usenames]{xcolor}
	\definecolor{gray1}{RGB}{230,230,230}
	\definecolor{gray2}{RGB}{200,200,200}
	\definecolor{gray3}{RGB}{160,160,160}
\usepackage[active]{srcltx}
\usepackage{stackengine}
\usepackage[bookmarksnumbered=true]{hyperref}
\usepackage[nottoc]{tocbibind}

\usepackage{graphicx}
\usepackage{float}
\usepackage{tikz}
\usepackage{mathabx}
\usepackage{bbm}

\newtheorem{definition}{Definition}[section]
\newtheorem{lemma}[definition]{Lemma}
\newtheorem{proposition}[definition]{Proposition}
\newtheorem{theorem}[definition]{Theorem}
\newtheorem{remark}[definition]{Remark}
\newtheorem{corollary}[definition]{Corollary}

\numberwithin{equation}{section}

\def\eps{\varepsilon}

\def\b{\mathcal{B}}
\newcommand{\fl}[1]{\lfloor #1  \rfloor}

\title{A supersymmetric hierarchical model for weakly disordered $3d$ semimetals}
\author[1]{Giovanni Antinucci}
\author[1]{Luca Fresta}
\author[2]{Marcello Porta}
\affil[1]{Institute of Mathematics, University of Zurich, Winterthurerstrasse 190, 8057 Zurich, Switzerland}
\affil[2]{Department of Mathematics, University of T\"ubingen, Auf der Morgenstelle 10, 72076 T\"ubingen, Germany}

\begin{document}

\maketitle

\begin{abstract}
In this paper we study a hierarchical supersymmetric model for a class of gapless, three-dimensional, weakly disordered quantum systems, displaying pointlike Fermi surface and conical intersections of the energy bands in the absence of disorder. We use rigorous renormalization group methods and supersymmetry to compute the correlation functions of the system. We prove algebraic decay of the two-point correlation function, compatible with delocalization. A main technical ingredient is the multiscale analysis of massless bosonic Gaussian integrations with purely imaginary covariances, performed via iterative stationary phase expansions.
\end{abstract}

\section{Introduction}

An important conjecture in mathematical quantum mechanics is that disordered, noninteracting, $3d$ quantum systems display a localization/delocalization transition as function of the disorder strength \cite{A, AALR}. The simplest model that is expected to give rise to such transition is the Anderson model, described by a random Schr\"odinger operator 
\begin{equation}\label{eq:anderson}
H_{\omega} = -\Delta + \gamma V_{\omega}\;,\qquad \text{on $\ell^{2}(\mathbb{Z}^{3})$}
\end{equation}
with $-\Delta$ the lattice Laplacian and $V_{\omega}$ a random potential, {\it e.g.} $(V_{\omega} \psi)(x) = \omega(x) \psi(x)$ with $\{ \omega(x) \}_{x\in \mathbb{Z}^{3}}$ i.i.d. random variables with variance $O(1)$.

From a mathematical viewpoint, a lot is known about this problem for {\it strong} disorder, $|\gamma| \gg 1$. There, one expects wave packets not to spread in time, and transport to be suppressed (zero conductivity). This phenomenon has been rigorously understood for general $d$-dimensional models starting from the seminal work \cite{FS}, where a KAM-type multiscale analysis approach to localization was developed, and later via the fractional moments method \cite{AM}. See \cite{AW} for a pedagogical review of mathematical results on Anderson localization.

Instead, for {\it small disorder} much less is known from a rigorous viewpoint. In three dimensions, one expects nontrivial transport, and an emergent diffusive behavior of the quantum dynamics. Unfortunately, so far no fully satisfactory rigorous result is available on this problem. Results have been obtained for tree graphs and similar structures, \cite{K, ASW, AW, FHS, BHKY, BKY, KS, Sa}. The analogous problem for random matrix models is much better understood, see \cite{EY} for a review of recent results. Concerning short ranged lattice models, important progress has been obtained in \cite{ESY1, ESY2, ESY3}, where diffusion for the Anderson model has been proven in the scaling limit, and in \cite{DSZ, DS}, where a localization/delocalization transition for a supersymmetric effective model has been established (see also \cite{DMR} for more recent extensions).

The starting point of \cite{DSZ, DS} is a mapping of the disorder-averaged correlations of the Anderson model into those of an interacting supersymmetric quantum field theory model. This mapping was first introduced in physics in \cite{E} (see also \cite{W}, for a related approach based on the replica trick), and allows to import field-theoretic methods to study random Schr\"odinger operators. Let us briefly describe it. Consider a general class of random Schr\"odinger operators, $H_{\omega} = H + \gamma V_{\omega}$, with $H$ a short-ranged lattice Schr\"odinger operator, on a finite sublattice $\Lambda$ of $\mathbb{Z}^{3}$. Let $G_{\omega}(x,y; \mu - i\varepsilon)$ be the Green's function:
\begin{equation}
G_{\omega}(x,y; \mu -  i\varepsilon) := \langle \delta_{x}, \frac{1}{H_{\omega} - \mu + i\varepsilon} \delta_{y}\rangle\;.
\end{equation}
The parameter $\mu\in \mathbb{R}$ plays the role of chemical potential, while $\varepsilon>0$ introduces a regularization of the Green's function. The relevant setting for weakly disordered metals is $\mu \in \sigma (H)$ (as $L\to \infty$). It is well known that the Green's function can be represented as the covariance of a Gaussian Grassmann field, as follows:
\begin{equation}\label{eq:Gres}
iG_{\omega}(x,y; \mu -  i\varepsilon) = \frac{\int\, [\prod_{x\in \Lambda} d\psi_{x}^{+} d\psi_{x}^{-}]\, e^{-(\psi^{+}, C_{\omega}^{-1} \psi^{-})} \psi^{-}_{x} \psi^{+}_{y}}{ \int\, [\prod_{x\in \Lambda} d\psi^{+}_{x} d\psi^{-}_{x}]\, e^{-(\psi^{+}, C_{\omega}^{-1} \psi^{-}) }}\;,
\end{equation}
with $C_{\omega}^{-1} := -i(H_{\omega} - \mu) + \varepsilon$; the reason for the multiplication by the trivial factor $i$ will be clear in a moment. The denominator is the determinant of the matrix $C^{-1}_{\omega}$, which is a random object; as a consequence, the expression (\ref{eq:Gres}) is not very useful for the purpose of computing the disorder average. 

The key remark is that the reciprocal of a determinant can be written as a {\it complex} Gaussian integral:
\begin{eqnarray}\label{eq:SUSY0}
&&iG_{\omega}(x,y; \mu -  i\varepsilon) \\
&&= \int [\prod_{x\in \Lambda} d\psi_{x}^{+} d\psi_{x}^{-}] [\prod_{x\in \Lambda} d\phi_{x}^{+} d\phi_{x}^{-}]\, e^{-(\psi^{+}, C_{\omega}^{-1} \psi^{-})} e^{-(\phi^{+}, C_{\omega}^{-1} \phi^{-})} \psi^{-}_{x} \psi^{+}_{y}\;.\nonumber
\end{eqnarray}
Suppose that $V_{\omega}(x) = \omega(x)$, with $\{ \omega(x) \}$ i.i.d. Gaussian variables with variance $1$. One has, by the Hubbard-Stratonovich formula:
\begin{eqnarray}\label{eq:SUSY}
&&i\mathbb{E}_{\omega} G_{\omega}(x,y; \mu -  i\varepsilon)\\
&&= \int \mathcal{D}[\psi, \phi]\, e^{-(\psi^{+}, C^{-1} \psi^{-})} e^{-(\phi^{+}, C^{-1} \phi^{-})} e^{-\lambda \sum_{x} (\phi^{+}_{x} \phi^{-}_{x} + \psi^{+}_{x} \psi^{-}_{x} )^{2}}\psi^{-}_{x} \psi^{+}_{y}\nonumber
\end{eqnarray}
where $ \mathcal{D}[\psi, \phi] = [\prod_{x\in \Lambda} d\psi_{x}^{+} d\psi_{x}^{-}] [\prod_{x\in \Lambda} d\phi_{x}^{+} d\phi_{x}^{-}]$, the inverse covariance is $C^{-1}$ $  = -i(H - \mu) + \varepsilon$ and $\lambda = \gamma^{2}/2$. The same trick can be applied to rewrite the average of the product of Green's functions, by introducing internal degrees of freedom for the fields, labelling different copies of $G_{\omega}$. Internal degrees of freedom for $H$ ({\it e.g.}, spin or sublattice labels) can also be taken into account in a similar way.

Eq. (\ref{eq:SUSY}) is an {\it exact formula} for the averaged Green's function of the model on a finite volume. It allows to recast the problem of computing the averaged Green's function for a random Schr\"odinger operator into a statistical mechanics/quantum field theory problem. The factor $i$ allows to circumvent the fact that the operator $H - \mu$ need not be positive. Moreover, the parameter $\varepsilon>0$ allows to avoid singularities in the determinant at the denominator, and to make sense of the complex Gaussian integrals. The problem we now have to face is to construct this interacting quantum field theory model for $\lambda$ small, {\it uniformly} in the volume of the system and as $\varepsilon \to 0^{+}$. 

A formal approach often adopted in the physics literature is to perform a saddle point analysis for the full Gaussian superfield $\Phi^{\pm} = (\phi^{\pm}, \psi^{\pm})$, see \cite{E2} for a review. As a result, one obtains remarkable predictions about the behavior of the systems, such as the emergence of random matrix statistics for the eigenvalue distribution of $H_{\omega}$. Making this strategy rigorous, however, presents very serious mathematical challenges, which so far have been rigorously tackled only for a class of effective supersymmetric models, \cite{CFGK, DSZ, DS}, or in the context of random matrix models (mean field regime) \cite{SS1, SS2}.

Another possibility, less explored from a rigorous viewpoint, is to apply rigorous renormalization group (RG) methods to construct the Gibbs state of the interacting supersymmetric model for small $\lambda$, that is to evaluate the integral in Eq. (\ref{eq:SUSY}) via a convergent multiscale analysis. Similar methods have been recently used in \cite{Ma}, for an analysis at all orders in renormalized perturbation theory of the correlation functions of an effective supersymmetric model of graphene in the presence of random gauge fields. See also \cite{Be, MPR1, MPR2, P} for earlier approaches to disordered systems via a combination of RG and random matrix techniques. For quantum systems with quasi-random disorder, rigorous RG techniques have been used to prove the existence of localization in the ground state of the interacting fermionic chains \cite{Mdis}.

In the present context, the most serious difficulties one has to face in a nonperturbative application of RG methods are:
\begin{itemize}
\item[(i)] the large field problem, due to the unboundedness of the bosonic fields; 
\item[(ii)] the infrared problem, which arises whenever $\mu$ lies in the spectrum of $H$; 
\item[(iii)] the presence of a purely imaginary covariance for the bosonic integration. 
\end{itemize}
The goal of this paper is to present a rigorous solution to these problems, in a simple yet nontrivial case. As usual in condensed matter physics models, the geometry of the Fermi surface determines how severe are the infrared divergences appearing in the naive perturbative expansion. In particular, in the context of interacting fermionic systems, the rigorous study of the ground state of models with general extended Fermi surfaces is so far out of the reach of the existing rigorous RG methods. Important progress has been achieved in \cite{DR1, DR2}, for the low temperature construction of jellium, in \cite{BGM}, for low temperature analysis of the $2d$ Hubbard model on the square lattice, and in $\cite{FKT}$, for the Fermi liquid construction of $2d$ models with asymmetric Fermi surface.

Here we shall consider a class of $3d$ quantum systems with {\it pointlike} Fermi surface; these are models for {\it Weyl semimetals}, see \cite{AMV} for a review. Weyl semimetals are a class of recently discovered condensed matter systems \cite{Hweyl}, that might be thought as a $3d$ generalization of graphene. In these models, the (translation invariant) Schr\"odinger operator $H$ can be written as $H = \int_{\mathbb{T}^{3}}^{\oplus} dk\, \hat H(k)$, with $\hat H(k)$ the {\it Bloch Hamiltonian} of the model, $k$ the quasi-momentum of the particle, and $\mathbb{T}^{3}$ the Brillouin zone. The energy bands of $\hat H(k)$ display conical intersections at the Fermi level $\mu$, at a finite number of Fermi points, also called Weyl nodes, $k=k_{F}^{\alpha}$, $\alpha = 1, \ldots, 2M$. As a consequence of this fact one has, up to oscillating prefactors:
\begin{equation}\label{eq:weyl}
G(x,y; \mu) = \int_{\mathbb{T}^{3}} dk\, e^{ik\cdot (x - y)}\frac{1}{\hat H(k) - \mu} \sim \frac{1}{\| x - y \|^{2}}\qquad \text{as $\|x - y\|\to \infty$.}
\end{equation}
It turns out that the reduced dimensionality of the Fermi surface allows to use RG methods to construct the low/zero temperature interacting Gibbs state of the model, both in two (corresponding to graphene-like systems) and three dimensions, and to prove universality results for transport coefficients; see \cite{GM, GMP, GJMP, GMP2, Ma1, Ma2, GMPweyl}. We refer the reader to \cite{Mabook} for a review of recent applications of rigorous RG methods interacting condensed matter systems.

Here we shall focus on three dimensional disordered Weyl semimetals, in the presence of weak disorder and no interactions. Heuristic perturbative analysis suggests that disorder is irrelevant in the renormalization group sense, see \cite{Lu, Fra1, Fra2} for early renormalization group approaches to disordered semimetals. Nevertheless, the emergence of localized states at the Weyl points for weak disorder has been recently proposed in \cite{Huse}, and then challenged in \cite{Altland}; see also \cite{Chen} and references therein for numerical simulations, showing that the delocalized phase in Weyl semimetals with well separated nodes is is robust against weak disorder. In this paper we shall consider a {\it SUSY hierarchical approximation} for disordered Weyl semimetals: the connection between the hierarchical model and the original lattice model lies in the scaling of the superfield covariance, which will be chosen so to match the decay properties of the massless Green's function of the original model, Eq. (\ref{eq:weyl}). Let us also point out that our model describes an interacting SUSY field associated to one Weyl node: from the point of view of Weyl semimetals, it is relevant for the description of disorders that do not couple different Fermi points.

Hierarchical models played an important role in the development of rigorous RG methods, \cite{Dys, BlSi, GaKn, Gal, GK}. For instance, we mention the study of the hierarchical $\varphi^{4}_{4}$ theory \cite{GK}, which paved the way to the construction of the full lattice $\varphi^{4}_{4}$ theory \cite{GK2}. The connection between the two models is provided by a cluster expansion \cite{GK2}, technically similar to a high temperature expansion in classical statistical mechanics. See also \cite{BBS1, BBS2} for a recent extension of this result to SUSY $\varphi^{4}_{4}$, relevant for the study of the weakly selfavoiding walk, and \cite{BBS3} for a detailed discussion of the hierarchical approximation of the model.

Hierarchical models have also been considered in the context of random Schr\"odinger operators \cite{Bo, Mo, K1, K2, vSW1, vSW2}, see also \cite{MG, Pa} for discussions about the connection with the Anderson localization/delocalization transition. There, the model is defined on a one-dimensional lattice, and the range of the hierarchical hopping is tuned to fix the effective dimension of the system. The works \cite{Mo, K1, K2, vSW1, vSW2} prove that, as long as the hopping is summable, the model is in the localized phase.

In this paper, we rigorously construct the SUSY hierarchical version of $3d$ Weyl semimetals with well-separated Weyl nodes, and we prove algebraic decay of correlations; the decay exponents are the same as those of the non disordered model. Our RG analysis is inspired by the block spin transformation of \cite{GK, GK2}; in particular, the study of the bosonic sector of the theory is performed thanks to the careful control of the growth of the analyticity domain of the effective action as a function of the complex bosonic field, and to the iteration of suitable analyticity (Cauchy) estimates. With respect to \cite{GK}, an important simplification in our case is due to the fact the interaction (hence the disorder) is irrelevant in the renormalization group sense. However, in contrast to \cite{GK}, the Gaussian covariances are purely imaginary, which means that the single step of RG has to be performed exploiting oscillations. Also, one has to deal with the extra presence of fermionic fields, which makes the analysis considerably more involved with respect to a purely bosonic theory. Also, a key role in our construction is played by supersymmetry, that allows to reduce the number of running coupling constants, and to prove the equality (up to a sign) of fermionic and bosonic correlations. If combined with a suitable cluster expansion, we expect our result to extend to the full lattice model; we postpone this study to future work. The only other application of cluster expansion techniques and RG methods to QFT models with complex covariances we are aware of is the work \cite{BFKT}, on the construction of the ultraviolet sector of interacting three dimensional lattice bosonic systems.

The paper is organized as follows. In Section \ref{sec:model} we introduce the model we will study, and we will state our main result, Theorem \ref{thm:main}. In Section \ref{sec:RGpart} we develop the RG method, that we will first apply to the construction of the effective potential of the theory. Then, in Section \ref{sec:2pt} we apply this strategy to the computation of the two-point function of the model, which allows to prove our main result. Finally, in Appendix \ref{app:mu} we discuss the flow of the counterterm fixing the choice of the interacting chemical potential; in Appendix \ref{app:osc} we prove some key technical results; while in Appendix \ref{app:SUSY} we discuss the (super-)symmetries of the model.

\section{The model}\label{sec:model}

\subsection{The hierarchical Gaussian superfield}
Let $N \in \mathbb{N}$, $L\in 2\mathbb{N} $. Let $\Lambda\subset \mathbb{N}^{3}$ be the set:
\begin{equation}
\Lambda := \Big\{ x\in \mathbb{N}^{3} \mid 0 \leq x_{i} < L^{N},\; i=1,2,3 \Big\}\;.
\end{equation}
Let $\Lambda^{(1)} :=L^{-1}\Lambda \cap  \mathbb{N}^{3}$. Later, it will be convenient to look at $\Lambda$ as being covered by disjoint blocks $\b^{(1)}_{z}$ of side $L$ and labelled by $z \in \Lambda^{(1)}$:
\begin{eqnarray}
\b^{(1)}_{z} &:=& \Big\{ x\in \Lambda \mid z_i \leq L^{-1} x_{i} < z_{i} +1,\; i =1,2,3 \Big\} \;. \nonumber\\
\Lambda &=& \bigcup_{z \in \Lambda^{(1)}}\b_{z}^{(1)}\;.
\end{eqnarray}
More generally, for any $1\leq k\leq N$, we set $\Lambda^{(k)} :=L^{-k}\Lambda \cap \mathbb{N}^{3}$. Obviously, $\Lambda^{(k+1)} \subset \Lambda^{(k)}$. We set, for any $z \in \Lambda^{(k)}$, with the understanding $\Lambda^{(0)} \equiv \Lambda$:
\begin{eqnarray}
\b^{(k)}_{z} &:=& \{ x  \in \Lambda^{(k-1)} \mid z_{i} \leq L^{-1} x_{i} < z_{i} +1,\; i =1,2,3 \} \;,
\nonumber\\
\Lambda^{(k-1)} &=& \bigcup_{z \in \Lambda^{(k)}}\b_{z}^{(k)}\;.
\end{eqnarray}
%
%
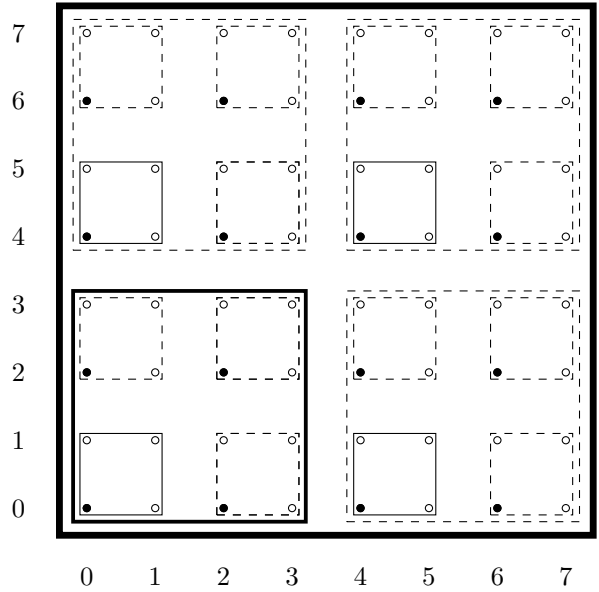
\begin{figure}[ht]\centering
\begin{tikzpicture}[scale=.9]
    \foreach \i in {0,...,7}
      \path[black] (\i+1,0) node{\i} (0,\i+1) node{\i};
    \foreach \i in {1,...,8}
      \foreach \j in {1,...,8}{
        \draw (\i,\j) circle(1.5pt);
            };  
             \foreach \i in {1,3,5,7}
      \foreach \j in {1,3,5,7}{
        \draw[black, fill] (\i,\j) circle(1.5pt);
            };         
       \draw (0.9, 0.9) rectangle (2.1, 2.1);
       \draw[ dashed] (0.9, 2.9) rectangle (2.1, 4.1);
       \draw (0.9, 4.9) rectangle (2.1, 6.1);
       \draw[ dashed] (0.9, 6.9) rectangle (2.1, 8.1);
       \draw[ dashed] (2.9, 0.9) rectangle (4.1, 2.1);
       \draw[ dashed] (2.9, 2.9) rectangle (4.1, 4.1);
       \draw[ dashed] (2.9, 4.9) rectangle (4.1, 6.1);
       \draw[dashed] (2.9, 6.9) rectangle (4.1, 8.1);
        \draw (4.9, 0.9) rectangle (6.1, 2.1);
       \draw[ dashed] (4.9, 2.9) rectangle (6.1, 4.1);
       \draw (4.9, 4.9) rectangle (6.1, 6.1);
       \draw[ dashed] (4.9, 6.9) rectangle (6.1, 8.1);
        \draw[ dashed] (6.9, 0.9) rectangle (8.1, 2.1);
       \draw[ dashed] (6.9, 2.9) rectangle (8.1, 4.1);
       \draw[ dashed] (6.9, 4.9) rectangle (8.1, 6.1);
       \draw[ dashed] (6.9, 6.9) rectangle (8.1, 8.1);
        \draw[ dashed] (2.9, 0.9) rectangle (4.1, 2.1);
       \draw[ dashed] (2.9, 2.9) rectangle (4.1, 4.1);
       \draw[ dashed] (2.9, 4.9) rectangle (4.1, 6.1);
       \draw[line width=0.5mm] (0.8, 0.8) rectangle (4.2, 4.2);
       \draw[ dashed] (0.8, 4.8) rectangle (4.2, 8.2);
     \draw[dashed] (4.8, 0.8) rectangle (8.2, 4.2);
       \draw[dashed] (4.8, 4.8) rectangle (8.2, 8.2);
       \draw[line width=0.9mm] (0.6,0.6) rectangle (8.4, 8.4);
  \end{tikzpicture}
\caption{Two-dimensional example of the hierarchy of blocks, with $L=2$ and $N=3$. The small circles are the lattice sites, elements of $\Lambda \equiv \Lambda^{(0)}$. The smallest squares constitute the blocks $\b^{(1)}$, which, if identified with their labels form the lattice $\Lambda^{(1)}$. The smallest blocks containing $\b^{(1)}$ are the blocks $\b^{(2)}$ and so on. Notice that blocks are labelled by referring to the coordinates of the bottom left element they contain, here highlighted as black circles, or continuous-line blocks.}\label{fig:blocks}
\end{figure}
See Fig. \ref{fig:blocks}. Given $x\in \Lambda$, the box $\mathcal{B}^{(1)}$ containing $x$ is $\mathcal{B}^{(1)}_{\fl{L^{-1}x}}$, where $\fl{a}$ denotes the vector in $\mathbb{N}^{3}$ which approximates $a$ from below. Notice that the box $\mathcal{B}^{(2)}$ containing $\fl{L^{-1} x} \in \Lambda^{(1)}$ is $\mathcal{B}^{(2)}_{\fl{L^{-2} x}}$; and so on. In this way, one defines a {\it hierarchy} of boxes, where a box at a given scale contains the lattice points of the previous scale. We are now ready to introduce the hierarchical Gaussian superfield. Roughly speaking, one associates to any point $x\in \Lambda$ a sum of independent Gaussian variables, each of them corresponding to a box within the hierarchy described above. We will follow the definition of the hierarchical model of \cite{GK}, which captures the main features of the multiscale decomposition of the full lattice Gaussian field \cite{GK2}. We define a complex Gaussian field $\phi_{x,\sigma}$ and a pair of Grassmann Gaussian fields $\psi^{+}_{x,\sigma}$, $\psi^{-}_{x,\sigma}$ as follows, for all $x\in \Lambda$, and for all spins $\sigma = \uparrow\downarrow$:
\begin{equation}\label{eq:deffields}
\phi_{x,\sigma} := \sum_{h = 0}^{N-1} L^{-h} A_{\fl{ L^{-h} x} } \zeta^{(h)}_{\phi, \fl{ L^{-h-1} x } , \sigma}\;,
\quad 
\psi^{\pm}_{\sigma,x} := \sum_{h=0}^{N-1}  L^{-h} A_{ \fl{ L^{-h} x} }\zeta^{(h)\pm}_{\psi, \fl{L^{-h - 1} x}, \sigma}
\end{equation}
for suitable independent complex Gaussian fields $\zeta^{(h)}_{\phi}$ and Grassmann Gaussian fields $\zeta^{(h)\pm}_{\psi}$, whose covariance will be defined in Eq~\eqref{eq:definition_h_covariance}. The function $A$ is the scale-independent local version of the kernels that appear in the block-spin transformation \cite{GK2}. It is chosen so that $A_{y} = \pm 1$, for all $y\in \Lambda^{(0)}$, with $\sum_{y \in \b^{(1)}_{z}} A_{y} = 0$ for all $z\in \Lambda^{(1)}$, and it is invariant under translations by $L$ in each direction.  

Recall the defining properties of Grassmann variables\footnote{As usual, the symbols $\{d\zeta^{(h)\varepsilon}_{\psi, x, \sigma}\}$ form a Grassmann algebra, anticommuting with the algebra generated by $\{\zeta^{(h)\varepsilon}_{\psi, x, \sigma}\}$}:
\begin{eqnarray}\label{eq:grassmann}
&&\zeta^{(h)\varepsilon}_{\psi, x, \sigma} \zeta^{(k)\varepsilon'}_{\psi, y, \sigma'} + \zeta^{(k)\varepsilon'}_{\psi, y, \sigma'} \zeta^{(h)\varepsilon}_{\psi, x, \sigma} = 0\;,\nonumber\\
&&\int d \zeta^{(h)\varepsilon}_{\psi, x, \sigma} = 0\;,\qquad \int d \zeta^{(h)\varepsilon}_{\psi, x, \sigma}\, \zeta^{(h)\varepsilon}_{\psi, x, \sigma} = 1\;,
\end{eqnarray}
for all possible choices of the labels. Notice that, in Eq. (\ref{eq:deffields}), the fields $\zeta^{(h)\pm}_{\psi, \fl{ L^{-h - 1} x}, \sigma}$, $\zeta^{(h)\pm}_{\phi, \fl{ L^{-h - 1} x}, \sigma}$ are labelled by the same points $\fl{L^{-h - 1} x} \in \Lambda^{(h+1)}$ that label the blocks $\b^{(h+1)}_{z}$. 

Setting $\phi^{-} := \phi$ and $\phi^{+} := \overline{\phi}$, we shall collect both complex and Grassmann fields in a single Gaussian {\it superfield} $\Phi_{x,\sigma}^{\pm}$:
\begin{equation}
\Phi^{\pm}_{x,\sigma} := ( \phi^{\pm}_{x,\sigma},\, \psi^{\pm}_{x,\sigma} )\;.
\end{equation}
Setting also $\zeta^{(h)+}_{\phi} \equiv \zeta^{(h)}_{\phi}$, $\zeta^{(h)-}_{\phi} = \overline{\zeta^{(h)}_{\phi}}$, we shall also introduce single scale superfields $\zeta^{(h)}_{x,\sigma}$ as $\zeta^{(h)\pm}_{x,\sigma} = (\zeta^{(h)\pm}_{\phi, x, \sigma},\, \zeta^{(h)\pm}_{\psi, x,\sigma})$. In terms of those, we rewrite the superfield $\Phi^{\pm}$ as:
\begin{equation}\label{eq:gausshier}
\Phi^{\pm}_{x,\sigma} = \sum_{h = 0}^{N-1} L^{-h} A_{\fl{ L^{-h} x} } \zeta^{(h)}_{\fl{ L^{-h-1} x}, \sigma}\;.
\end{equation}
Eq. (\ref{eq:gausshier}) defines the hierarchical Gaussian superfield. Notice that the term in the above sum labelled by a given $h$ varies on the length scale $L^{h}$, and it is of size $L^{-h}$. In particular, replacing the above sum by a truncated sum starting from the scale $k$, one obtains a field that varies on length scale $L^{k}$, and it is of size $L^{-k}$. For later convenience, it is useful to rescale this truncated field, in such a way that it varies on scale $1$, and it is of size $1$, for all $k\geq 0$. We introduce, for all $x\in \Lambda^{(k)}$:
\begin{equation}\label{eq:fieldk}
\Phi^{(\geq k)\pm}_{x,\sigma} := L^{k}\sum_{h = k}^{N-1 } L^{-h} A_{\fl{ L^{-h} L^{k} x} } \zeta^{(h)}_{\fl{ L^{-h-1} L^{k} x}, \sigma}\;.
\end{equation}
Notice that $\Phi^{(\geq k)\pm}_{x,\sigma}$ satisfies the recursion relation:
\begin{equation}\label{eq:decompPHI}
\Phi^{(\geq k)\pm}_{x,\sigma} = L^{-1} \Phi^{(\geq k+1)\pm}_{\fl{x/L},\sigma} + A_{x} \zeta^{(k)}_{\fl{ x/L },\sigma}\;.
\end{equation}
This decomposition has a clear meaning: the field $\Phi^{(\geq k)}_{x}$ is written as the sum of a term which is constant in the block $\b^{(k+1)}_{\fl{x/L}}$, the average of the field in the block, plus a fluctuation with zero sum in the same block.

We shall choose the covariance of the independent single scale superfields $\zeta^{(h)}$ as, for $\sharp = \phi, \psi$:
\begin{eqnarray}\label{eq:definition_h_covariance}
C^{(h)}_{\sharp; \sigma, \sigma'}(x,y) &:=& \langle \zeta^{(h)-}_{\sharp,\sigma, \fl{L^{-h-1}x}} \zeta^{(h)+}_{\sharp, \sigma', {L^{-h-1}y}} \rangle\nonumber\\
&=& -i\delta_{\sigma,\sigma'} \delta_{\fl{L^{-h-1}x}, \fl{L^{-h-1}y}}\;.
\end{eqnarray}
Thus, the covariance of the full superfield is:
\begin{eqnarray}\label{eq:2ptbare}
C_{\sharp;\sigma, \sigma'}^{(\leq N-1)}(x,y) &=& \sum_{h=0}^{N-1} L^{-2h} C^{(h)}_{\sharp; \sigma, \sigma'}(x,y) \nonumber\\
&=&\delta_{\sigma, \sigma'} \frac{-i}{d(x,y)^{2}} \sum_{h=k(x,y)}^{N} \frac{A_{\fl{ L^{-h+1} x}} A_{\fl{ L^{-h+1}y }}}{L^{2(h-k)}}\;,
\end{eqnarray}
where $d(x,y)$ is the hierarchical distance between $x$ and $y$:
\begin{equation}
d(x,y) := L^{k(x,y)}\;,\qquad k(x,y): = \min\{k\in \mathbb{N}: \fl{ x/L^{(k+1)}} = \fl{ y/L^{(k+1)}} \}.
\end{equation}
Notice that this covariance mimics the real space algebraic decay of the Green's function of the full lattice model, Eq. (\ref{eq:weyl}). In particular, the algebraic decay of the covariance implies that the Gaussian superfield is massless.

\subsection{The Gibbs state of the interacting hierarchical model}
The goal of this paper is to study weak perturbations of the massless Gaussian superfield defined in the previous section. We define:
\begin{equation}
V(\Phi) := \lambda \sum_{x\in \Lambda} ( \Phi_{x}^{+}\cdot \Phi_{x}^{-} )^{2} + i \mu \sum_{x\in \Lambda} (\Phi_{x}^{+}\cdot \Phi_{x}^{-})
\end{equation}
where $(\Phi_{x}^{+}\cdot \Phi_{x}^{-}) := \sum_{\sigma} [\phi^{+}_{\sigma, x} \phi^{-}_{\sigma, x} + \psi^{+}_{\sigma, x} \psi^{-}_{\sigma, x}]$. The first term plays the role of the many-body interaction for the superfield, while the second term will fix the chemical potential of the system, and will be suitably chosen later on. Given an analytic function $P(\Phi)$, with $\Phi$ as in Eq. (\ref{eq:gausshier}), we define, for $\lambda > 0$ and $|\mu|\leq C|\lambda|$:
\begin{eqnarray}
\label{eq:Gibbs_state_definition}
\langle P(\Phi) \rangle_{N} &:=& \frac{1}{\mathcal{Z}_{N}}\int \big[\prod_{h=0}^{N-1} d\mu(\zeta^{(h)})\big]\,  e^{-V(\Phi)} P(\Phi)\;, \\
d\mu(\zeta^{(h)})&:=& \prod_{x\in \Lambda^{(h+1)}} d\mu(\zeta^{(h)}_{x})\;,\qquad d\mu(\zeta^{(h)}_{x}) := d\mu_{\phi}(\zeta_{\phi,x}^{(h)}) d\mu_{\psi}(\zeta_{\psi,x}^{(h)})\;,\nonumber\\
d\mu_{\sharp}(\zeta_{\sharp,x}^{(h)}) &:=& -
\Big[ \prod_{\sigma = \uparrow\downarrow} d\zeta^{(h)+}_{\sharp, x,\sigma}d\zeta^{(h)-}_{\sharp, x,\sigma}\Big]\, e^{-i \sum_{\sigma = \uparrow\downarrow} \zeta^{(h)+}_{\sharp, x,\sigma} \zeta^{(h)-}_{\sharp, x,\sigma}}\;. \nonumber
\end{eqnarray}
Fermionic integration, $\sharp = \psi$, is defined in Eq.~(\ref{eq:grassmann}), while for bosonic integration, $\sharp = \phi$, we use the convention $d\zeta^{(h)+}_{\phi,x,\sigma} d\zeta^{(h)-}_{\phi,x,\sigma} := \pi^{-1} d \text{Re}\, \zeta^{(h)}_{\phi,x,\sigma} d\text{Im}\, \zeta^{(h)}_{\phi,x,\sigma}$. The minus sign in front of the RHS of last equation in \eqref{eq:Gibbs_state_definition} is for normalisation purposes, indeed $\lim _{\epsilon \rightarrow 0^{+}}\int d\mu_{\phi}(\zeta_{\phi,x}^{(h)}) \, e^{-\epsilon \sum_{\sigma = \uparrow \downarrow} \zeta^{(h)+}_{\phi, x,\sigma} \zeta^{(h)-}_{\phi, x,\sigma}}  =1$, $\int d\mu_{\psi}(\zeta_{\psi, x}^{(h)}) =1$. It is important to notice that although $ \prod_{h = 0}^{N-1} \prod_{x \in \Lambda^{(h)}}  d \mu _{\phi}(\zeta^{(h)}_{\phi,x})$ does not provide any decay, the integral is well-defined for all $\lambda >0$ by the presence of the quartic interaction.

The normalization factor $\mathcal{Z}_{N}$ is the partition function of the model,
\begin{equation}\label{eq:ZN}
\mathcal{Z}_{N} = \int \big[\prod_{h=0}^{N-1} d\mu(\zeta^{(h)})\big]\,  e^{-V(\Phi)}\;;
\end{equation}
as discussed in Appendix \ref{app:SUSY}, $\mathcal{Z}_{N} = 1$ by the localization theorem of supersymmetric integration, \cite{BT, SZ}.
\begin{remark} In the following, we shall use the symbols $C, \widetilde{C}, K, \tilde{K}$ for generic universal constants. Whenever a constant will depend on $L$, we shall denote it explicitly, {\it i.e.} with the symbol $C_{L}$.
\end{remark}
\subsection{Main result}
In this paper we shall construct the Gibbs state of the hierarchical SUSY model. In particular, we shall focus on the two-point correlation function; the same methods can be extended in a straightforward way to all higher order correlations. The next theorem is our main result.
\begin{theorem}\label{thm:main} There exists $L_{0} \in 2\mathbb{N}$ such that for all $L\in 2\mathbb{N}$, $L\geq L_{0}$, and for all $0< \theta < 1/2$, $N\in \mathbb{N}$ the following is true. There exists $\bar \lambda$ independent of $N$ such that for all $0< \lambda < \bar \lambda$ there exists a unique function $\mu \equiv \mu(\lambda) \in \mathbb{C}$, $|\mu(\lambda)| \leq C|\lambda|$ and a constant $K > 0$ such that, for all $x,y \in \Lambda$:
\begin{eqnarray}
\langle \phi^{+}_{\sigma, x} \phi^{-}_{\sigma', y} \rangle_{N} &=& -\langle \psi^{+}_{\sigma, x}\psi^{-}_{\sigma', y} \rangle_{N} \\
&=&  \frac{-i \delta_{\sigma, \sigma'}}{d(x,y)^{2}} 
\left ( \sum_{h=k(x,y)}^{N} \frac{A_{\fl{ L^{-h+1} x}} A_{\fl{ L^{-h+1}y }}}{L^{2(h-k)}} + \mathcal{E}_{N}(x,y)\right )\,,\nonumber
\end{eqnarray}
with: 
\begin{equation}
| \mathcal{E}_{N}(x,y) | \leq K \frac{\lambda^{\theta}}{d(x,y)^{\theta}}\;.
\end{equation}
\end{theorem}
\begin{remark} The constants $\bar \lambda$ and $K$ only depend on $L$ and $\theta$.
\end{remark}
The proof of Theorem \ref{thm:main} is based on rigorous renormalization group methods. The parameter $\mu$ plays the role of bare chemical potential; technically, its choice allows to control the relevant direction of the RG flow, and to prove the convergence to a Gaussian fixed point. In the original lattice model, the choice of $\mu$ would correspond to a shift of the Fermi level associated to the Weyl phase, induced by the disorder. We stated the theorem for $\mu \in \mathbb{C}$; we believe that with some extra effort one could actually prove that $\mu \in \mathbb{R}$, but we will not need this improvement in our analysis. Notice that for a general disordered lattice model $\mu$ has to be real, since otherwise the Green's function would decay exponentially fast by a Combes-Thomas estimate.

As mentioned in the introduction, the main difficulties in the RG analysis are due to the massless covariance of the superfield, to the unboundedness of the bosonic field (large field problem) and to the fact that the covariance is purely imaginary. In order to perform the single-scale integration in the RG, one has to exploit {\it oscillations} in the Gaussian integration.
\medskip

\noindent{\underline{\it Relation with the full lattice model.}} Before discussing the proof, let us briefly comment on the relevance of this result for the understanding of the behavior of the full lattice model, beyond the hierarchical approximation. As discussed in the introduction, the supersymmetric representation can be used to study the averaged resolvent of the lattice model, with Gaussian disorder:
\begin{equation}
\mathbb{E}_{\omega} \frac{1}{-i (H_{\omega} - \mu) + \eps}(x,y)\;,\qquad \eps > 0\;.
\end{equation}
Notice that we will not be concerned with the average of the {\it absolute value} of the resolvent, which is expected to diverge as $\eps \to 0$. Our analysis focuses on a hierarchical approximation of the SUSY field-theoretical representation of the averaged Green's function. More precisely, our hierarchical field is related to the quasi-particle field associated to {\it one} conical intersection; from the point of view of the original lattice model, this would amount to introducing a momentum cutoff in the covariance of the disorder, which has the effect of keeping the Fermi points decoupled.

The resolvent can be used to compute the Fermi projector in a finite volume, via functional calculus \cite{AG, AW}:
\begin{equation}\label{eq: disorder_average_fermi_projector}
P_{\omega}(H_{\omega} \leq \mu) = \oint_{\mathcal{C}_{\mu}} \frac{dz}{2\pi i}\, \frac{1}{z - H_{\omega} } =  Q_{1}(H_{\omega} - \mu) + Q_{2}(H_{\omega} - \mu)
\end{equation}
where $\mathcal{C}_{\mu}$ is the counter-clockwise path in the complex plane, $(-\infty - i) \to (\mu - i) \to (\mu + i) \to (-\infty + i)$, and where:
\begin{equation}
\begin{split}
Q_{1}(H_{\omega} - \mu)  & := \frac{1}{2 \pi} \int _{-1}^{1}d \eta\, \frac{1}{i\eta + \mu - H_{\omega} }\;,
\qquad
\\
Q_{2}(H_{\omega} - \mu)  & := \frac{1}{2 \pi i} \int _{-\infty}^{0} d u \, \Big [\frac{1}{u -i+ \mu - H_{\omega}} - \frac{1}{u +i + \mu - H_{\omega}} \Big]\;.
\end{split}
\end{equation}
Here we are assuming that the point $\mu$ does not belong to the spectrum of $H_{\omega}$ (which consists of eigenvalues, since we are in a finite volume); as discussed below, this is true with probability one. The $u$ integration in $Q_{2}$ converges absolutely, uniformly in all disorder realizations. Moreover, a Combes-Thomas estimate \cite{AG} gives $|Q_{2}(H_{\omega} - \mu)(x,y)|\leq Ce^{-c\|x-y\|}$ uniformly in the disorder and in the systems size. Consider now the $Q_{1}$ term. Let $\gamma$ be the disorder strength, recall Eq. (\ref{eq:anderson}). The following bound holds true, for a broad class of disordered lattice models in a finite volume\cite{AM} (including the finite volume version of (\ref{eq:anderson})):
\begin{equation}
\mathbb{P}_{\omega} \Big(\Big| \frac{1}{H_{\omega} - z}(x;y) \Big| \geq t\Big) \leq \frac{C}{|\gamma| t}\;,\qquad \forall t > 0\;.
\end{equation}
This estimate implies that, uniformly in $x, y$ and in the system size, and for some $C_{\gamma,\theta}<\infty$, see (2.12) of \cite{AM}:
\begin{equation}\label{eq:frac}
\mathbb{E}_{\omega} \Big| \frac{1}{H_{\omega} - z}(x;y) \Big|^{\theta} \leq C_{\gamma,\theta}\;,\qquad \forall\, 0\leq \theta < 1\;.
\end{equation}
In particular, the probability that any point $\mu \in \mathbb{R}$ belongs to the spectrum of $H_{\omega}$ is zero; hence, Eq. (\ref{eq: disorder_average_fermi_projector}) holds true for almost all disorder configurations. Moreover, the bound (\ref{eq:frac}) together with the estimate $| (H_{\omega} - z)^{-1}(x;y) |^{1-\theta}\leq |\text{Im}\, z|^{\theta - 1}$, can be used to prove that the quantity $\mathbb{E}_{\omega} (H_{\omega} - \mu - i\eta)^{-1}(x;y)$ is absolutely integrable in $\eta \in [-1;1]$. Therefore, we can interchange the $\eta$ integration with the average over the disorder:
\begin{equation}
\mathbb{E}_{\omega} Q_{1}(H_{\omega} - \mu) = \frac{1}{2\pi} \int_{-1}^{1} d\eta\, \mathbb{E}_{\omega} \frac{1}{i\eta + \mu - H_{\omega}}(x;y)\;.
\end{equation}
Thus, one can deduce decay properties of the averaged Fermi projector starting from the decay properties of the averaged Green's function. This is exactly what we study in the present paper, in the hierarchical approximation, at the most singular point $\eta = 0$. We prove that the averaged Green's function at $\eta = 0$ decays as $\|x - y\|^{-2}$ at large distances (as in the disorder-free case). 
%
%
More generally, for $\eta \neq 0$ a straightforward extension of the analysis performed in the present paper would give a decay at most as $\|x - y\|^{-2} e^{- |\eta| \| x - y \|}$ for the hierarchical approximation of the resolvent. Therefore, based on our result for the hierarchical approximation, the natural conjecture for the full lattice model is that:
\begin{equation}\label{eq:Pdec}
\mathbb{E}_{\omega} P(H_{\omega} \leq \mu)(x,y) \sim \frac{1}{\| x - y \|^{3}}\qquad \text{as $\| x - y \|\to \infty$.}
\end{equation} 
Notice that this is {\it not} what happens in the regime of strong disorder, where one has the fractional moment bound $\mathbb{E}_{\omega} | (H_{\omega} - z)^{-1} (x,y)|^{\theta}\leq Ce^{-c\|x - y\|}$ for $0<\theta < 1$, an estimate which can be used to prove Anderson localization (see \cite{AW} for a review). In particular, the fractional moment bound implies that \cite{AG}:
\begin{equation}
\mathbb{E}_{\omega} |P(H_{\omega} \leq \mu)(x,y)| \leq Ce^{-c\|x - y\|}\;.
\end{equation}
As discussed in the introduction, in order to extend our main result, Theorem \ref{thm:main}, to the full lattice model, one has to combine the RG analysis introduced in this paper with a cluster expansion; we defer this nontrivial extension to future work.
\section{Renormalization group analysis: the effective potential flow}\label{sec:RGpart}

In this section we discuss the renormalization group analysis of our model. We shall start by computing the effective potential of the theory, at all scales. The main result is the expression (\ref{eq:U_h}), for the effective potential on an arbitrary scale. This will play an important role in the computation of the two-point correlation function, postponed to Section \ref{sec:2pt}.

\subsection{The effective potential}

The Gibbs state of the model will be constructed in an iterative fashion, integrating the fluctuation fields starting from the scale $h=0$ until the last scale, $h=N$. A key role in this iterative strategy will be played by the study of the following map (recall Eq. (\ref{eq:decompPHI})):
\begin{eqnarray}\label{eq:TRG}
\mathcal{U}^{(h+1)}(\Phi^{(\geq h+1)}) &:=& T_{\text{RG}} \mathcal{U}^{(h)}(\Phi^{(\geq h)})\nonumber\\ 
T_{\text{RG}} \mathcal{U}^{(h)}(\Phi^{(\geq h)}) &:=& \frac{1}{\mathcal{N}^{(h)}}\int
 d\mu(\zeta^{(h)})\, \mathcal{U}^{(h)}(L^{-1} \Phi^{(\geq h+1)} + A \zeta^{(h)})\nonumber\\
\mathcal{U}^{(0)}(\Phi^{(\geq 0)}) &:=& e^{-V(\Phi^{(\geq 0)})}\;,
\end{eqnarray}
where $(L^{-1} \Phi^{(\geq h+1)} + A \zeta^{(h)})_{x,\sigma} = L^{-1} \Phi^{(\geq h+1)}_{\fl{x/L},\sigma} + A_{ x } \zeta^{(h)}_{\fl{ x/L }, \sigma}$ and
where the normalization factor $\mathcal{N}^{(h)}$ is given by:
\begin{equation}
\mathcal{N}^{(h)} := \int d\mu(\zeta^{(h)})\, \mathcal{U}^{(h)}(A \zeta^{(h)})\;.
\end{equation}
The main simplification introduced in the hierarchical model with respect to the original lattice model is that, on any scale $h$, the argument of the integral factorizes: $\mathcal{U}^{(h)}(\Phi^{(\geq h)}) = \prod_{x\in \Lambda^{(h)}} U^{(h)}(\Phi^{(\geq h)}_{x})$. Therefore, for all $x\in \Lambda^{(h)}$, one has the following local version of the map defined in Eq. (\ref{eq:TRG}):
\begin{eqnarray}\label{eq:trg}
U^{(h+1)}(\Phi^{(\geq h+1)}_{\lfloor x/ L\rfloor }) &=& t_{\text{RG}} U^{(h)}(\Phi^{(\geq h)}_{x}) \nonumber\\
t_{\text{RG}} U^{(h)}(\Phi^{(\geq h)}_{x}) &:=& \frac{1}{N^{(h)}} \int d\mu(\zeta^{(h)}_{\lfloor x/L\rfloor })\, \prod_{y \in \b^{(h+1)}_{\lfloor x/L \rfloor}} U^{(h)}(L^{-1} \Phi^{(\geq h+1)}_{\lfloor x/L\rfloor } + A_{y} \zeta^{(h)}_{\lfloor x/L \rfloor}) \nonumber\\
U^{(0)}(\Phi_{x}^{(\geq 0)}) &:=& e^{-\lambda (\Phi_{x}^{(\geq 0)+}\cdot \Phi^{(\geq 0)-}_{x})^{2} - i\mu (\Phi^{(\geq 0)+}_{x}\cdot \Phi^{(\geq 0)-}_{x})}\;,
\end{eqnarray}
with normalization factor:
\begin{equation}
N^{(h)} := \int d\mu(\zeta^{(h)}_{\lfloor x/L\rfloor})\, \prod_{y  \in \b^{(h+1)}_{\lfloor x/L\rfloor}} U^{(h)}(A_{y} \zeta^{(h)}_{\lfloor x/L\rfloor})\;.
\end{equation}
The assumption $A_y=\pm 1$, $\sum_{y \in \mathcal{B}^{(1)}_{x}} A_{y} = 0$ implies:
\begin{eqnarray}
\label{eq:local_map_U}
&&U^{(h+1)}(\Phi^{(\geq h+1)}_{\lfloor x/L\rfloor}) =\\
&&\frac{1}{N^{(h)}}\int d\mu(\zeta^{(h)}_{\lfloor x/L \rfloor})\, \big[ U^{(h)}(\Phi^{(\geq h+1)}_{\lfloor x/L\rfloor}/L + \zeta^{(h)}_{\lfloor x/L\rfloor}) U^{(h)}(\Phi^{(\geq h+1)}_{\lfloor x/L\rfloor}/L - \zeta^{(h)}_{\lfloor x/L\rfloor })  \big]^{\frac{L^{3}}{2}}\nonumber\\
&&N^{(h)} = \int d\mu(\zeta^{(h)}_{\lfloor x/L\rfloor})\, \big[ U^{(h)}(\zeta^{(h)}_{\lfloor x/L\rfloor}) U^{(h)}(- \zeta^{(h)}_{\lfloor x/L\rfloor})  \big]^{\frac{L^{3}}{2}}\;.\nonumber
\end{eqnarray}
The main technical goal of this paper is the control of this map. We shall prove that, as $N\to \infty$, the iteration (\ref{eq:local_map_U}) converges to a unique Gaussian fixed point, for a suitable choice of the bare chemical potential $\mu$. In the following, we shall drop the dependence of the superfields on the scale and position labels, since they will play no role in the study of the single step of the iteration.
\subsection{Integration of the zero scale}
In this section we will discuss the first iteration of the map defined in Eq. (\ref{eq:trg}). The iteration of the map on later scales will be performed inductively; the correct inductive assumptions will be motivated by the discussion of this section.
\subsubsection{Setting up the integration}

Given two superfields $\Phi^{\pm} = (\phi^{\pm}, \psi^{\pm})$ and $\zeta^{\pm} = (\zeta^{\pm}_{\phi}, \zeta^{\pm}_{\psi})$, we define:
\begin{eqnarray}\label{eq:dot}
&&\quad\qquad\qquad\qquad\qquad\qquad\Phi \cdot \zeta := \phi \cdot \zeta_{\phi} + \psi \cdot \zeta_{\psi}\;,\\
&&\phi\cdot \zeta_{\phi} := \frac{1}{2} ( \sum_{\sigma = \uparrow\downarrow} \phi^{+}_{\sigma} \zeta_{\phi, \sigma}^{-} + \zeta_{\phi, \sigma}^{+} \phi^{-}_{\sigma}  )\;,\quad \psi\cdot \zeta_{\psi} := \frac{1}{2}( \sum_{\sigma = \uparrow\downarrow} \psi^{+}_{\sigma} \zeta_{\psi, \sigma}^{-} + \zeta_{\psi, \sigma}^{+} \psi^{-}_{\sigma} )\;.\nonumber
\end{eqnarray}
Setting $\psi = (\psi^{+}_{\uparrow}, \psi^{-}_{\uparrow}, \psi^{+}_{\downarrow}, \psi^{-}_{\downarrow})^{T}$, $\zeta_{\psi} = (\zeta^{+}_{\uparrow}, \zeta^{-}_{\uparrow}, \zeta^{+}_{\downarrow}, \zeta^{-}_{\downarrow})^{T}$, the fermionic product can also be represented as:
\begin{equation}
\psi\cdot \zeta_{\psi} = \frac{1}{2} \psi^{T}( i\sigma_{2} \otimes \mathbbm{1}_{2} ) \zeta_{\psi}\;, 
\end{equation}
where $i\sigma_{2} = \begin{pmatrix} 0 & 1 \\ - 1 & 0 \end{pmatrix}$ acts on the components of the Grassmann vector with a given spin. Also, it is easy to check that $\psi \cdot \psi = \sum_{\sigma}\psi_{\sigma}^{+}\psi_{\sigma}^{-}$. 

Concerning the bosonic product $\phi\cdot \zeta_{\phi}$ in (\ref{eq:dot}), notice that, setting $\phi^{\pm}_{\sigma} = \phi_{1,\sigma} \pm i \phi_{2,\sigma}$, with $\phi_{i,\sigma} \in \mathbb{R}$:
\begin{equation}\label{eq:prod}
\phi \cdot \zeta_{\phi} = \sum_{\sigma = \uparrow\downarrow} [ \phi_{1,\sigma} \zeta_{\phi, 1, \sigma} + \phi_{2,\sigma} \zeta_{\phi, 2, \sigma}]\;.
\end{equation}
Hence, (\ref{eq:prod}) coincides with the usual scalar product of the following vectors in $\mathbb{R}^{4}$: 
\begin{equation}\label{eq:R4}
\phi := (\phi_{1,\uparrow},  \phi_{2,\uparrow}, \phi_{1,\downarrow}, \phi_{2,\downarrow})^{T}\;,\qquad \zeta_{\phi} := ( \zeta_{\phi, 1,\uparrow},  \zeta_{\phi, 2,\uparrow}, \zeta_{\phi, 1,\downarrow}, \zeta_{\phi, 2, \downarrow} )^{T}\;.
\end{equation}
In particular, $(\phi\cdot \phi) = \|\phi\|^{2} = \sum_{i,\sigma} \phi_{i,\sigma}^{2} = \sum_{\sigma} \phi^{+}_{\sigma} \phi^{-}_{\sigma}$. 

Later, we will be interested in considering the extension of Eq. (\ref{eq:prod}) for {\it complex} $\phi_{1,\sigma}$ and $\phi_{2,\sigma}$. In this general case, Eq. (\ref{eq:prod}) does not define a scalar product on $\mathbb{C}^{4}$, and $\phi\cdot \phi\neq \|\phi\|^{2} = \sum_{i, \sigma} |\phi_{i,\sigma}|^{2}$: what is missing to define a scalar product is the complex conjugate on the first factor. Nevertheless, we still have $|(\phi\cdot \zeta_{\phi})|\leq \|\phi\|\|\zeta_{\phi}\|$. The usual scalar product in $\mathbb{C}^{n}$ will be denoted by $\langle \cdot, \cdot \rangle$, $\langle v, w \rangle = \sum_{j} \overline{v}_{j}w_{j}$.

Setting 
\begin{equation}
U^{(0)}(\Phi) = e^{-\lambda (\Phi\cdot \Phi)^{2} - i\mu (\Phi \cdot \Phi)}\;,
\end{equation}
we shall discuss the evaluation of:
\begin{equation}\label{eq:U1}
U^{(1)}(\Phi) = \frac{1}{N^{(0)}} \int d\mu(\zeta)\, \big[ U^{(0)}(\Phi/L + \zeta) U^{(0)}(\Phi/L - \zeta)  \big]^{\frac{L^{3}}{2}}\;.
\end{equation}
Eq. (\ref{eq:U1}) defines the effective interaction of the hierarchical model on scale $1$. In order to perform the integration, we explicitly rewrite:
\begin{eqnarray}
&&U^{(0)}(\Phi/L + \zeta) U^{(0)}(\Phi/L - \zeta) \\
&&\qquad =  e^{-\frac{2\lambda}{L^{4}} (\Phi \cdot \Phi)^{2} - 2\lambda (\zeta\cdot \zeta)^{2} - \frac{8\lambda}{L^{2}} ( \Phi\cdot \zeta )^{2} - \frac{4\lambda}{L^{2}} (\Phi \cdot \Phi) (\zeta\cdot \zeta)  - \frac{2i\mu}{L^{2}} (\Phi\cdot \Phi) - 2i\mu (\zeta\cdot \zeta) }\;.\nonumber
\end{eqnarray}
Therefore,
\begin{eqnarray}\label{eq:Uzeta}
&&\int d\mu(\zeta)\, [U^{(0)}(\Phi/L + \zeta) U^{(0)}(\Phi/L - \zeta)]^{\frac{L^{3}}{2}} = e^{-\frac{\lambda}{L} (\Phi \cdot \Phi)^{2} - i\mu L (\Phi \cdot \Phi)} \nonumber\\
&&\qquad \cdot \int d\mu(\zeta)\, e^{-4 \lambda L (\Phi\cdot \zeta)^{2} - 2\lambda L (\Phi\cdot \Phi)(\zeta\cdot \zeta) - \lambda L^{3} (\zeta\cdot \zeta)^{2} - i\mu L^{3}(\zeta\cdot \zeta)}\nonumber\\
&&\qquad \equiv e^{-\frac{\lambda}{L} (\Phi \cdot \Phi)^{2} - i\mu L (\Phi \cdot \Phi)}  \int d\mu(\zeta)\, e^{-V^{(0)}(\Phi, \zeta)}\;,
\end{eqnarray}
where we introduced:
\begin{equation}
V^{(0)}(\Phi, \zeta) := 4 \lambda L (\Phi\cdot \zeta)^{2} + 2\lambda L (\Phi \cdot \Phi)(\zeta\cdot \zeta) + \lambda L^{3} (\zeta\cdot \zeta)^{2} + i\mu L^{3} (\zeta\cdot \zeta)\;.
\end{equation}
Our goal will be to discuss the integration of the fluctuation superfield $\zeta$. We shall start by integrating its fermionic component.
\subsubsection{Integration of the fermionic fluctuation field}

We rewrite the effective interaction on the zero scale as:
\begin{eqnarray}
V^{(0)}(\Phi, \zeta) &=& V_{\text{b}}^{(0)}(\Phi, \zeta_{\phi}) + V^{(0)}_{\text{f}}(\Phi, \zeta) \\
V^{(0)}_{\text{b}}(\Phi, \zeta_{\phi}) &:=& V^{(0)}(\Phi, \zeta)|_{\zeta_{\psi} = 0}\;.\nonumber
\end{eqnarray}
Explicitly:
\begin{eqnarray}
V^{(0)}_{\text{b}}(\Phi, \zeta_{\phi}) &=& 4 \lambda L (\phi\cdot \zeta_{\phi})^{2} + 2\lambda L (\Phi \cdot \Phi)(\zeta_{\phi}\cdot \zeta_{\phi}) + \lambda L^{3} (\zeta_{\phi}\cdot \zeta_{\phi})^{2} \nonumber\\&& + i\mu L^{3} (\zeta_{\phi}\cdot \zeta_{\phi})\nonumber\\
V_{\text{f}}^{(0)}(\Phi, \zeta) &=&  \lambda L^{3} (\zeta_{\psi}\cdot \zeta_{\psi})^{2} + 2\lambda L^{3} (\zeta_{\psi}\cdot \zeta_{\psi}) (\zeta_{\phi}\cdot \zeta_{\phi}) + 2\lambda L(\Phi\cdot \Phi)(\zeta_{\psi}\cdot \zeta_{\psi})\nonumber\\&& + 4 \lambda L (\psi\cdot \zeta_{\psi})^{2} + 8 \lambda L (\phi\cdot \zeta_{\phi}) (\psi\cdot \zeta_{\psi}) + i\mu L^{3} (\zeta_{\psi}\cdot \zeta_{\psi})\;.\nonumber
\end{eqnarray}
Therefore, we rewrite the integral in Eq. (\ref{eq:Uzeta}) as:
\begin{equation}\label{eq:fermstep}
\int d\mu(\zeta)\, e^{-V^{(0)}(\Phi, \zeta)} = \int d\mu_{\phi}(\zeta_{\phi}) e^{-V_{\text{b}}^{(0)}(\Phi, \zeta_{\phi})} \int d\mu_{\psi}(\zeta_{\psi})\, e^{-V_{\text{f}}^{(0)}(\Phi, \zeta)}\;.
\end{equation}
The integration of the Grassmann field $\zeta_{\psi}$ will be performed by writing the function $\exp\{ - V_{\text{f}}^{(0)}(\Phi, \zeta)  \}$ as a linear combination of finitely many monomials in the $\zeta_{\psi}$ variable, due to the fact that the Grassmann algebra is finite. Then, we evaluate the Gaussian integration using the rules of Grassmann calculus, Eqs. (\ref{eq:grassmann}). We get:
\begin{eqnarray}\label{eq:R00}
\int d\mu_{\psi}(\zeta_{\psi})\, e^{-V_{\text{f}}^{(0)}(\Phi, \zeta)} = \sum_{n=0,1,2} C^{(0)}_{n}(\phi/L, \zeta_{\phi}) (\psi\cdot \psi)^{n} L^{-2n}\;,
\end{eqnarray}
where the functions $C^{(0)}_{n}(\phi/L, \zeta_{\phi})$ are polynomials in $\phi/L, \zeta_{\phi}$, which we shall estimate. The $(\psi\cdot \psi)$-dependence of the outcome of the integration follows by symmetry, see Corollary \ref{cor:symgra}. We shall proceed as follows. A general function of the supervectors $\Phi$ and $\zeta$ can be written as 
\begin{equation}
f(\Phi,\zeta) = \sum_{\underline{a},\underline{b}} f_{\underline{a},\underline{b}}(\phi,\zeta_{\phi}) \, \psi^{\underline{a}} \, \zeta_{\psi}^{\underline{b}} \;,
\end{equation}
where the summation is over multi-indices $\underline{a}, \underline{b} \in \{0,1 \}^{\{\pm \} \times \{\uparrow,\downarrow \}}$ and $\psi^{\underline{a}}$, $\zeta_{\psi}^{\underline{b}}$ are the corresponding monomial (uniquely defined once an order for the set $\{\pm \} \times \{\uparrow,\downarrow \}$ is chosen).
For such multi-indices $\underline{a} = (a_{\sigma}^{\eps})^{\eps = \pm}_{\sigma = \uparrow\downarrow}$, we also set $|\underline{a}| := \sum _{\varepsilon,\sigma}  a^{\varepsilon}_{\sigma}$. In our setting, the coefficients of the expansion $f_{\underline{a},\underline{b}}$ are functions of $\phi$ and $\zeta_{\phi}$.
\begin{definition}
We say that $f(\psi,\zeta_{\psi}) = \sum_{\underline{a},\underline{b}} f_{\underline{a},\underline{b}}\, \psi^{\underline{a}}\, \zeta_{\psi}^{\underline{b}}$ satisfies $(\kappa,\mathcal{N},\mathcal{M})$-bounds if 
\begin{equation}
\left | f_{\underline{a},\underline{b}} \right | \leq \kappa \, \mathcal{N}^{|\underline{a}|} \mathcal{M}^{|\underline{b}|}
\end{equation}
for some $\kappa,\mathcal{N},\mathcal{M} \geq 0$ (the convention $0^{0} \equiv 1$ is understood). If the function only depends on $\psi$, $f(\psi) = \sum_{\underline{a}} f_{\underline{a}} \, \psi^{\underline{a}}$, we shall say that it satisfies $(\kappa, \mathcal{N})$-bounds if $\left | f_{\underline{a}} \right | \leq \kappa \, \mathcal{N}^{|\underline{a}|}$.
\end{definition}
\begin{remark}\label{rmk:fermionica_analytic_bounds}
The interesting part of the definition is the power law (i.e. analytic) structure of the bounds, which exhibits useful properties under algebraic manipulation and fermion integration, see below.
Notice that many choices of $\kappa, \mathcal{N}, \mathcal{M}$ can of course be made. In particular, since the number of coefficients to bound is finite, for any choice of $ \mathcal{N}, \mathcal{M}$, a suitable  $\kappa$ can always be found.
\end{remark}
Let us take the function $f(\Phi, \zeta)$ to be the exponential $e^{- V_{\text{f}}^{(0)}(\Phi, \zeta)}$. We would like to prove $(\kappa, \mathcal{N}, \mathcal{M})$-bounds for this function. We write:
\begin{equation}\label{eq:sumexp0}
\exp\{- V_{\text{f}}^{(0)}(\Phi, \zeta)\} = \exp \Big\{ -\sum_{j=1}^{7} V_{\text{f},j}^{(0)}(\Phi, \zeta)  \Big\}
\end{equation}
where:
\begin{eqnarray}
&&V_{\text{f},1}^{(0)}(\Phi, \zeta) = \lambda L^{3} (\zeta_{\psi}\cdot \zeta_{\psi})^{2}\;,\quad V_{\text{f},2}^{(0)}(\Phi, \zeta) = 2\lambda L^{3} (\zeta_{\psi} \cdot \zeta_{\psi}) (\zeta_{\phi}\cdot \zeta_{\phi}) \nonumber\\
&& V_{\text{f},3}^{(0)}(\Phi, \zeta) = 2\lambda L (\phi \cdot \phi) (\zeta_{\psi}\cdot \zeta_{\psi})\;,\quad V_{\text{f},4}^{(0)}(\Phi, \zeta) = 2\lambda L (\psi\cdot \psi) (\zeta_{\psi} \cdot \zeta_{\psi})\nonumber\\
&& V_{\text{f},5}^{(0)}(\Phi, \zeta) = 4\lambda L (\psi \cdot \zeta_{\psi})^{2}\;,\quad V_{\text{f},6}^{(0)}(\Phi, \zeta) = 8\lambda L (\phi \cdot \zeta_{\phi})(\psi \cdot \zeta_{\psi})\nonumber\\
&& V_{\text{f},7}^{(0)}(\Phi, \zeta) = i\mu L^{3}(\zeta_{\psi} \cdot \zeta_{\psi})\;.
\end{eqnarray}
Being each term $V_{\text{f},j}^{(0)}(\Phi, \zeta)$ given by a sum of products of an even number of Grassmann variables, we can write:
\begin{equation}\label{eq:V0prod}
\exp\{- V_{\text{f}}^{(0)}(\Phi, \zeta)\} = \prod_{j=1}^{7}\exp \Big\{ - V_{\text{f},j}^{(0)}(\Phi, \zeta)  \Big\}\;,
\end{equation}
where every term in the product can be expanded in a finite sum. To begin, we shall derive $(\kappa, \mathcal{N}, \mathcal{M})$-bounds for every contribution to the product. Let us denote by $(\kappa_{j}, \mathcal{N}_{j}, \mathcal{M}_{j})$ the parameters of the $(\kappa, \mathcal{N}, \mathcal{M})$-bound for the $j$-th term in the product. A simple computation gives:
\begin{eqnarray}\label{eq:kappaNi}
&&(\kappa_{1}, \mathcal{N}_{1}, \mathcal{M}_{1}) = (1, 0, \lambda^{\frac{1}{4}} L^{\frac{3}{4}})\;,\; (\kappa_{2}, \mathcal{N}_{2}, \mathcal{M}_{2}) = (1, 0, (2\lambda L^{3})^{\frac{1}{2}} \| \zeta_{\phi} \|)\;,\nonumber\\
&&(\kappa_{3}, \mathcal{N}_{3}, \mathcal{M}_{3}) = (1, 0, (2\lambda L)^{\frac{1}{2}} \|\phi\|)\;,\; (\kappa_{4}, \mathcal{N}_{4}, \mathcal{M}_{4}) = \Big(1, \frac{\lambda^{\frac{1}{4}}}{L}, 2^{\frac{1}{2}} \lambda^{\frac{1}{4}} L^{\frac{3}{2}}\Big)\nonumber\\
&& (\kappa_{5}, \mathcal{N}_{5}, \mathcal{M}_{5}) = \Big(1, \frac{\lambda^{\frac{1}{4}}}{L}, 2 \lambda^{\frac{1}{4}} L^{\frac{3}{2}}\Big)\;,\; (\kappa_{6}, \mathcal{N}_{6}, \mathcal{M}_{6}) = \Big(1, \frac{\lambda^{\frac{1}{4}}}{L}, 8 \lambda^{\frac{3}{4}} L^{2} \|\phi\| \|\zeta_{\phi}\|\Big)\nonumber\\
&& (\kappa_{7}, \mathcal{N}_{7}, \mathcal{M}_{7}) = (1, 0, C\lambda^{\frac{1}{2}} L^{\frac{3}{2}})\;.
\end{eqnarray}
As already pointed out in Remark \ref{rmk:fermionica_analytic_bounds}, we stress that these choices are not unique. However, as it will be clear in the following, the choice $\mathcal{N} = \lambda^{1/4}/L$ is consistent with the size of the ``small field region'' for the bosonic field, to be introduced later. Also, the value $\kappa = 1$ is natural, because it is equal to the value of all the entries in the product (\ref{eq:V0prod}) for vanishing fermionic fields.

Next, we would like to get a $(\kappa, \mathcal{N}, \mathcal{M})$-bound for the product in Eq. (\ref{eq:V0prod}). To this end, we shall use the following lemma.
\begin{lemma}\label{lem:kappaNprod} Suppose $f(\psi,\zeta_{\psi})$ and $g(\psi,\zeta_{\psi})$ satisfy respectively $(\kappa_{1}, \mathcal{N}_{1}, \mathcal{M}_{1})$-bounds and $(\kappa_{2}, \mathcal{N}_{2}, \mathcal{M}_{2})$-bounds. Then, their product $f(\psi,\zeta_{\psi})g(\psi,\zeta_{\psi})$ satisfies $(\kappa_{1} \kappa_{2}, (\mathcal{N}_{1} + \mathcal{N}_{2}), (\mathcal{M}_{1} + \mathcal{M}_{2}))$-bounds.
\end{lemma}
\begin{proof} Set $h(\psi,\zeta_{\psi}) := f(\psi,\zeta_{\psi})g(\psi,\zeta_{\psi}) \equiv \sum_{\underline{a},\underline{b}} \, h_{\underline{a},\underline{b}} \, \psi^{\underline{a}} \, \zeta_{\psi}^{\underline{b}}$. The coefficients $h_{\underline{a},\underline{b}}$ are given by the following expression:
\begin{equation}
h_{\underline{a},\underline{b}} = 
\sum_{\substack{\underline{a_{1}},\underline{b_{1}},\underline{a_{2}},\underline{b_{2}} \\ (\underline{a_{1}},\underline{b_{1}})+ (\underline{a_{2}},\underline{b_{2}}) = (\underline{a},\underline{b})}} \,
\, \mathrm{sign}(\underline{a_{1}},\underline{b_{1}}; \underline{a_{2}},\underline{b_{2}})
f_{\underline{a_{1}},\underline{b_{1}}} \,  \,g_{\underline{a_{2}},\underline{b_{2}}}\;,
\end{equation}
for a suitable sign $\mathrm{sign}(\underline{a_{1}},\underline{b_{1}} ; \underline{a_{2}},\underline{b_{2}}) \in \{ \pm \}$, which we shall leave unspecified. By the triangle inequality and by the hypotheses on $f_{\underline{a_{1}},\underline{b_{1}}}$ and $g_{\underline{a_{2}},\underline{b_{2}}}$:
\begin{equation}
\begin{split}
|h_{\underline{a},\underline{b}} | 
& \leq \kappa \kappa'
\sum_{\substack{\underline{a_{1}},\underline{b_{1}},\underline{a_{2}},\underline{b_{2}} \\ (\underline{a_{1}},\underline{b_{1}})+ (\underline{a_{2}},\underline{b_{2}}) = 
(\underline{a},\underline{b})}} \mathcal{N}^{|\underline{a_{1}}|} \,(\mathcal{N}')^{|\underline{a_{2}}|} \,\mathcal{M}^{|\underline{b_{1}}|} \, (\mathcal{M}')^{|\underline{b_{2}}|}
\\
& = \kappa \kappa' \left (\mathcal{N} + \mathcal{N}' \right )^{|\underline{a}|}  \left (\mathcal{M} + \mathcal{M}' \right )^{|\underline{b}|} \;,
\end{split}
\end{equation}
and the claim is proven. 
\end{proof}
Let us go back to (\ref{eq:V0prod}). By Eqs. (\ref{eq:kappaNi}) together with Lemma \ref{lem:kappaNprod}, we get that $\exp\{- V_{\text{f}}^{(0)}(\Phi, \zeta)\}$ satisfies $(\kappa, \mathcal{N}, \mathcal{M})$ bounds with, for $\lambda$ small enough:
\begin{eqnarray}\label{eq:kappaN1}
&&\kappa = 1\;,\qquad \mathcal{N} = 3L^{-1}\lambda^{\frac{1}{4}}\;,\nonumber\\
&&\mathcal{M} = 8\lambda^{\frac{1}{4}} L^{\frac{3}{2}} + (2\lambda L)^{\frac{1}{2}} \|\phi\| + (2\lambda L^{3})^{\frac{1}{2}} \|\zeta_{\phi}\| + 8 \lambda^{\frac{3}{4}} L^{2} \|\phi\| \|\zeta_{\phi}\|\;.
\end{eqnarray}
\begin{remark}\label{rem:fn0}
For later use, notice that the function $\exp\{- V_{\text{f}}^{(0)}(\Phi, \zeta)\} - 1$ satisfies the same $(\kappa, \mathcal{N}, \mathcal{M})$-bounds: this is due to the fact that subtraction by $1$ simply sets to zero the first coefficient of the Grassmann expansion, $f_{\underline{0}, \underline{0}} = 0$. Also, the Grassmann expansion of the function $\exp\{- V_{\text{f}}^{(0)}(\Phi, \zeta)\} - 1$ is such that $f_{\underline{a}, \underline{0}} = 0$ for all $\underline{a}$.
\end{remark}
Next, we have to perform the Grassmann Gaussian integration with respect to the variable $\zeta_{\psi}$. The result is a function of the Grassmann variable $\psi$. The next lemma allows to get bounds on the coefficients of the new Grassmann polynomial, in terms of the $(\kappa, \mathcal{N}, \mathcal{M})$-bounds of the integrand. 
\begin{lemma}\label{lem:kappaNint} Suppose $f(\psi,\zeta_{\psi}) = \sum_{\underline{a},\underline{b}} f_{\underline{a},\underline{b}}\, \psi^{\underline{a}}\, \zeta_{\psi}^{\underline{b}}$ satisfies $(\kappa,\mathcal{N},\mathcal{M})$-bounds. Then, the function:
\begin{equation}
h(\psi) := \int d\mu_{\psi}(\zeta_{\psi}) f(\psi, \zeta_{\psi})
\end{equation}
satisfies $(\kappa', \mathcal{N'})$-bounds with:
\begin{equation}
\kappa' = \kappa (1 + 12 \mathcal{M}^{2} + 2\mathcal{M}^{4})\;,\quad \mathcal{N}' = \mathcal{N}\;.
\end{equation}
Furthermore, if $f_{\underline{a},\underline{0}} = 0$ for all $\underline{a}$, then:
\begin{equation}
\kappa' = \kappa (12 \mathcal{M}^{2} + 2\mathcal{M}^{4})\;,\quad \mathcal{N}' = \mathcal{N}\;.
\end{equation}
\end{lemma}
\begin{proof} Set $h(\psi) = \sum_{\underline{a}}  h_{\underline{a}} \, \psi^{\underline{a}}$. Notice that:
\begin{equation}
\int d\mu_{\psi}(\zeta_{\psi})\, 1 = 1\;,\qquad \int d\mu_{\psi} (\zeta_{\psi})\, \zeta_{\psi}^{\underline{b}} = 0\quad \text{unless $|\underline{b}|$ is even.}
\end{equation}
If $|\underline{b}|$ is even, the outcome of the integration is bounded by $2$. Therefore:
\begin{eqnarray}
|h_{\underline{a}}| &\leq& \sum_{\underline{b}} |f_{\underline{a}, \underline{b}}| \Big| \int d\mu_{\psi}(\zeta_{\psi})\, \zeta_{\psi}^{\underline{b}} \Big|\nonumber\\
&\leq& \kappa \mathcal{N}^{|\underline{a}|} (1 + 12 \mathcal{M}^{2} + 2\mathcal{M}^{4})\;.
\end{eqnarray}
If furthermore $f_{\underline{a},\underline{0}} = 0$ for all $\underline{a}$, then the $1$ in the last parenthesis is not present.
\end{proof}
Eqs. (\ref{eq:kappaN1}) together with Lemma \ref{lem:kappaNint} allow to estimate the coefficients $C_{n}^{(0)}(\phi/L, \zeta_{\phi})$ in Eq. (\ref{eq:R00}). This is the content of the next proposition.
\begin{proposition}\label{prp:Cn0} There exists a universal constant $C>0$ such that the following is true. The coefficients $C_{n}^{(0)}(\phi/L, \zeta_{\phi})$ satisfy the following bounds, for $\lambda$ small enough and for all $\zeta_{\phi} \in \mathbb{C}^{4}$, $\phi \in \mathbb{C}^{4}$:
\begin{eqnarray}\label{eq:bdC}
&&\Big| C^{(0)}_{n}(\phi/L, \zeta_{\phi}) - \delta_{n,0} \Big| \\
&&\leq K\alpha(\phi, \zeta_{\phi}) \lambda^{\frac{n}{2}} \Big(\lambda^{\frac{1}{2}} L^{3} + \lambda L \|\phi\|^{2} + \lambda L^{3} \|\zeta_{\phi}\|^{2} + \lambda^{\frac{3}{2}} L^{4} \|\phi\|^{2} \|\zeta_{\phi}\|^{2} \Big)\nonumber
\end{eqnarray}
where:
\begin{equation}
\alpha(\phi, \zeta_{\phi}) = 1 + \lambda L \|\phi\|^{2} + \lambda L^{3} \|\zeta_{\phi}\|^{2} + \lambda^{\frac{3}{2}} L^{4} \|\phi\|^{2} \|\zeta_{\phi}\|^{2}\;.
\end{equation}

\end{proposition}
\begin{proof} Let us rewrite the outcome of the Gaussian Grassmann integration in Eq. (\ref{eq:fermstep}) as $1 + h(\psi; \phi, \zeta_{\phi})$, with, using that $\int d\mu_{\psi}(\zeta_{\psi}) 1 = 1$:
\begin{eqnarray}
h(\psi; \phi, \zeta_{\phi}) &=& \int d\mu_{\psi}(\zeta_{\psi})\, \Big(e^{-V_{\text{f}}^{(0)}(\Phi, \zeta)} - 1\Big) \nonumber\\
&=& \sum_{\underline{a}, \underline{b}} f_{\underline{a}, \underline{b}}(\phi, \zeta_{\phi}) \psi^{\underline{a}}  \int d\mu_{\psi}(\zeta_{\psi}) \zeta_{\psi}^{\underline{b}}\nonumber\\
&\equiv& \sum_{\underline{a}} h_{\underline{a}}(\phi, \zeta_{\phi}) L^{-|\underline{a}|}\psi^{\underline{a}}\;.
\end{eqnarray}
Recalling Remark \ref{rem:fn0}, the function $e^{-V_{\text{f}}^{(0)}(\Phi, \zeta)} - 1$ satisfies $(\kappa, \mathcal{N}, \mathcal{M})$-bounds with $\kappa$, $\mathcal{N}$, $\mathcal{M}$ given by (\ref{eq:kappaN1}). Also, by the same remark, $f_{\underline{a}, \underline{0}} = 0$ for all $\underline{a}$. Therefore, by Lemma \ref{lem:kappaNint}, the function $h(\psi; \phi, \zeta_{\phi})$ satisfies $(\kappa, \mathcal{N})$-bounds with:
\begin{eqnarray}\label{eq:hest0}
\kappa &=& C\alpha(\phi, \zeta_{\phi})\Big(\lambda^{\frac{1}{2}} L^{3} + \lambda L \|\phi\|^{2} + \lambda L^{3} \|\zeta_{\phi}\|^{2} + \lambda^{\frac{3}{2}} L^{4} \|\phi\|^{2} \|\zeta_{\phi}\|^{2} \Big) \nonumber\\
\alpha(\phi, \zeta_{\phi}) &:=& 1 + \lambda L \|\phi\|^{2} + \lambda L^{3} \|\zeta_{\phi}\|^{2} + \lambda^{\frac{3}{2}} L^{4} \|\phi\|^{2} \|\zeta_{\phi}\|^{2}\nonumber\\
\mathcal{N} &=& 3 L^{-1} \lambda^{\frac{1}{4}}\;.
\end{eqnarray}
The final statements follows by using that $C_{0} - 1= h_{\underline{0}}$, $C_{1} = h_{\underline{a}}$ with $|\underline{a}| = 2$ and $C_{2} = h_{\underline{a}}$ with $|\underline{a}| = 4$.
\end{proof}
\subsubsection{Integration of the bosonic fluctuation field}
Next, we shall compute:
\begin{equation}\label{eq:intbos0}
\int d\mu_{\phi}(\zeta_{\phi})\, e^{-V_{\text{b}}^{(0)}(\Phi, \zeta_{\phi})} \sum_{n=0,1,2} C^{(0)}_{n}(\phi/L, \zeta_{\phi}) (\psi \cdot \psi)^{n} L^{-2n}\;.
\end{equation}
We rewrite:
\begin{eqnarray}
V_{\text{b}}^{(0)}(\Phi, \zeta_{\phi}) &=& \widetilde V^{(0)}_{\text{b}}(\phi, \zeta_{\phi}) + \widehat{V}^{(0)}_{\text{b}}(\Phi, \zeta_{\phi})\nonumber\\
\widetilde V^{(0)}_{\text{b}}(\phi, \zeta_{\phi}) &:=& V_{\text{b}}^{(0)}(\Phi, \zeta_{\phi})\mid_{\psi = 0}\;;
\end{eqnarray}
that is:
\begin{eqnarray}\label{eq:tildeV0}
\widetilde V^{(0)}_{\text{b}}(\phi, \zeta_{\phi}) &=& 4 \lambda L (\phi\cdot \zeta_{\phi})^{2} + 2\lambda L (\phi \cdot \phi)(\zeta_{\phi}\cdot \zeta_{\phi}) + \lambda L^{3} (\zeta_{\phi}\cdot \zeta_{\phi})^{2} \nonumber\\ && + i\mu L^{3} (\zeta_{\phi}\cdot \zeta_{\phi})\\
\widehat{V}^{(0)}_{\text{b}}(\Phi, \zeta_{\phi}) &=& 2\lambda L (\psi \cdot \psi) (\zeta_{\phi}\cdot \zeta_{\phi})\;.\nonumber
\end{eqnarray}
We then get:
\begin{eqnarray}\label{eq:Rn}
&&e^{-V_{\text{b}}^{(0)}(\Phi, \zeta_{\phi})} \sum_{n=0,1,2} C^{(0)}_{n}(\phi/L, \zeta_{\phi}) (\psi\cdot \psi)^{n} L^{-2n} \\
&&\equiv  e^{-\widetilde V_{\text{b}}^{(0)}(\phi, \zeta_{\phi})} \sum_{n=0,1,2}  D^{(0)}_{n}(\phi/L, \zeta_{\phi})(\psi\cdot \psi)^{n} L^{-2n}\;,\nonumber
\end{eqnarray}
where:
\begin{equation}\label{eq:Dn}
D_{n}^{(0)}(\phi/L, \zeta_{\phi}) = \sum_{k = 0}^{n} \frac{(-2\lambda L^{3})^{k} (\zeta_{\phi}\cdot \zeta_{\phi})^{k}}{k!}   C_{n-k}^{(0)}(\phi/L, \zeta_{\phi})\;.
\end{equation}
The bounds for these new coefficients are collected in the next proposition.
\begin{proposition}\label{prp:bdD} There exists a universal constant $C>0$ such that the following is true. The coefficients $D_{n}^{(0)}(\phi/L, \zeta_{\phi})$ satisfy the following bounds, for $\lambda$ small enough and for all $\zeta_{\phi} \in \mathbb{C}^{4}$, $\phi \in \mathbb{C}^{4}$:
\begin{eqnarray}\label{eq:propR}
&&\Big| D^{(0)}_{n}(\phi/L, \zeta_{\phi}) - \delta_{n,0} \Big| \nonumber\\
&&\leq C \tilde \alpha(\phi, \zeta_{\phi}) \lambda^{\frac{n}{2}} ( \lambda^{\frac{1}{2}} L^{3} + \lambda L \|\phi\|^{2} + \lambda L^{3} \|\zeta_{\phi}\|^{2} + \lambda^{\frac{3}{2}} L^{4} \|\phi\|^{2} \|\zeta_{\phi}\|^{2} ) \nonumber\\
&&\quad + C\lambda^{n} L^{3n} \|\zeta_{\phi}\|^{2n} (1 - \delta_{n,0})\;,
\end{eqnarray}
where:
\begin{equation}
\tilde \alpha(\phi, \zeta_{\phi}) := \alpha(\phi, \zeta_{\phi}) (1 + \lambda^{\frac{1}{2}} L^{3} \| \zeta_{\phi} \|^{2} + \lambda L^{6} \| \zeta_{\phi} \|^{4} )\;. \nonumber\\
\end{equation}
\end{proposition}
\begin{proof} The statement for $n=0$ is trivial, since $D_{0}^{(0)} = C^{(0)}_{0}$. Consider now $n\neq 0$. From Eq. (\ref{eq:Dn}) we get:
\begin{eqnarray}\label{eq:Dn00}
|D^{(0)}_{n}(\phi/L, \zeta_{\phi})| &\leq& C\sum_{k=0}^{n} \lambda^{k} L^{3k} \|\zeta_{\phi}\|^{2k} |C^{(0)}_{n-k}(\phi/L, \zeta_{\phi})| \nonumber\\
&\leq& C\sum_{k=0}^{n} \lambda^{k} L^{3k} \|\zeta_{\phi}\|^{2k} |C^{(0)}_{n-k}(\phi/L, \zeta_{\phi}) - \delta_{n-k,0}|\nonumber\\&& + C\lambda^{n} L^{3n} \|\zeta_{\phi}\|^{2n}\;. 
\end{eqnarray}
Using the bounds for $C^{(0)}_{n-k}$, Proposition \ref{prp:Cn0}, we get:
\begin{eqnarray}
|D^{(0)}_{n}(\phi/L, \zeta_{\phi})| &\leq& C \lambda^{\frac{n}{2}}\Big(\sum_{k=0}^{n} \lambda^{\frac{k}{2}} L^{3k} \|\zeta_{\phi}\|^{2k} \Big)\\&& \cdot \alpha(\phi, \zeta_{\phi}) ( \lambda^{\frac{1}{2}} L^{3} + \lambda L \|\phi\|^{2} + \lambda L^{3} \|\zeta_{\phi}\|^{2} + \lambda^{\frac{3}{2}} L^{4} \|\phi\|^{2} \|\zeta_{\phi}\|^{2} )\nonumber\\
&& + C\lambda^{n} L^{3n} \|\zeta_{\phi}\|^{2n}\;.\nonumber
\end{eqnarray}
which concludes the proof.
\end{proof}
We are now ready to integrate the $\zeta_{\phi}$ variable. We shall perform the integration for $\phi \in \mathbb{C}^{4}$: the reason being that the bounds on the kernels on the next scales will be obtained via Cauchy estimates. In order to integrate the field $\zeta_{\phi}$ we shall exploit the oscillations of the complex Gaussian, via the next lemma.
\begin{lemma}[Stationary phase expansion.]\label{lem:osc}
Let $f$ be a Schwartz function on $\mathbb{R}^{4}$. Then, for any $m\in \mathbb{N}$, the following identity holds true:
\begin{equation}\label{eq:lemgen}
\int d\mu_{\phi}(\zeta_{\phi}) f(\zeta_{\phi}) = \sum_{j = 0}^{m-1} d_{j} (\Delta^{j} f)(0) + \mathcal{E}_{m}(f)\nonumber\\
\end{equation}
where $\Delta$ denotes the Laplacian, $\Delta = \sum_{i=1,2}\sum_{\sigma = \uparrow, \downarrow} \partial^{2}_{\zeta_{\phi, i, \sigma}}$, $d_{j} = \frac{(-i)^{j}}{4^{j} j!}$ and
\begin{equation}\label{eq:remcau}
|\mathcal{E}_{m}(f)| \leq C_{m} \int_{\mathbb{R}^{4}} dp\, \|p\|^{2m} |\hat f(p)|\;.
\end{equation}
\end{lemma}
\begin{proof} See Appendix \ref{app:osc}.
\end{proof}
The next lemma is the key technical tool that we will use to estimate the derivatives and the error terms arising from the stationary phase expansion.
\begin{lemma}[Bounds for stationary phase expansions]\label{lem:cau} Let $f(z) = f(z_{1}, \ldots, z_{4})$ be a complex-valued function on $\mathbb{C}^{4}$.
\begin{itemize}
\item[(a)] (Cauchy estimate) Let $R'>0$, and suppose that $f(z)$ is an analytic function in all $z_{i}$, $i=1,\ldots, 4$, for $z\in B_{R'} \subset \mathbb{C}^{4}$, with $B_{R'}$ the ball of radius $R'$ centered at $z=0$. Suppose that $|f(z)| \leq f_{R'}$ for all $z\in B_{R'}$. Let $0\leq R<R'$. Then, for all multi-indices $\alpha \in \mathbb{N}^{4}$ and for all $z\in B_{R}$:
\begin{equation}\label{eq:cau0}
|\partial_{z}^{\alpha} f(z)|\leq \frac{C_{\alpha}}{(R' - R)^{|\alpha|}} f_{R'}\;,
\end{equation}
with $|\alpha| = \sum_{i} \alpha_{i}$.
\item[(b)] (Decay of Fourier transforms) Let $W'>0$, and suppose that $f(z)$ is an analytic function in all $z_{i}$, $i=1,\ldots, 4$, for $z\in \mathbb{R}^{4}_{W'}$ with:
\begin{equation}
\mathbb{R}^{4}_{W'} := \{ z\in \mathbb{C}^{4} \mid | \text{Im}\, z_{i} | < W',\, i=1,\ldots, 4 \}\;. 
 \end{equation}
Let furthermore $\hat f(p_{1},\ldots, p_{4})$ be the Fourier transform of the restriction of $f$ to $\mathbb{R}^{4}$. Suppose that, for all $0\leq W<W'$ there exist a constant $0\leq F_{W}(f)<\infty$ such that for all $w\in \mathbb{R}^{4}_{W}$:
\begin{equation}\label{eq:FW}
\int_{\mathbb{R}^{4}} dx\, |f(w + x)| \leq F_{W}(f)\;.
\end{equation}
Then, for all $k\in \mathbb{N}$, there exist $C_{k}>0$ such that:
\begin{equation}\label{eq:estW}
|\hat f(p)| \leq \frac{C_{k} F_{W}(f)}{1 + (W\|p\| )^{k}}\;.
\end{equation}
In particular, let $\mathcal{E}_{m}(f)$ be the error term in the stationary phase expansion (\ref{eq:lemgen}). Then, there exists a universal constant $K_{m}>0$ such that:
\begin{equation}\label{eq:rembd}
|\mathcal{E}_{m}(f)| \leq  K_{m} {W}^{-4 - 2m} F_{W}(f)\;.
\end{equation}
\end{itemize}
\end{lemma}
\begin{proof} See Appendix \ref{app:osc}.
\end{proof}
%
%
%
%
%
As a test run, let us estimate the normalization factor $N^{(0)}$ in Eq. (\ref{eq:U1}). Notice that, as a consequence of the localization theorem, see Theorem \ref{prp:SUSY}, supersymmetry (in the sense of Definition \ref{def:SUSY}) implies that, see Remark \ref{rem:Nh}:
\begin{equation}
N^{(0)} = 1\;.
\end{equation}
Nevertheless, the simple procedure discussed below will be generalized to the computation of the effective potential and of the correlation functions, where the localization theorem cannot be applied because of the lack of supersymmetry. We rewrite:
\begin{equation}\label{eq:norm1}
N^{(0)} = \int d\mu_{\phi}(\zeta_{\phi})\, e^{-\widetilde V_{\text{b}}^{(0)}(0, \zeta_{\phi})} D^{(0)}_{0}(0, \zeta_{\phi})\;.
\end{equation}
Apply Lemma \ref{lem:osc} recalling that $\widetilde V_{\text{b}}^{(0)}(0, 0) = 0$:
\begin{equation}\label{eq:N0exp}
N^{(0)} = D^{(0)}_{0}(0,0) + d_{1} \Big(\Delta e^{-\widetilde V_{\text{b}}^{(0)}(0, \cdot)} D^{(0)}_{0}(0, \cdot)\Big)(0) + \mathcal{E}_{2}\big(e^{-\widetilde V_{\text{b}}^{(0)}(0, \cdot)} D^{(0)}_{0}(0, \cdot)\big)\;;
\end{equation}
let us estimate the various terms. By Proposition \ref{prp:bdD}, we have $|D^{(0)}_{0}(0,0) - 1| \leq C \lambda^{\frac{1}{2}} L^{\frac{3}{2}}$. To estimate the second term, we write (all derivatives correspond to the field $\zeta_{\phi}$):
\begin{eqnarray}
&&\Big(\Delta e^{-\widetilde V_{\text{b}}^{(0)}(0, \cdot)} D^{(0)}_{0}(0, \cdot)\Big)(0) = \Big(\Delta e^{-\widetilde V_{\text{b}}^{(0)}(0, \cdot)}\Big)(0) D^{(0)}_{0}(0,0)\nonumber\\&& + \Big(\Delta D^{(0)}_{0}(0, \cdot)\Big)(0) + 2\Big( \nabla e^{-\widetilde V_{\text{b}}^{(0)}(0, \cdot)}  \Big)(0) \cdot \Big(\nabla D^{(0)}_{0}(0, \cdot) \Big)(0)\;.
\end{eqnarray}
From the definition of $\widetilde{V}^{(0)}_{\text{b}}(\phi, \zeta_{\phi})$, Eq. (\ref{eq:tildeV0}), we get:
\begin{equation}
\Big( \nabla e^{-\widetilde V_{\text{b}}^{(0)}(0, \cdot)}  \Big)(0) = 0\;,\qquad \Big| \Big(\Delta e^{-\widetilde V_{\text{b}}^{(0)}(0, \cdot)}\Big)(0) \Big| \leq C|\mu| \leq K|\lambda|\;.
\end{equation}
Moreover, Proposition \ref{prp:bdD} together with Cauchy estimates for the derivatives give:
\begin{eqnarray}
|D^{(0)}_{0}(0,0) - 1| &\leq& C \lambda^{\frac{1}{2}} L^{3}\nonumber\\
\Big| \Big(\Delta D^{(0)}_{0}(0, \cdot)\Big)(0) \Big| &\leq& C_{2} \lambda^{\frac{1}{2}}\lambda^{\frac{1}{2}}  L^{3}  \nonumber\\
&\leq& K \lambda L^{3}\;. 
\end{eqnarray}
The second estimate follows from the Cauchy estimate (\ref{eq:cau0}) and from the bound (\ref{eq:propR}), after taking $R = \lambda^{-\frac{1}{4}}$.

Consider now the last term in Eq. (\ref{eq:N0exp}). We claim that, for some $L$-dependent constant $C_{L}>0$:
\begin{equation}\label{eq:N0}
|\mathcal{E}_{2}\big(e^{-\widetilde V_{\text{b}}^{(0)}(0, \cdot)} D^{(0)}_{0}(0, \cdot)\big)| \leq C_{L}\lambda\;.
\end{equation}
To prove this estimate we proceed as follows. Let $f(\zeta_{\phi}) := e^{-\widetilde V_{\text{b}}^{(0)}(0, \zeta_{\phi})} D^{(0)}_{0}(0, \zeta_{\phi})$. This function is trivially entire in $\zeta_{\phi}\in \mathbb{C}^{4}$. Let us estimate it. We have:
\begin{equation}\label{eq:356}
\big| e^{-\widetilde V_{\text{b}}^{(0)}(0, \zeta_{\phi})} \big| \leq e^{- \lambda L^{3} \text{Re}\,\left ((\zeta_{\phi}\cdot \zeta_{\phi})^{2}\right ) + C |\mu| L^{3} \|\zeta_{\phi}\|^{2}}\;.
\end{equation}
Also,
\begin{eqnarray}
\text{Re}\,\left ((\zeta_{\phi}\cdot \zeta_{\phi})^{2}\right ) &=& (\text{Re}\, (\zeta_{\phi}\cdot \zeta_{\phi}))^{2} - (\text{Im}\, (\zeta_{\phi}\cdot \zeta_{\phi}))^{2} \nonumber\\
&=& ( \|\text{Re} \zeta_{\phi}\|^{2} - \| \text{Im} \zeta_{\phi} \|^{2} )^{2} - ( 2( \text{Im}\zeta_{\phi}\cdot \text{Re}\zeta_{\phi} ) )^{2}\nonumber\\
&\geq& ( \|\text{Re} \zeta_{\phi}\|^{2} - \| \text{Im} \zeta_{\phi} \|^{2} )^{2} - 4\| \text{Im}\,\zeta_{\phi}  \|^{2} \|\text{Re}\,\zeta_{\phi} \|^{2}\nonumber\\
&=& \|\zeta_{\phi}\|^{4} - 8 \| \text{Im}\zeta_{\phi} \|^{2} \| \text{Re}\zeta_{\phi} \|^{2}\;.
\end{eqnarray}
Let $\zeta_{\phi}\in \mathbb{R}_{W}^{4}$, with $W = \lambda^{-1/4}$. Then:
\begin{equation}\label{eq:358}
\text{Re}\,\left ((\zeta_{\phi}\cdot \zeta_{\phi})^{2}\right ) \geq \|\zeta_{\phi}\|^{4} - 8 \lambda^{-\frac{1}{2}} \| \zeta_{\phi}\|^{2}\;.
\end{equation}
Therefore, using that $|\mu|\leq C|\lambda|$, for some $K>0$ and taking $\lambda$ small enough:
\begin{equation}\label{eq:bdexp}
\big| e^{-\widetilde V_{\text{b}}^{(0)}(0, \zeta_{\phi})} \big| \leq \left\{ \begin{array}{cc} K e^{-\frac{L^{3}}{8} \lambda \|\zeta_{\phi}\|^{4}} & \text{for $\|\zeta_{\phi}\|> 4 \lambda^{-\frac{1}{4}}$} \\ K e^{c L^{3}} & \text{for $\|\zeta_{\phi}\|\leq 4 \lambda^{-\frac{1}{4}}$.}\end{array} \right.
\end{equation}
Recalling the bound (\ref{eq:propR}) for $D^{(0)}_{0}(0, \zeta_{\phi})$, one finds that the bound (\ref{eq:FW}) is satisfied, with $F_{W}(e^{-\widetilde V_{\text{b}}^{(0)}(0, \cdot)} D^{(0)}_{0}(0, \cdot)) = C_{L}\lambda^{-1}$. Therefore, from Eq. (\ref{eq:rembd}):
\begin{equation}
|\mathcal{E}_{2}\big(e^{-\widetilde V_{\text{b}}^{(0)}(0, \cdot)} D^{(0)}_{0}(0, \cdot)\big)| \leq \widetilde{C}_{L} \lambda^{\frac{1}{4}(4 + 4)} \lambda^{-1} = \widetilde{C}_{L} \lambda\;.
\end{equation}
This concludes the check of Eq. (\ref{eq:N0}). 

Next, we shall consider, for $\phi \neq 0$:
\begin{equation}
\int d\mu_{\phi}(\zeta_{\phi})\, e^{-\widetilde V_{\text{b}}^{(0)}(\phi, \zeta_{\phi})} D^{(0)}_{n}(\phi/L, \zeta_{\phi})\;;
\end{equation}
to do this, we shall discuss separately a small and a large field regime for the bosonic field $\phi\in \mathbb{C}^{4}$. More precisely, we shall consider separately the following cases: $\phi/L \in \mathbb{S}^{(0)}$, the {\it small field set}, and $\phi/L \in \mathbb{L}^{(0)}$, the {\it large field set}:
\begin{eqnarray}
\mathbb{S}^{(0)} &:=& \{ \phi \in \mathbb{C}^{4} \mid \| \phi \|\leq \lambda^{-1/4}  \}\;,\nonumber\\\mathbb{L}^{(0)} &:=& \Big\{ \phi \in \mathbb{C}^{4} \mid \|\phi\|> \lambda^{-1/4}\;,\quad \| \text{Im}\, \phi \|\leq \lambda^{-1/4} \Big\}\;.
\end{eqnarray}
%
\subsubsection{Small field regime}\label{sec:sfregime}
Let $\phi \in L\mathbb{S}^{(0)}$. Let us define:
\begin{equation}\label{eq:Rgen}
E^{(0)}_{n}(\phi) := L^{-2n}\int d\mu_{\phi}(\zeta_{\phi})\, e^{-\widetilde V_{\text{b}}^{(0)}(\phi, \zeta_{\phi})} D^{(0)}_{n}(\phi/L, \zeta_{\phi})\;.
\end{equation}
As a preliminary remark, notice that $E^{(0)}_{n}$ is analytic in $\phi \in L\mathbb{S}^{(0)}$, since the integral is absolutely convergent in $\zeta_{\phi}$ uniformly for $\phi \in L\mathbb{S}^{(0)}$ and the integrand is entire in $\phi$ (analyticity follows from dominated convergence theorem and from Morera's theorem). Let us now prove bounds for $E_{n}^{(0)}(\phi)$.

Consider first the case $n=0$. By the stationary phase expansion, Lemma \ref{eq:lemgen}, we get:
\begin{eqnarray}\label{eq:E00}
E^{(0)}_{0}(\phi) &=& D^{(0)}_{0}(\phi/L, 0) + d_{1}\Big(\Delta e^{-\widetilde V_{\text{b}}^{(0)}(\phi, \cdot)} D^{(0)}_{0}(\phi/L, \cdot)\Big)(0)\nonumber\\&& + \mathcal{E}_{2}\Big(e^{-\widetilde V_{\text{b}}^{(0)}(\phi, \cdot )} D^{(0)}_{0}(\phi/L, \cdot)\Big)\;.
\end{eqnarray}
We shall proceed as we did for the analysis of the normalization factor $N^{(0)}$. Consider the first term in Eq. (\ref{eq:E00}). By Proposition \ref{prp:bdD}, we get, for $\| \phi \| \leq \lambda^{-\frac{1}{4}} L$:
\begin{equation}\label{eq:E00a}
|D^{(0)}_{0}(\phi/L, 0) - 1| \leq C\lambda^{\frac{1}{2}} L^{3}\;.
\end{equation}
To estimate the second term in Eq. (\ref{eq:E00}), we use a Cauchy estimate. To begin, notice that for $\| \phi \| \leq \lambda^{-\frac{1}{4}}L$ and for $\|\zeta_{\phi}\| \leq \lambda^{-\frac{1}{4}}$, by Proposition \ref{prp:bdD}:
\begin{equation}\label{eq:D001}
|D^{(0)}_{0}(\phi/L, \zeta_{\phi}) - 1| \leq C L^{12} \lambda^{\frac{1}{2}}\;.
\end{equation}
Therefore, by the Cauchy bound (\ref{eq:cau0}), with $R' = \lambda^{-1/4}$ and $R = 0$:
\begin{equation}\label{eq:D002}
| \Delta D^{(0)}_{0}(\phi/L, \cdot)(0) | \leq K L^{12} \lambda\;.
\end{equation}
On the other hand, recalling the definition of $\widetilde V_{\text{b}}^{(0)}$, Eq. (\ref{eq:tildeV0}):
\begin{eqnarray}\label{eq:D003}
\Big|\Big(\Delta e^{-\widetilde V_{\text{b}}^{(0)}(\phi, \cdot)}\Big)(0)\Big| = \Big| (\Delta \widetilde V_{\text{b}}^{(0)}(\phi, \cdot))(0)  \Big| &\leq& C(\lambda L \|\phi\|^{2} + |\mu| L^{3}) \nonumber\\
&\leq& C\lambda^{\frac{1}{2}} L^{3}\;.
\end{eqnarray}
Using that:
\begin{eqnarray}
&&\Big(\Delta e^{-\widetilde V_{\text{b}}^{(0)}(\phi, \cdot)} D^{(0)}_{0}(\phi/L, \cdot)\Big)(0) = \Big(\Delta e^{-\widetilde V_{\text{b}}^{(0)}(\phi, \cdot)}\Big)(0) D^{(0)}_{0}(\phi/L,0)\nonumber\\&& + \Big(\Delta D^{(0)}_{0}(\phi/L, \cdot)\Big)(0) + 2\Big( \nabla e^{-\widetilde V_{\text{b}}^{(0)}(\phi, \cdot)}  \Big)(0) \cdot \Big(\nabla D^{(0)}_{0}(\phi/L, \cdot) \Big)(0) \nonumber\\
&& = \Big(\Delta e^{-\widetilde V_{\text{b}}^{(0)}(\phi, \cdot)}\Big)(0) D^{(0)}_{0}(\phi/L,0) + \Big(\Delta D^{(0)}_{0}(\phi/L, \cdot)\Big)(0)
\end{eqnarray}
we finally get, from the bounds (\ref{eq:D001}), (\ref{eq:D002}), (\ref{eq:D003}), for $\lambda$ small enough:
\begin{equation}\label{eq:E00b}
\Big|\Big(\Delta e^{-\widetilde V_{\text{b}}^{(0)}(\phi, \cdot)} D^{(0)}_{0}(\phi/L, \cdot)\Big)(0)\Big|\leq C\lambda^{\frac{1}{2}}L^{3}\;.
\end{equation}
Consider now the third term in Eq. (\ref{eq:E00}), which is the remainder in the stationary phase expansion. We shall first extend the bound (\ref{eq:bdexp}) to $\phi \in L\mathbb{S}^{(0)}$. For $\| \phi \| \leq L\lambda^{-\frac{1}{4}}$, $\|\zeta_{\phi}\| \leq 4\lambda^{-\frac{1}{4}}$ we have:
\begin{equation}\label{eq:tildeV1}
| e^{-\widetilde V_{\text{b}}^{(0)}(\phi, \zeta_{\phi})} | \leq K e^{cL^{3}}\;.
\end{equation}
Let $W = \lambda^{-\frac{1}{4}}$. We have, for $\| \phi \| \leq L\lambda^{-\frac{1}{4}}$, $\zeta_{\phi} \in \mathbb{R}_{W}^{4}$, $\| \zeta_{\phi} \| > 4\lambda^{-\frac{1}{4}}$, proceeding as in Eq. (\ref{eq:356})-(\ref{eq:358}):
\begin{eqnarray}\label{eq:tildeV20}
| e^{-\widetilde V_{\text{b}}^{(0)}(\phi, \zeta_{\phi})} | &\leq& K e^{-\frac{L^{3}}{8} \lambda \|\zeta_{\phi}\|^{4} + 6 \lambda L \|\phi\|^{2} \|\zeta_{\phi}\|^{2}} \nonumber\\
&\leq& \widetilde{K} e^{-\frac{L^{3}}{16} \lambda \|\zeta_{\phi}\|^{4}}\;.
\end{eqnarray}
In conclusion, the bounds (\ref{eq:tildeV1}), (\ref{eq:tildeV20}), together with the estimate (\ref{eq:propR}) for $D_{0}^{(0)}(\phi/L, \zeta_{\phi})$, give:
\begin{equation}\label{eq:tildeV2}
\Big| e^{-\widetilde V_{\text{b}}^{(0)}(\phi, \zeta_{\phi})} D^{(0)}_{0}(\phi/L, \zeta_{\phi}) \Big| \leq \left\{ \begin{array}{cc} C_{L} & \text{$\|\zeta_{\phi}\| \leq \lambda^{-\frac{1}{4}}$} \\ C e^{-\frac{L^{3}}{16} \|\zeta_{\phi}\|^{4}} \lambda^{3} L^{18} \|\zeta_{\phi}\|^{8} & \text{$\|\zeta_{\phi}\| > 4\lambda^{\frac{1}{4}}$, $\zeta_{\phi} \in \mathbb{R}_{W}^{4}$.} \end{array} \right.
\end{equation}
Hence, we are in the position to apply Lemma \ref{lem:cau}, with:
\begin{equation}
F_{W}\Big( e^{-\widetilde V_{\text{b}}^{(0)}(\phi, \cdot)} D^{(0)}_{0}(\phi/L, \cdot) \Big) \leq \widetilde{C}_{L} \lambda^{-1}\;.
\end{equation}
We get, by Lemma \ref{lem:cau}:
\begin{equation}
\Big| \mathcal{E}_{2}\Big(e^{-\widetilde V_{\text{b}}^{(0)}(\phi, \cdot )} D^{(0)}_{0}(\phi/L, \cdot)\Big) \Big| \leq K_{L} \lambda^{\frac{1}{4}( 4 + 4)} \lambda^{-1} \leq K_{L} \lambda\;.
\end{equation}
This bound together with (\ref{eq:E00a}), (\ref{eq:E00b}) imply, for $\lambda$ small enough:
\begin{equation}\label{eq:bdgen}
| E^{(0)}_{0}(\phi) - 1 | \leq CL^{3} \lambda^{\frac{1}{2}}\;.
\end{equation}
Notice also that $E^{(0)}_{0}(0) = N^{(0)}$; by supersymmetry, see Remark \ref{rem:Nh},
\begin{equation}\label{eq:SUSYE0}
E_{0}^{(0)}(0) = 1\;.
\end{equation}
Let us now consider $E^{(0)}_{n}(\phi)$ for $n=1,2$. By the stationary phase expansion:
\begin{eqnarray}\label{eq:Endec}
E^{(0)}_{n}(\phi) &=& L^{-2n} D^{(0)}_{n}(\phi/L, 0) + L^{-2n}d_{1}\Big(\Delta e^{-\widetilde V_{\text{b}}^{(0)}(\phi, \cdot)} D^{(0)}_{n}(\phi/L, \cdot)\Big)(0)\nonumber\\&& + L^{-2n}\mathcal{E}_{2}\Big(e^{-\widetilde V_{\text{b}}^{(0)}(\phi, \cdot )} D^{(0)}_{n}(\phi/L, \cdot)\Big)\;.
\end{eqnarray}
By Proposition \ref{prp:bdD}:
\begin{equation}\label{eq:D0n1}
L^{-2n} | D^{(0)}_{n}(\phi/L, 0) | \leq C L^{3 - 2n} \lambda^{\frac{n}{2} + \frac{1}{2}}\;.
\end{equation}
To estimate the second term, we proceed as follows.
%
%
%
%
We write:
\begin{eqnarray}\label{eq:D0n2}
&&L^{-2n} \Big(\Delta e^{-\widetilde V_{\text{b}}^{(0)}(\phi, \cdot)} D^{(0)}_{n}(\phi/L, \cdot)\Big)(0) \\
&&= L^{-2n}\Big(\Delta e^{-\widetilde V_{\text{b}}^{(0)}(\phi, \cdot)}\Big)(0) D^{(0)}_{n}(\phi/L,0) + L^{-2n}\Big(\Delta D^{(0)}_{n}(\phi/L, \cdot)\Big)(0)\;;\nonumber
\end{eqnarray}
the first term is estimated using the bound (\ref{eq:D003}). We get:
\begin{equation}
\Big|L^{-2n}\Big(\Delta e^{-\widetilde V_{\text{b}}^{(0)}(\phi, \cdot)}\Big)(0) D^{(0)}_{n}(\phi/L,0)\Big|\leq K L^{6 - 2n} \lambda^{\frac{n}{2} + 1}\;.
\end{equation}
Concerning the second term, it is convenient to rewrite the derivative in terms of the $C^{(0)}_{n}$ coefficients, recall the definition (\ref{eq:Dn}). We get:
\begin{eqnarray}\label{eq:cauDh0}
L^{-2n}\Big(\Delta D^{(0)}_{n}(\phi/L, \cdot)\Big)(0) &=& L^{-2n}\Big(\Delta C^{(0)}_{n}(\phi/ L, \cdot)\Big)(0)\nonumber\\&& + L^{-2n} (2\lambda L^{3}) 8 C^{(0)}_{n-1}(\phi/L, 0)\;.
\end{eqnarray}
We estimate the right-hand side using the bounds for the $C^{(0)}_{n}$ coefficients, (\ref{eq:bdC}), plus a Cauchy estimate with $R = \lambda^{-\frac{1}{4}}$ for the first term. We get:
\begin{equation}
\Big|L^{-2n}\Big(\Delta D^{(0)}_{n}(\phi/L, \cdot)\Big)(0)\Big| \leq C L^{-2n + 6} \lambda^{\frac{n}{2} + 1} + C L^{6 - 2n} \lambda^{\frac{n}{2} + 1}\;.
\end{equation}
Therefore:
\begin{equation}\label{eq:L2n}
\Big| L^{-2n} \Big(\Delta e^{-\widetilde V_{\text{b}}^{(0)}(\phi, \cdot)} D^{(0)}_{n}(\phi/L, \cdot)\Big)(0) \Big| \leq C_{L} \lambda^{\frac{n}{2} + 1}\;.
\end{equation}
Consider now the remainder term in the stationary phase expansion. Let us choose again $W = \lambda^{-1/4}$. Proceeding as in Eqs. (\ref{eq:tildeV1})-(\ref{eq:tildeV2}), we get:
\begin{eqnarray}\label{eq:tildeV3}
&&\Big| e^{-\widetilde V_{\text{b}}^{(0)}(\phi, \zeta_{\phi})} D^{(0)}_{n}(\phi/L, \zeta_{\phi}) \Big| \nonumber\\&&\qquad \leq \left\{ \begin{array}{cc} C_{L} \lambda^{\frac{n}{2}} & \text{$\|\zeta_{\phi}\| \leq \lambda^{-\frac{1}{4}}$} \\ C e^{-\frac{L^{3}}{16} \|\zeta_{\phi}\|^{4}} \lambda^{3 + \frac{n}{2}} L^{18} \|\zeta_{\phi}\|^{8} & \text{$\|\zeta_{\phi}\| > 4\lambda^{\frac{1}{4}}$, $\zeta_{\phi} \in \mathbb{R}_{W}^{4}$.} \end{array} \right.
\end{eqnarray}
Therefore, by Lemma \ref{lem:cau}:
\begin{equation}\label{eq:D0n3}
\Big| L^{-2n}\mathcal{E}_{2}\Big(e^{-\widetilde V_{\text{b}}^{(0)}(\phi, \cdot )} D^{(0)}_{n}(\phi/L, \cdot)\Big) \Big| \leq K_{L} \lambda^{\frac{n}{2} + 1}\;.
\end{equation}
In conclusion, for $\lambda$ small enough, the expression (\ref{eq:Endec}) for $E_{n}^{(0)}(\phi)$, together with the bounds (\ref{eq:D0n1}), (\ref{eq:L2n}), (\ref{eq:D0n3}) imply, for $n=1,2$:
\begin{eqnarray}\label{eq:bdgen2}
|E_{n}^{(0)}(\phi)| &\leq& 4C L^{3 - 2n} \lambda^{\frac{n}{2} + \frac{1}{2}} + K_{L} \lambda^{\frac{n}{2} + 1} \nonumber\\
&\leq& KL^{3 - 2n} \lambda^{\frac{n}{2} + \frac{1}{2}}\;.
\end{eqnarray}
As it will be clear later on, the bounds (\ref{eq:bdgen}), (\ref{eq:bdgen2}) are not enough to iterate the multiscale integration on higher scales. We shall isolate the dangerous contributions by introducing a {\it localization operation}, as follows. 

By symmetry considerations, see Appendix \ref{app:SUSY}, Remark \ref{rem:sym}, for $\phi \in \mathbb{R}^{4}$ the function $E^{(0)}_{n}(\phi)$ is radial: we shall write $E^{(0)}_{n}(\phi) = E^{(0)}_{n}(\|\phi\|)$, with a slight abuse of notation. 

\medskip

\noindent{\underline{\it Localization and renormalization.}} We define, for $\phi \in L\mathbb{S}^{(0)}$:
\begin{equation}\label{eq:tay}
\mathcal{L} E^{(0)}_{n}(\phi) := \left\{ \begin{array}{ccc} E^{(0)}_{2}(0) & \text{if $n=2$} \\ E^{(0)}_{1}(0) + \frac{1}{2}(\phi\cdot \phi) \partial^{2}_{\|\phi\|} E^{(0)}_{1}(0) & \text{if $n=1$} \\ \frac{1}{2}(\phi\cdot \phi)\partial^{2}_{\|\phi\|}E^{(0)}_{0}(0) + \frac{1}{4!} (\phi\cdot \phi)^{2} \partial_{\|\phi\|}^{4}E^{(0)}_{0}(0)  & \text{if $n=0$.} \end{array} \right.
\end{equation}
That is, the $\mathcal{L}$ operator extracts the first few orders in the Taylor expansion of $E^{(0)}_{n}$. To see this, notice that $(\phi \cdot \phi)$ is just the analytic continuation of $\|\phi\|^{2}$ from $\phi \in \mathbb{R}^{4}$ to $\phi \in \mathbb{C}^{4}$, recall Eq. (\ref{eq:prod}). Also, notice that in the expansion in $\phi \in \mathbb{R}^{4}$ of $E^{(0)}_{n}(\phi)$ odd powers of $\|\phi\|$ are forbidden, due to the fact that they are not analytic in $\phi_{i,\sigma}$. Hence Eq. (\ref{eq:tay}) collects the first few orders in the Taylor expansion in $\phi$, for $\phi \in L \mathbb{S}^{(0)}$. The terms $E^{(0)}_{1}(0)$ and $\frac{1}{2}(\phi\cdot \phi)\partial^{2}_{\|\phi\|}E^{(0)}_{0}(0)$ are {\it relevant} in the renormalization group terminology. As it will be clear later on, they correspond to an expanding direction in the RG flow. The other terms are {\it irrelevant}, thus strictly speaking there should be no need to localize them. Nevertheless, the above procedure turns out to simplify the analysis of the large field regime.

Correspondingly, we define the {\it renormalization} operation $\mathcal{R}$ so that the function $\mathcal{R} E_{n}^{(0)}(\phi)$ contains all the higher order terms in the Taylor expansion. Recalling that $E_{0}^{(0)}(0) = 1$, Eq. (\ref{eq:SUSYE0}):
\begin{eqnarray}
E_{0}^{(0)}(\phi) &=& 1 + \mathcal{L} E^{(0)}_{0}(\phi) + \mathcal{R} E^{(0)}_{0}(\phi) \nonumber\\
E_{n}^{(0)}(\phi) &=& \mathcal{L} E^{(0)}_{n}(\phi) + \mathcal{R} E^{(0)}_{n}(\phi)\;,\qquad n=1,2\;.
\end{eqnarray}
To estimate the derivatives, we will use again Cauchy bounds. More precisely, we will consider the functions $E^{(0)}_{n}(\phi)$ in a much smaller domain than the original analyticity domain $L \mathbb{S}^{(0)}$, for which the bounds (\ref{eq:bdgen}), (\ref{eq:bdgen2}) hold true. By the general estimate (\ref{eq:cau0}), we will use that every derivative introduces a gain with respect to the $L^{\infty}$ bound of $E^{(0)}_{n}(\phi)$ in $L \mathbb{S}^{(0)}$, proportional to the inverse of the distance between the smaller domain and $L \mathbb{S}^{(0)}$.

Given $\lambda_{1} \in \mathbb{C}$ such that $|L \lambda_{1} - \lambda | \leq C \lambda^{3/2}$, we define:
\begin{equation}
\mathbb{S}^{(1)} := \{ \phi \in \mathbb{C}^{4} \mid \|\phi\|\leq \vert \lambda_{1} \vert^{-1/4} \} \subset L\mathbb{S}^{(0)}\;.
\end{equation}
Notice that, for $L$ large enough and some universal constant $c >0$:
\begin{equation}\label{eq:distan}
\text{dist}(L\mathbb{S}^{c}_{0}, \mathbb{S}_{1}) = L\lambda^{-1/4} - \vert \lambda_{1} \vert^{-1/4} \geq c L \lambda^{-1/4}\;.
\end{equation}
Let us estimate $\mathcal{R} E^{(0)}_{n}(\phi)$ in $\mathbb{S}^{(1)}$ as a Lagrange remainder. We have, using the bounds (\ref{eq:bdgen}), (\ref{eq:bdgen2}), together with (\ref{eq:distan}) and the Cauchy estimate (\ref{eq:cau0}), for all $\phi \in \mathbb{S}^{(1)}$:
\begin{eqnarray}\label{eq:RE0}
&&\big| \mathcal{R} E^{(0)}_{2}(\phi) \big| \leq C L^{-3} \lambda^{2} \|\phi\|^{2}\;,\qquad \big|\mathcal{R} E^{(0)}_{1}(\phi)\big| \leq C L^{-3}\lambda^{2} \|\phi\|^{4}\;,\nonumber\\
&&\qquad\qquad\qquad \quad \big|\mathcal{R} E^{(0)}_{0}(\phi) \big| \leq C L^{-3} \lambda^{2} \|\phi\|^{6}\;.
\end{eqnarray}
Also, again by Cauchy estimates:
\begin{equation}\label{eq:Ebeta}
|\partial_{\|\phi\|}^{2} E^{(0)}_{1}(0)|\leq C L^{-1}\lambda^{\frac{3}{2}}\;,\quad | \partial_{\|\phi\|}^{2} E^{(0)}_{0}(0) |\leq CL\lambda\;,\quad | \partial_{\|\phi\|}^{4} E^{(0)}_{0}(0) |\leq CL^{-1}\lambda^{3/2}\;.
\end{equation}
Finally, we set:
\begin{eqnarray}\label{eq:Ebeta2}
&&\gamma_{\psi,2}^{(0)} := E^{(0)}_{1}(0)\;,\qquad \gamma_{\phi,2}^{(0)} := \frac{1}{2}\partial^{2}_{\|\phi\|} E^{(0)}_{0}(0) \\
&&\gamma_{\psi\psi,4}^{(0)} := E^{(0)}_{2}(0)\;,\quad \gamma_{\phi\psi,4}^{(0)} := \frac{1}{2} \partial^{2}_{\|\phi\|} E^{(0)}_{1}(0)\;,\quad \gamma_{\phi\phi,4}^{(0)} := \frac{1}{4!} \partial_{\|\phi\|}^{4} E^{(0)}_{0}(0)\;.\nonumber
\end{eqnarray}
By supersymmetry, see Corollary \ref{cor:cons}, Appendix \ref{app:SUSY}:
\begin{equation}\label{eq:smallfin0}
\gamma_{\psi,2}^{(0)} = \gamma_{\phi,2}^{(0)} =: \gamma^{(0)}_{2}\;,\qquad \gamma_{\psi\psi,4}^{(0)} = \frac{1}{2} \, \gamma_{\phi\psi,4}^{(0)} = \gamma_{\phi\phi,4}^{(0)} =: \gamma^{(0)}_{4}\;.
\end{equation}
This concludes the discussion of the small field regime.
\subsubsection{Large field regime}\label{sec:verylarge0}
Let $\phi \in L \mathbb{L}^{(0)}$. With respect to the small field region, here we have to face the extra difficulty that the terms  ``$\phi^{2}\zeta_{\phi}^{2}$'' in $\widetilde{V}_{\text{b}}^{(0)}$ might be large, recall Eq. (\ref{eq:tildeV0}). In the small field region, we could control these terms using the quartic term in $\zeta_{\phi}$, and the smallness of $\phi$. In the large field region, we shall exploit the sign of the real part of such terms. This is the content of the next inequality.

Let $0< \eps < 1/4$, and consider $\zeta_{\phi}$ such that $\| \text{Im}\, \zeta_{\phi} \| \leq \lambda^{-\frac{1}{4} + \eps}$. Then, for $\| \text{Im}\,( \phi/L \pm \zeta_{\phi}) \| \leq  \lambda^{-\frac{1}{4}}$:
\begin{equation}\label{eq:mixed0}
\big| e^{\widetilde{V}^{(0)}_{\text{b}}(\phi, \zeta_{\phi})} \big| \leq C_{L} e^{\frac{\lambda^{1 + 4\eps}}{L} \| \phi \|^{4} - \lambda \frac{L^{3}}{2} \|\zeta_{\phi}\|^{4}}\;.
\end{equation}
This holds as a special case of Proposition \ref{prp:mixed}, proven in Appendix \ref{app:osc}.
Therefore, we are in the position to apply Lemma \ref{lem:osc}. We rewrite:
\begin{eqnarray}
E^{(0)}_{n}(\phi) & = & \int d\mu_{\phi}(\zeta_{\phi})\, g^{(0)}_{n}(\phi, \zeta_{\phi})\nonumber\\
g^{(0)}_{n}(\phi, \zeta_{\phi}) & : = & L^{-2n} e^{-\widetilde V_{\text{b}}^{(0)}(\phi, \zeta_{\phi})} D^{(0)}_{n}(\phi/L, \zeta_{\phi}) \;.
\end{eqnarray}
The functions $g^{(0)}_{n}(\phi, \zeta_{\phi})$ are entire in $\zeta_{\phi}$ for all $\phi$. The bounds for the functions $D_{n}^{(0)}$ in (\ref{eq:propR}) imply the following (non optimal) estimates, for all $\phi$, $\zeta_{\phi}$ in $\mathbb{C}^{4}$, and for a universal constant $C>0$
\begin{eqnarray}\label{eq:D0nany}
|D^{(0)}_{n}(\phi/L, \zeta_{\phi})| &\leq& C L^{-2n} \lambda^{\frac{n}{2}} e^{C \lambda L^{4} \|\phi\|^{2} + C \lambda^{\frac{1}{2}}  \|\zeta_{\phi}\|^{2}} \nonumber\\
| D^{(0)}_{n}(\phi/L, 0) | &\leq& CL^{-2n} \lambda^{\frac{n}{2}} e^{C \lambda L \|\phi\|^{2}}\;.
\end{eqnarray}
Hence, for $\zeta_{\phi} \in \mathbb{C}^{4}$, $\|\text{Im}\, \zeta_{\phi}\| \leq \lambda^{-\frac{1}{4} + \eps}$, with $0<\eps < \frac{1}{4}$, the bound (\ref{eq:D0nany}) together with the bound (\ref{eq:mixed0}) imply:
\begin{eqnarray}\label{eq:gn0}
|g^{(0)}_{n}(\phi, \zeta_{\phi})| \leq K_{L} \lambda^{\frac{n}{2}} e^{C \lambda L^{4} \|\phi\|^{2} + C \lambda^{\frac{1}{2}}  \|\zeta_{\phi}\|^{2}} e^{\frac{\lambda^{1 + 4\eps}}{L} \| \phi \|^{4} - \frac{L^{3}}{2} \lambda \|\zeta_{\phi}\|^{4} }\;.
\end{eqnarray}
To estimate $E_{n}^{(0)}(\phi)$ efficiently, we shall perform a stationary phase expansion. We have:
\begin{equation}
E^{(0)}_{n}(\phi) = g^{(0)}_{n}(\phi, 0) + \mathcal{E}_{1}(g^{(0)}_{n}(\phi, \cdot))\;.
\end{equation}
From the second of (\ref{eq:D0nany}):
\begin{equation}\label{eq:large01}
|g^{(0)}_{n}(\phi, 0)|\leq CL^{-2n} \lambda^{\frac{n}{2}} e^{C \lambda L \|\phi\|^{2}}\;;
\end{equation}
instead, to estimate the remainder term from the stationary phase expansion, we use that, for $W = \lambda^{-\frac{1}{4} + \eps}$, thanks to the bound (\ref{eq:gn0}):
\begin{equation}
F_{W}(g^{(0)}_{n}(\phi, \cdot)) \leq C_{L} L^{-2n} \lambda^{\frac{n}{2}} e^{C \lambda L^{4} \|\phi\|^{2} + \frac{\lambda^{1 + 4\eps}}{L} \|\phi\|^{4}}\;.
\end{equation}
Hence, by Lemma \ref{lem:cau}:
\begin{equation}\label{eq:large02}
| \mathcal{E}_{1}( g^{(0)}_{n}(\phi,\cdot) ) | \leq \widetilde C_{L}L^{-2n} \lambda^{\frac{n}{2} + \frac{1}{2} - 4\varepsilon}  e^{C \lambda L^{4} \|\phi\|^{2} + \frac{\lambda^{1 + 4\eps}}{L} \|\phi\|^{4}}\;.
\end{equation}
Consider first $n=0$. Using that $\|\phi\| > L\lambda^{-1/4}$ we have, for any $\frac{1}{2} \leq \delta < 1$, for $\lambda$ small enough and $L$ large enough, from the bounds (\ref{eq:large01}), (\ref{eq:large02}):
\begin{equation}
| g^{(0)}_{0}(\phi, 0) | \leq \frac{\delta}{2} e^{\frac{\lambda}{8L} \|\phi\|^{4}}\;,\qquad | \mathcal{E}_{1}( g^{(0)}_{0}(\phi,\cdot) ) | \leq \frac{\delta}{2} e^{\frac{\lambda}{8L} \|\phi\|^{4}}\;.
\end{equation}
The first bound follows from the fact that the combination $\lambda L \|\phi\|^{2} - \frac{\lambda}{8 L} \|\phi\|^{4}$ can be arbitrarily negative, for $L$ large enough uniformly in $\phi \in L \mathbb{L}^{(0)}$. The second bound follows from the observation that, for $\lambda$ small enough and $L$ large enough uniformly in $\phi \in L \mathbb{L}^{(0)}$, the combination $\lambda L^{4} \|\phi\|^{2} + \frac{\lambda^{1 + 4\eps}}{L} \|\phi\|^{2} - \frac{\lambda}{8L} \|\phi\|^{4}$ is negative. Hence,
\begin{equation}\label{eq:bdE00large}
|E_{0}^{(0)}(\phi)| \leq \delta e^{\frac{\lambda}{8L} \|\phi\|^{4}}\;.
\end{equation}
Consider now $E_{n}^{(0)}$ for $n=1,2$. By the above reasoning, for $\lambda$ small enough and $L$ large enough:
\begin{eqnarray}
|g^{(0)}_{n}(\phi, 0)| &\leq& L^{-2n}\frac{\lambda^{\frac{n}{2}}}{2} e^{\frac{\lambda}{8L} \|\phi\|^{4}}\;,\\
| \mathcal{E}_{1}( g^{(0)}_{n}(\phi,\cdot) ) | &\leq& L^{-2n}\frac{\lambda^{\frac{n}{2}}}{2} e^{\frac{\lambda}{8L} \|\phi\|^{4}}\;,\qquad n=1,2\;.\nonumber
\end{eqnarray}
Therefore,
\begin{equation}\label{eq:bdE0nlarge}
| E^{(0)}_{n}(\phi) | \leq L^{-2n}\lambda^{\frac{n}{2}} e^{\frac{\lambda}{8L} \|\phi\|^{4}}\;, \qquad n = 1,2\;.
\end{equation}
This concludes the discussion of the large field regime.

\subsubsection{The effective potential on scale $h = 1$}\label{sec:scale1}
We obtained:
\begin{equation}
U^{(1)}(\Phi) = e^{-\frac{\lambda}{L} (\Phi\cdot \Phi)^{2} - i\mu L (\Phi\cdot \Phi)}\sum_{n=0,1,2} E^{(0)}_{n}(\phi) (\psi\cdot \psi)^{n}\;.
\end{equation}
where the functions $E_{n}^{(0)}(\phi)$ are analytic in $\phi \in L\mathbb{S}^{(0)} \cup L\mathbb{L}^{(0)}$. Moreover, they satisfy the bounds (\ref{eq:bdgen}), (\ref{eq:bdgen2}) for $\phi \in L\mathbb{S}^{(0)}$ and the bounds (\ref{eq:bdE00large}), (\ref{eq:bdE0nlarge}) for $\phi \in L\mathbb{L}^{(0)}$. Also, the renormalized functions $\mathcal{R} E^{(0)}_{n}(\phi)$ satisfy the bounds (\ref{eq:RE0}) for $\phi \in \mathbb{S}^{(1)}$. 

To conclude the discussion of the scale zero, we have to {\it renormalize} the coupling constant and the chemical potential, by taking into account the terms extracted with the localization procedure in Section \ref{sec:sfregime}. 
\medskip

\noindent{\it \underline{Small field bounds.}}  We rewrite:
\begin{eqnarray}\label{eq:ULUR}
U^{(1)}(\Phi) &=& e^{-\frac{\lambda}{L} (\Phi\cdot \Phi)^{2} - i\mu L (\Phi\cdot \Phi)} \sum_{n=0,1,2} \mathcal{R} E^{(0)}_{n}(\phi) (\psi\cdot \psi)^{n} \\
&& + e^{-\frac{\lambda}{L} (\Phi\cdot \Phi)^{2} - i\mu L (\Phi\cdot \Phi)} \sum_{n=0,1,2} ( \delta_{n,0} + \mathcal{L} E^{(0)}_{n}(\phi) ) (\psi \cdot \psi)^{n}\nonumber\\
&\equiv& U^{(1)}_{\mathcal{R}}(\Psi) + U^{(1)}_{\mathcal{L}}(\Psi)\;.
\end{eqnarray}
Consider first $U^{(1)}_{\mathcal{L}}$. We define:
\begin{eqnarray}\label{eq:rcc0}
&&\lambda_{1} := L^{-1} \lambda + \beta_{4}^{(0)}\;,\qquad \mu_{1} := \mu L + \beta^{(0)}_{2}\;,\nonumber\\
&&\beta^{(0)}_{2} := i \gamma_{2}^{(0)}\;,\qquad \beta^{(0)}_{4} := -\gamma_{4}^{(0)} - \frac{\gamma_{2}^{(0)2}}{2}\;,
\end{eqnarray}
where, by Eqs. (\ref{eq:Ebeta}), (\ref{eq:Ebeta2}):
\begin{equation}\label{eq:betabds}
| \beta^{(0)}_{2} | \leq CL \lambda\;,\qquad | \beta^{(0)}_{4} | \leq CL^{-1} \lambda^{3/2}\;.
\end{equation}
We then rewrite:
\begin{eqnarray}
&&U^{(1)}_{\mathcal{L}}(\Phi) = e^{- \lambda_{1} (\Phi\cdot \Phi)^{2} - i \mu_{1} (\Phi \cdot \Phi)}\nonumber\\
&& \quad \cdot e^{\beta^{(0)}_{4} (\Phi\cdot \Phi)^{2} + i\beta^{(0)}_{2} (\Phi\cdot \Phi)} \sum_{n=0,1,2} ( \delta_{n,0} + \mathcal{L} E^{(0)}_{n}(\phi) ) (\psi \cdot \psi)^{n}\nonumber\\
&& \equiv e^{- \lambda_{1} (\Phi\cdot \Phi)^{2} - i \mu_{1} (\Phi \cdot \Phi)} \widetilde{U}^{(1)}_{\mathcal{L}}(\Phi)\;.
\end{eqnarray}
The function $\widetilde{U}^{(1)}_{\mathcal{L}}(\Phi)$ can be expanded in powers of $(\psi\cdot \psi)$, in terms of suitable coefficients $U^{(1)}_{\mathcal{L}; n}(\phi)$:
\begin{equation}
\widetilde{U}^{(1)}_{\mathcal{L}}(\Phi) = \sum_{n=0,1,2} \widetilde{U}^{(1)}_{\mathcal{L};n}(\phi) (\psi\cdot \psi)^{n}\;. 
\end{equation}
Notice that, by construction, thanks to the definitions (\ref{eq:rcc0}), $\widetilde U^{(1)}_{\mathcal{L}}(0) = 1$ and moreover the function $\widetilde U^{(1)}_{\mathcal{L}}(\Phi) - 1$ has a Taylor expansion in $\Phi = 0$ that starts from order $6$. Consider now $U^{(1)}_{\mathcal{R}}$ in Eq. (\ref{eq:ULUR}). We rewrite it as:
\begin{eqnarray}
U^{(1)}_{\mathcal{R}}(\Phi) &=& e^{-\frac{\lambda}{L} (\Phi\cdot \Phi)^{2} - i\mu L (\Phi\cdot \Phi)} \sum_{n=0,1,2} \mathcal{R} E^{(0)}_{n}(\phi) (\psi\cdot \psi)^{n} \nonumber\\
&\equiv& e^{- \lambda_{1} (\Phi\cdot \Phi)^{2} - i \mu_{1} (\Phi \cdot \Phi)} \widetilde{U}^{(1)}_{\mathcal{R}}(\Phi)\;,
\end{eqnarray}
where $\widetilde{U}^{(1)}_{\mathcal{R}}(\Phi)$ can also be expanded in powers of $(\psi\cdot \psi)$, for suitable coefficients $\widetilde{U}^{(1)}_{\mathcal{R};n}(\phi)$:
\begin{equation}
\widetilde{U}^{(1)}_{\mathcal{R}}(\Phi) = \sum_{n=0,1,2} \widetilde{U}^{(1)}_{\mathcal{R};n}(\phi) (\psi\cdot \psi)^{n}\;.
\end{equation}
All in all:
\begin{equation}
U^{(1)}(\Phi) = e^{- \lambda_{1} (\Phi\cdot \Phi)^{2} - i \mu_{1} (\Phi \cdot \Phi)} \sum_{n=0,1,2} R_{n}^{(1)}(\phi) (\psi\cdot \psi)^{2}\;, 
\end{equation}
with 
\begin{equation}
R_{n}^{(1)}(\phi) := \widetilde{U}^{(1)}_{\mathcal{L};n}(\phi) + \widetilde{U}^{(1)}_{\mathcal{R};n}(\phi)\;.
\end{equation}
To conclude the integration of the scale zero, we shall estimate the coefficients $R_{n}^{(1)}(\phi)$. The functions $R^{(1)}_{n}(\phi)$ are analytic in $\mathbb{S}^{(1)}$, with:
\begin{equation}
\mathbb{S}^{(1)} = \{ \phi \in \mathbb{C}^{4} \mid \| \phi \| \leq \vert \lambda_{1} \vert^{-1/4} \} \subset L \mathbb{S}^{(0)}\;.
\end{equation}
Consider first $\widetilde{U}^{(1)}_{\mathcal{L};n}(\phi)$. By construction, this function has a Taylor series in $\phi = 0$ that starts from order $6-2n$. By inspection, and using the bounds (\ref{eq:betabds}), we have:
\begin{eqnarray}\label{eq:ULsmall}
&&| \widetilde{U}^{(1)}_{\mathcal{L};0}(\phi) -1 | \leq C_{L} \lambda^{\frac{5}{2}} \|\phi\|^{6}\;,\qquad | \widetilde{U}^{(1)}_{\mathcal{L};1}(\phi) | \leq C_{L} \lambda^{\frac{5}{2}} \|\phi\|^{4}\;,\nonumber\\
&&\qquad\qquad\qquad\qquad | \widetilde{U}^{(1)}_{\mathcal{L};2}(\phi) | \leq C_{L} |\lambda|^{3} \|\phi\|^{2}\;.
\end{eqnarray}
Consider now $\widetilde{U}^{(1)}_{\mathcal{R};n}(\phi)$. By inspection:
\begin{eqnarray}
| \widetilde{U}^{(1)}_{\mathcal{R};0}(\phi) | &\leq& C | \mathcal{R} E^{(0)}_{0}(\phi) | \\
| \widetilde{U}^{(1)}_{\mathcal{R};1}(\phi) | &\leq& C | \mathcal{R} E^{(0)}_{1}(\phi) | + C_{L}(\lambda + \lambda^{\frac{3}{2}}\|\phi\|^{2}) |\mathcal{R} E^{(0)}_{0}(\phi) |\nonumber\\ 
| \widetilde{U}^{(1)}_{\mathcal{R};2}(\phi) | &\leq& C |\mathcal{R} E^{(0)}_{2}(\phi)| + C_{L}(\lambda + \lambda^{\frac{3}{2}}\|\phi\|^{2})|\mathcal{R} E^{(0)}_{1}(\phi)|\nonumber\\
&& + C_{L}(\lambda^{\frac{3}{2}} + \lambda^{3} \|\phi\|^{4}) |\mathcal{R} E^{(0)}_{0}(\phi)|\;.\nonumber
\end{eqnarray}
The constant $C$ takes into account the fact that, for $\phi \in \mathbb{S}^{(1)}$ and for $\lambda$ small enough, there exists a universal constant $K$ such that $\big | e^{|\beta_{4}| \|\phi\|^{4} + |\beta_{2}| \|\phi\|^{2} } \big | \leq K$. In conclusion, from the bounds (\ref{eq:RE0}):
\begin{eqnarray}\label{eq:URsmall}
| \widetilde{U}^{(1)}_{\mathcal{R};0}(\phi) | &\leq& \widetilde C L^{-3} \lambda^{2} \|\phi\|^{6} \nonumber\\
| \widetilde{U}^{(1)}_{\mathcal{R};1}(\phi) | &\leq& \widetilde C L^{-3} \lambda^{2} \|\phi\|^{4} + \widetilde{C}_{L} \lambda^{3} \|\phi\|^{6} \nonumber\\&\leq& 2\widetilde{C} L^{-3} \lambda^{2} \|\phi\|^{4}\nonumber\\
| \widetilde{U}^{(1)}_{\mathcal{R};2}(\phi) | &\leq& \widetilde{C}L^{-3} \lambda^{2} \|\phi\|^{2} + \widetilde{C}_{L} \lambda^{3} \|\phi\|^{4} + \widetilde{C}_{L} \lambda^{\frac{7}{2}} \|\phi\|^{6}\nonumber\\
&\leq& 2 \widetilde{C} L^{-3} \lambda^{2} \|\phi\|^{2}\;.
\end{eqnarray}
Therefore, putting the bounds (\ref{eq:ULsmall}), (\ref{eq:URsmall}) together, we get, for $L$ large enough and for $\phi \in \mathbb{S}^{(1)}$:
\begin{eqnarray}\label{eq:small1}
&&|R^{(1)}_{2}(\phi)| \leq |\lambda_{1}|^{2}\|\phi\|^{2}\;,\qquad |R^{(1)}_{1}(\phi)|\leq |\lambda_{1}|^{2}  \|\phi\|^{4}\;,\nonumber\\&&
\quad\quad\quad\qquad | R^{(1)}_{0}(\phi) - 1 | \leq |\lambda_{1}|^{2} \|\phi\|^{6}\;.
\end{eqnarray}
\noindent{\it \underline{Large field bounds.}} We write again:
\begin{eqnarray}\label{eq:large11}
U^{(1)}(\Phi) &=& e^{- \lambda_{1} (\Phi\cdot \Phi)^{2} - i \mu_{1} (\Phi \cdot \Phi)}\nonumber\\
&& \quad \cdot e^{\beta^{(0)}_{4} (\Phi\cdot \Phi)^{2} + i\beta^{(0)}_{2} (\Phi\cdot \Phi)} \sum_{n=0,1,2} E_{n}^{(0)}(\phi) (\psi \cdot \psi)^{n}\nonumber\\
&\equiv& e^{- \lambda_{1} (\Phi\cdot \Phi)^{2} - i \mu_{1} (\Phi \cdot \Phi)} \sum_{n = 0,1,2,} R_{n}^{(1)}(\phi) (\psi \cdot \psi)^{n}
\end{eqnarray}
with:
\begin{eqnarray}\label{eq:Rexlarge}
R^{(1)}_{0}(\phi) &=& e^{\beta^{(0)}_{4} (\phi\cdot \phi)^{2} + i\beta^{(0)}_{2} (\phi\cdot \phi)} E_{0}^{(0)}(\phi) \\
R^{(1)}_{1}(\phi) &=& e^{\beta^{(0)}_{4} (\phi\cdot \phi)^{2} + i\beta^{(0)}_{2} (\phi\cdot \phi)} \Big( E_{1}^{(0)}(\phi) + E_{0}^{(0)}(\phi) ( i\beta^{(0)}_{2} + 2\beta^{(0)}_{4}(\phi \cdot \phi)) \Big) \nonumber\\
R^{(1)}_{2}(\phi) &=& e^{\beta^{(0)}_{4} (\phi\cdot \phi)^{2} + i\beta^{(0)}_{2} (\phi\cdot \phi)} \Big( E_{2}^{(0)}(\phi) + E_{1}^{(0)}(\phi) ( i\beta^{(0)}_{2} + 2\beta^{(0)}_{4}(\phi \cdot \phi))\nonumber\\&& + E_{0}^{(0)}(\phi) ( -\beta^{(0)2}_{2} + \beta^{(0)}_{4} + 4\beta^{(0)2}_{4}(\phi\cdot \phi)^{2} ) \Big)\;.\nonumber
\end{eqnarray}
The functions $R^{(1)}_{n}(\phi)$ are analytic in $\mathbb{L}^{(1)}$, with:
\begin{equation}
\mathbb{L}^{(1)} := \{ \phi \in \mathbb{C}^{4} \mid \|\phi\| \geq  \vert \lambda_{1} \vert^{-1/4}\;,\quad \|\text{Im}\,\phi\|\leq  \vert \lambda_{1} \vert^{-1/4} \}\;,
\end{equation}
and they satisfy the following bounds. Suppose first that $\phi$ is in the ``very large'' fields, $\phi \in \mathbb{L}^{(1)} \cap L\mathbb{L}^{(0)}$. Consider first $R^{(1)}_{0}(\phi)$. Recall the estimates (\ref{eq:betabds}) for $\beta^{(0)}_{4}$ and $\beta^{(0)}_{2}$, and (\ref{eq:bdE00large}) for $E_{0}^{(0)}(\phi)$. Using that, for $\lambda$ small enough,
\begin{equation}
C \frac{\lambda^{\frac{3}{2}}}{L} \|\phi\|^{4} + CL\lambda \|\phi\|^{2} + \frac{\lambda}{8L} \|\phi\|^{2} \leq \frac{\lambda}{6L} \|\phi\|^{4}\;,
\end{equation}
we have:
\begin{equation}
R^{(1)}_{0}(\phi) \leq \delta e^{\frac{\lambda}{6L} \|\phi\|^{4}} \leq \delta e^{c_{1} |\lambda_{1}| \|\phi\|^{4}}\;,
\end{equation}
with $c_{1} = \frac{1}{6} + |\lambda|^{\frac{1}{2}}$. Consider $R^{(1)}_{1}(\phi)$. Proceeding as for $R^{(1)}_{0}(\phi)$, and using also the bound (\ref{eq:bdE0nlarge}) for $E_{1}^{(0)}(\phi)$:
\begin{eqnarray}
|R^{(1)}_{1}(\phi)| &\leq& e^{C\frac{\lambda^{\frac{3}{2}}}{L} \|\phi\|^{4} + C\lambda L \|\phi\|^{2}}\Big( |E_{1}^{(0)}(\phi)| + |E_{0}^{(0)}(\phi)| ( C L \lambda + C L^{-1} \lambda^{\frac{3}{2}} \|\phi\|^{2} ) \Big)\nonumber\\
&\leq& e^{C\frac{\lambda^{\frac{3}{2}}}{L} \|\phi\|^{4} + C\lambda L \|\phi\|^{2} + \frac{\lambda}{8 L} \|\phi\|^{4}} \Big( \frac{\lambda^{\frac{1}{2}}}{L^{2}} + \delta ( C L \lambda + C L^{-1} \lambda^{\frac{3}{2}} \|\phi\|^{2} )\Big)\nonumber\\
&\leq& e^{C\frac{\lambda^{\frac{3}{2}}}{L} \|\phi\|^{4} + C\lambda L \|\phi\|^{2} + \frac{\lambda}{8 L} \|\phi\|^{4}} \Big( \frac{\lambda^{\frac{1}{2}}}{L^{2}} + \delta C L \lambda \Big)\Big( 1 + K \frac{\lambda^{\frac{1}{2}}}{L^{2}} \|\phi\|^{2} \Big)\nonumber\\
&\leq& e^{C\frac{\lambda^{\frac{3}{2}}}{L} \|\phi\|^{4} + C\lambda L \|\phi\|^{2} + \frac{\lambda}{8 L} \|\phi\|^{4} + C\frac{\lambda^{\frac{1}{2}}}{L^{2}} \|\phi\|^{2}} \Big( \frac{\lambda^{\frac{1}{2}}}{L^{2}} + \delta C L \lambda \Big)\;.
\end{eqnarray}
For $\phi \in \mathbb{L}^{(1)} \cap L\mathbb{L}^{(0)}$, choosing $L$ large enough and $\lambda$ small enough, the argument of the exponential is bounded by $\frac{\lambda}{6 L} \|\phi\|^{4} \leq c_{1} |\lambda_{1}| \|\phi\|^{4}$. At the same time, the argument of the last parenthesis is bounded by $( \lambda / 2L )^{1/2} \leq |\lambda_{1}|^{1/2}$. Therefore:
\begin{equation}
| R^{(1)}_{1}(\phi) | \leq |\lambda_{1}|^{\frac{1}{2}} e^{c_{1} |\lambda_{1}| \|\phi\|^{4}}\;.
\end{equation} 
Finally, consider $R^{(1)}_{2}(\phi)$. Proceeding as for $R^{(1)}_{1}(\phi)$, we get:
\begin{eqnarray}
&&|R^{(1)}_{2}(\phi)| \leq e^{C\frac{\lambda^{\frac{3}{2}}}{L} \|\phi\|^{4} + C\lambda L \|\phi\|^{2}} \\
&&\quad \cdot \Big( |E_{2}^{(0)}(\phi)| + |E_{1}^{(0)}(\phi)| \big( CL\lambda + CL^{-1} \lambda^{\frac{3}{2}}\|\phi\|^{2}  \big)\nonumber\\&& \quad + |E_{0}^{(0)}(\phi)| \big( C L^{2} \lambda^{2} + CL^{-1} \lambda^{\frac{3}{2}} + CL^{-2} \lambda^{3}\|\phi\|^{4} \big) \Big)\;.\nonumber
\end{eqnarray}
Using that, recalling the bounds (\ref{eq:bdE0nlarge}) for $E^{(0)}_{n}(\phi)$, $n=1,2$:
\begin{eqnarray}
|E_{1}^{(0)}(\phi)| \big( CL\lambda + CL^{-1} \lambda^{\frac{3}{2}}\|\phi\|^{2}  \big) &\leq& |E_{1}^{(0)}(\phi)| CL \lambda (1 + KL^{-2} \lambda^{\frac{1}{2}}\|\phi\|^{2})\nonumber\\
&\leq& C L^{-2} \lambda^{\frac{3}{2}} e^{\frac{\lambda}{8L} \|\phi\|^{4} + K \frac{\lambda^{\frac{1}{2}}}{L^{2}} \|\phi\|^{2}}\;,
\end{eqnarray}
and that, recalling the bound (\ref{eq:bdE00large}) for $E_{0}^{(0)}(\phi)$, for $\lambda$ small enough:
\begin{eqnarray}
&&|E_{0}^{(0)}(\phi)| \big( C L^{2} \lambda^{2} + CL^{-1} \lambda^{\frac{3}{2}} + CL^{-2} \lambda^{3}\|\phi\|^{4} \big)\nonumber\\
&&\qquad \leq |E^{(0)}_{0}(\phi)|2C L^{-1} \lambda^{\frac{3}{2}} ( 1 + K L^{-1} \lambda^{\frac{3}{2}} \|\phi\|^{4} )\nonumber\\
&&\qquad \leq  2C L^{-1} \lambda^{\frac{3}{2}} \delta e^{\frac{\lambda}{8L} \|\phi\|^{4} + K \frac{\lambda^{\frac{3}{2}}}{L} \|\phi\|^{4}}\;. 
\end{eqnarray}
Putting everything together, taking $\lambda$ small enough and $L$ large enough, for $\phi \in \mathbb{L}^{(1)} \cap L\mathbb{L}^{(0)}$ and for a suitable universal constant $C>0$ we get:
\begin{eqnarray}
|R^{(1)}_{2}(\phi)| &\leq& e^{C\frac{\lambda^{\frac{3}{2}}}{L} \|\phi\|^{4} + C\frac{\lambda^{\frac{1}{2}}}{L^{2}} \|\phi\|^{2} + \frac{\lambda}{8 L} \|\phi\|^{4}}\Big( L^{-4} \lambda + L^{-2} \lambda^{\frac{3}{2}} + C L^{-1} \lambda^{\frac{3}{2}} \Big)\nonumber\\
&\leq& |\lambda_{1}| e^{c_{1} |\lambda_{1}| \|\phi\|^{4}}\;.
\end{eqnarray}
This concludes the discussion of the $R^{(1)}_{n}(\phi)$ coefficients for the ``very large'' field region $\phi \in \mathbb{L}^{(1)} \cap L \mathbb{L}^{(0)}$. Consider the ``moderately large'' field region $\phi \in \mathbb{L}^{1} \cap L \mathbb{S}^{(0)}$. Here we use again the expressions (\ref{eq:Rexlarge}) for the $R_{n}^{(0)}$ coefficients, together with the nonrenormalized bounds (\ref{eq:bdgen}), (\ref{eq:bdgen2}) for $E^{(0)}_{n}(\phi)$, $n=0,1,2$.

Recall that for $\phi \in \mathbb{L}^{1} \cap L \mathbb{S}^{(0)}$, we can use the bound $\|\phi\| \leq L\lambda^{-\frac{1}{4}}$. In particular,
\begin{eqnarray}\label{eq:exp-1}
\Big| e^{\beta^{(0)}_{4} (\phi\cdot \phi)^{2} + i\beta^{(0)}_{2} (\phi\cdot \phi)} - 1 \Big| &\leq& C |\beta^{(0)}_{4}|\|\phi\|^{4}  + C |\beta^{(0)}_{2}|\|\phi\|^{2}\nonumber\\
&\leq& K L^{3} \lambda^{\frac{1}{2}}\;.
\end{eqnarray}
Consider first $R_{0}^{(1)}(\phi)$. We get:
\begin{eqnarray}
R_{0}^{(1)}(\phi)  &=& e^{\beta^{(0)}_{4} (\phi\cdot \phi)^{2} + i\beta^{(0)}_{2} (\phi\cdot \phi)} (E_{0}^{(0)}(\phi) - 1 + 1) \\
&=& e^{\beta^{(0)}_{4} (\phi\cdot \phi)^{2} + i\beta^{(0)}_{2} (\phi\cdot \phi)} + e^{\beta^{(0)}_{4} (\phi\cdot \phi)^{2} + i\beta^{(0)}_{2} (\phi\cdot \phi)} (E_{0}^{(0)}(\phi) - 1)\nonumber\\
&=& 1 + \Big( e^{\beta^{(0)}_{4} (\phi\cdot \phi)^{2} + i\beta^{(0)}_{2} (\phi\cdot \phi)} - 1 \Big) + e^{\beta^{(0)}_{4} (\phi\cdot \phi)^{2} + i\beta^{(0)}_{2} (\phi\cdot \phi)} (E_{0}^{(0)}(\phi) - 1)\;;\nonumber
\end{eqnarray}
therefore, using (\ref{eq:exp-1}) and (\ref{eq:bdgen})
\begin{equation}
|R_{0}^{(1)}(\phi) - 1| \leq 2K L^{3} \lambda^{\frac{1}{2}} 
\end{equation}
Then, choosing $\lambda$ small enough:
\begin{eqnarray}\label{eq:R01}
R_{0}^{(1)}(\phi) &\leq& e^{2 K L^{3} \lambda^{\frac{1}{2}}} = e^{2 K L^{3} \lambda^{\frac{1}{2}} - c_{1} |\lambda_{1}| \|\phi\|^{4}} e^{c_{1} |\lambda_{1}| \|\phi\|^{4}} \nonumber\\
&\leq& \delta e^{c_{1} |\lambda_{1}| \|\phi\|^{4}}\;,\qquad \delta := e^{-\frac{1}{8}} < 1\;.
\end{eqnarray}
In the last step we used that $\|\phi\| \geq |\lambda_{1}|^{-\frac{1}{4}}$, since $\phi \in \mathbb{L}^{(1)}$. Next, consider $R^{(1)}_{1}(\phi)$. We get:
\begin{eqnarray}
|R^{(1)}_{1}(\phi)| &\leq& C \big( |E^{(0)}_{1}(\phi)| + |E^{(0)}_{0}(\phi)| ( CL \lambda + C L^{-1} \lambda^{\frac{3}{2}} \|\phi\|^{2} ) \big)\nonumber\\
&\leq& \widetilde{C} L \lambda + 2C ( CL \lambda + C L^{-1} \lambda^{\frac{3}{2}} \|\phi\|^{2} ) \nonumber\\
&=& 2\widetilde{C} L^{2} |\lambda_{1}| + 4C ( CL^{2} |\lambda_{1}| + C L^{\frac{1}{2}} |\lambda_{1}|^{\frac{3}{2}} \|\phi\|^{2} )\nonumber\\
&\leq& |\lambda_{1}|^{\frac{1}{2}}\;,
\end{eqnarray}
where in the last step we used that $\| \phi \| \leq L \lambda^{-\frac{1}{4}}$, and we took $\lambda$ small enough. Finally, consider $R_{2}^{(1)}(\phi)$. We get:
\begin{eqnarray}
&&|R_{2}^{(1)}(\phi)| \leq C\Big( |E_{2}^{(0)}(\phi)| + |E_{1}^{(0)}(\phi)| ( CL \lambda + CL^{-1} \lambda^{\frac{3}{2}}\|\phi\|^{2})\nonumber\\&&\qquad + |E_{0}^{(0)}(\phi)| ( CL^{2} \lambda^{2} + CL^{-1} \lambda^{\frac{3}{2}} + CL^{-2} \lambda^{3} \|\phi\|^{4}) \Big) \nonumber\\
&&\leq C \Big( K L^{-1} \lambda^{\frac{3}{2}} + K L \lambda ( CL \lambda + CL^{-1} \lambda^{\frac{3}{2}}\|\phi\|^{2})\nonumber\\
&&\qquad + 2 ( CL^{2} \lambda^{2} + CL^{-1} \lambda^{\frac{3}{2}} + CL^{-2} \lambda^{3} \|\phi\|^{4})   \Big)\nonumber\\
&&\leq |\lambda_{1}|\;,
\end{eqnarray}
where in the last step we used again that $\| \phi \| \leq L \lambda^{-\frac{1}{4}}$, and we chose $\lambda$ small enough.

Let us summarize the large field analysis. The $R^{(1)}_{n}(\phi)$ functions are analytic for $\phi \in \mathbb{L}^{(1)}$, and for those values of $\phi$ they satisfy the bounds:
\begin{equation}\label{eq:large12}
|R^{(1)}_{0}(\phi)| \leq \delta e^{c_{1} \vert \lambda_{1} \vert \|\phi\|^{4}}\;,\qquad |R^{(1)}_{n}(\phi)| \leq \vert \lambda_{1} \vert^{\frac{n}{2}} e^{c_{1} \vert \lambda_{1} \vert \|\phi\|^{4}}\qquad n=1,2\;,
\end{equation}
for $e^{-\frac{1}{8}}\leq \delta < 1$.
These bounds conclude the discussion of the integration of the scale zero.
\subsection{General integration step} \label{subsec: General_integration_step}

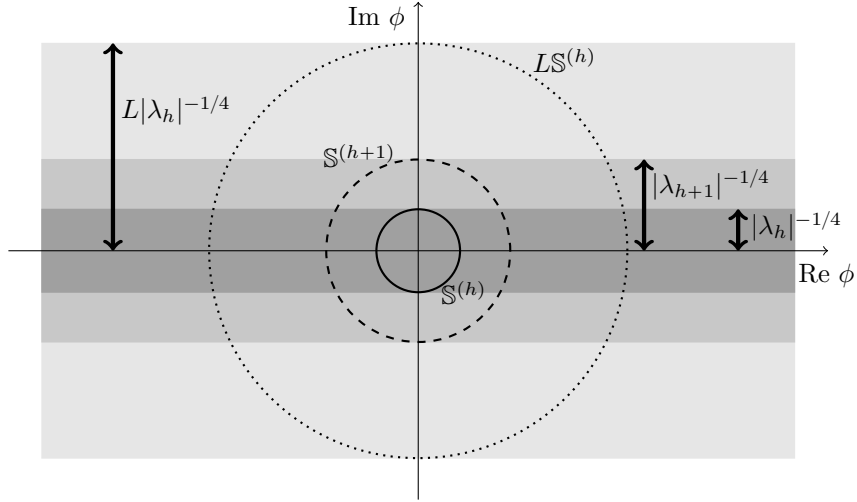
\begin{figure}[h]
\begin{tikzpicture}[scale=.55]
    \draw[fill, gray1] (-9, -5) rectangle (9, 5); 
    \draw[fill, gray2] (-9, -2.2) rectangle (9, 2.2); 
    \draw[fill, gray3] (-9, -1) rectangle (9, 1);
    \draw[thick, dotted] (0, 0) circle (5);
    \draw[thick,dashed] (0, 0) circle (2.2);
    \draw[thick] (0, 0) circle (1);
    \draw[ultra thick, <->] (-7.3,0)-- (-7.3, 5);
		\node[xshift=5mm] at (-6.7, 3.4) {$L|\lambda_{h}|^{-1/4}$};
	\draw[ultra thick, <->] (5.4,0)-- (5.4, 2.2);
		\node at (7, 1.6) {$|\lambda_{h+1}|^{-1/4}$};
	\draw[ultra thick, <->] (7.65,0)-- (7.65, 1);
		\node at (9.05, 0.6) {$|\lambda_{h}|^{-1/4}$} ;
	\draw[ ->] (0, -6) -- (0, 6);
    \draw[ ->](-9.8, 0) -- (9.8, 0);
    	\node[xshift=2mm] at (9.4, -0.6) {Re $\phi$};
    	\node[xshift=-2mm] at (-9.4, -0.6) {\textcolor{white}{Re $\phi$}};
		\node at (-1, 5.6) {Im $\phi$};
	\node at (3.5,4.6) {$L \mathbb{S}^{(h)}$};
	\node at (-1.4,2.3) {$\mathbb{S}^{(h+1)}$};
	\node at (1.1,-1) {$ \mathbb{S}^{(h)}$};
  \end{tikzpicture}
\caption{Schematic picture of the relevant analyticity domains and the small field sets. The strips represents the analyticity domains $\mathbb{S}^{(h)} \cup \mathbb{L}^{(h)}$, $\mathbb{S}^{(h+1)} \cup \mathbb{L}^{(h+1)}$, $L(\mathbb{S}^{(h)} \cup \mathbb{L}^{(h)})$ respectively from darker to brighter gray. The dotted, dashed and thick circles represent respectively the boundary of the small field sets $L\mathbb{S}^{(h)}$, $\mathbb{S}^{(h+1)}$ and $\mathbb{S}^{(h)}$. 
}\label{fig:analyticity}
\end{figure}
We are now ready to perform the integration of the general scale $h\geq 0$. We shall show inductively that the effective potentials satisfy certain properties and bounds that allow to iterate the map. These properties and bounds are the content of the next theorem.

\begin{theorem}[Effective potential flow]\label{thm:effective_potential_flow} Under the same assumptions of Theorem \ref{thm:main} the following is true. Let $C>0$, $0< \varepsilon < \frac{1}{4}$ and $e^{-\frac{1}{8}} \leq \delta <1$. Then, for $\lambda$ small enough, there exists a unique choice $\mu = \mu(\lambda) \in \mathbb{C}$, with $|\mu(\lambda)| \leq C\lambda$, such that for any $N\in \mathbb{N}$ and for any $0\leq h \leq N$ the effective potential $U^{(h)}(\Phi)$ can be written as:
\begin{equation}\label{eq:U_h}
U^{(h)}(\Phi) = e^{-\lambda_{h} (\Phi\cdot \Phi)^{2} - i\mu_{h} (\Phi\cdot \Phi)} \sum_{n=0,1,2} R^{(h)}_{n}(\phi) (\psi\cdot \psi)^{n}
\end{equation}
where:
\begin{equation}\label{eq:lambdah}
|L^h \lambda_h - \lambda | \leq C \lambda^{3/2}\;,\qquad |\mu_{h}|\leq C|\lambda_{h}|,\qquad \lambda_h,\,\mu_h \in \mathbb{C}\;,
\end{equation}
and where $R^{(h)}_{n}(\phi)$ are analytic functions in $\phi\in \mathbb{S}^{(h)}\cup \mathbb{L}^{(h)}$; the small field set $\mathbb{S}^{(h)}$ and the large field set $\mathbb{L}^{(h)}$ are defined as:
\begin{eqnarray}\label{eq:SLh}
\mathbb{S}^{(h)} &:=& \{ \phi \in \mathbb{C}^{4} \mid \|\phi\|\leq \vert \lambda_{h} \vert^{-1/4 } \}  \nonumber\\
\mathbb{L}^{(h)} &:=& \{ \phi \in \mathbb{C}^{4} \mid \|\phi\| > \vert \lambda_{h} \vert^{-1/4}\;,\quad \|\mathrm{Im}\,\phi\|\leq  \vert \lambda_{h} \vert^{-1/4} \}\;.
\end{eqnarray}
The functions $R^{(h)}_{n}(\phi)$ are radial for $\phi\in \mathbb{R}^{4}$, $R^{(h)}_{n}(\phi) \equiv R^{(h)}_{n}(\|\phi\|)$, and $R^{(h)}_{0}(0) = 1$. They satisfy the following bounds, for a universal constant $C>0$. 
\medskip

\noindent{\underline{\it Small field bounds.}} Let $\phi \in \mathbb{S}^{(h)}$. Then:
\begin{eqnarray}\label{eq:small}
&&|R^{(h)}_{2}(\phi)| \leq \vert \lambda_{h} \vert^{2}\|\phi\|^{2}\;,\qquad |R^{(h)}_{1}(\phi)|\leq \vert \lambda_{h} \vert^{2} \|\phi\|^{4}\nonumber\\
&&\quad\quad\qquad\quad | R^{(h)}_{0}(\phi) - 1 | \leq \vert \lambda_{h} \vert^{2} \|\phi\|^{6}\;.
\end{eqnarray}
\noindent{\underline{\it Large field bounds.}} Let $\phi \in \mathbb{L}^{(h)}$. Then:
\begin{equation}\label{eq:large}
|R^{(h)}_{0}(\phi)| \leq \delta e^{ c_{h} \vert \lambda_{h} \vert \|\phi\|^{4}}\;,\qquad |R^{(h)}_{n}(\phi)| \leq \vert \lambda_{h} \vert^{\frac{n}{2}} e^{c_h \vert \lambda_{h} \vert \|\phi\|^{4}}\qquad n=1,2\;,
\end{equation}
with $c_0 = 0$ and, for a universal constant $K>0$:
%
\begin{equation}\label{eq:indh}
c_{h} = \frac{1}{6} + \sum_{k= 0}^{h-1} K  \vert \lambda_{k} \vert^{\varepsilon}\;,\qquad \text{for $h>0$.}
\end{equation}
\end{theorem}
All the statements in the theorem are trivially true on scale $h=0$, and have been checked on scale $h=1$ in Section \ref{sec:scale1}. The goal of this section is to prove by induction that they propagate to scale $h+1$.  To be more precise, we shall prove all but the statement concerning the existence of $\mu(\lambda)$ and the bounds $|\mu _{h} | \leq C |\lambda_{h}|$: this part of the proof is postponed to Appendix~\ref{app:mu}.
\subsubsection{Setting up the integration}\label{sec:seth}
Recall the flow of the effective potentials, defined in Eq.~\eqref{eq:local_map_U}:
\begin{equation}\label{eq:Uh+1}
U^{(h+1)}(\Phi) = \frac{1}{N^{(h)}} \int d\mu(\zeta)\, [U^{(h)}(\Phi/L + \zeta) U^{(h)}(\Phi/L -\zeta)]^{\frac{L^{3}}{2}}\;,
\end{equation}
with normalization:
\begin{equation}
N^{(h)} = \int d\mu(\zeta)\, [U^{(h)}(\zeta) U^{(h)}(-\zeta)]^{\frac{L^{3}}{2}}\;.
\end{equation}
By the localization theorem, Theorem \ref{prp:SUSY}, $N^{(h)} = 1$, see Remark \ref{rem:Nh}. Let us consider the integrand in Eq. (\ref{eq:Uh+1}), with $U^{(h)}(\Phi)$ given as in the inductive assumptions, see \eqref{eq:U_h}-\eqref{eq:indh}. We introduce the notation:
\begin{equation}
f^{(h)}(\Phi) := \sum_{n=0,1,2} R_{n}^{(h)}(\phi) (\psi\cdot \psi)^{n}\;.
\end{equation}
We have:
\begin{equation}\label{eq:UU}
[U^{(h)}(\Phi/L + \zeta) U^{(h)}(\Phi/L -\zeta)]^{\frac{L^{3}}{2}} =  e^{-\widecheck{V}^{(h)}(\Phi/L, \zeta)} [f^{(h)}(\Phi/L + \zeta) f^{(h)}(\Phi/L - \zeta)]^{\frac{L^{3}}{2}}\;,
\end{equation}
%
where:
\begin{eqnarray}\label{eq:widecheckV}
\widecheck{V}^{(h)}(\Phi/L, \zeta) &:=& \frac{\lambda_{h}}{L} (\Phi \cdot \Phi)^{2} + \lambda_{h} L^{3} (\zeta\cdot \zeta)^{2} + 4\lambda_{h} L ( \Phi\cdot \zeta )^{2}\\&& + 2\lambda_{h} L (\Phi\cdot \Phi) (\zeta\cdot \zeta)  + i\mu_{h} L (\Phi\cdot \Phi) + i\mu_{h} L (\zeta\cdot \zeta)\;.\nonumber
\end{eqnarray}
We set:
\begin{equation}
[f^{(h)}(\Phi/L + \zeta) f^{(h)}(\Phi/L - \zeta)]^{\frac{L^{3}}{2}} = \sum_{\underline{a}, \underline{b}} B^{(h)}_{\underline{a},\underline{b}}(\phi/L, \zeta_{\phi}) L^{-|\underline{a}|}\psi^{\underline{a}}\zeta_{\psi}^{\underline{b}}\;;
\end{equation}
our first task will be to derive bounds for the functions $B^{(h)}_{\underline{a},\underline{b}}(\phi/L, \zeta_{\phi})$. To this end, the next lemma will be useful.
\begin{lemma}\label{lem:kappaNexp} 
\begin{itemize}
\item[(i)] Suppose $f(\psi)$ satisfies $(\kappa,\mathcal N)$-bounds. Then, $f(\psi \pm  \zeta_{\psi})$ satisfies $(\kappa,\mathcal{N}, \mathcal{N})$-bounds.
\item[(ii)] Suppose $f(\psi,\zeta_{\psi})$ satisfies $(\kappa,\mathcal{N},\mathcal{M})$-bounds and let $p \in \mathbb{N}$. Then, the function $\left (1+f(\psi,\zeta_{\psi})\right )^{p} -1$ satisfies $(K p\kappa,\mathcal{N},\mathcal{M})$-bounds for some constant $K$ depending on $\kappa p$ only.
\end{itemize}
\end{lemma}
\begin{proof} 

\noindent{\underline{\it Proof of item $(i)$}.} For simplicity, consider $h_{+}(\psi, \zeta_{\psi}) = f(\psi + \zeta_{\psi})$. Set: $h_{+}(\psi,\zeta_{\psi}) = \sum_{\underline{a},\underline{b}} \, h_{\underline{a},\underline{b}} \, \psi^{\underline{a}} \, \zeta_{\psi}^{\underline{b}}$. The claim follows by noticing that 
\begin{equation}
(\psi + \zeta_{\psi})^{\underline{a}} = \sum _{\substack{\underline{a'},\underline{b'} \\ \underline{a'} + \underline{b'} = \underline{a}}} \mathrm{sign}(\underline{a'},\underline{b'}) \, \psi^{\underline{a'}} \, \zeta_{\psi}^{\underline{b'}}
\end{equation}
where $\mathrm{sign}(\underline{a'},\underline{b'}) \in \{-1,1\}$ is left unspecified. Indeed
\begin{equation}
\begin{split}
h_{+}(\psi,\zeta_{\psi}) = \sum_{\underline{a}} f_{\underline{a}}\, (\psi + \zeta_{\psi})^{\underline{a}} = \sum_{\underline{a'},\underline{b'}} \mathrm{sign}(\underline{a'},\underline{b'}) f_{\underline{a'}+\underline{b'}} \, \psi^{\underline{a'}} \, \zeta_{\psi}^{\underline{b'}}\;,
\end{split}
\end{equation}
and hence, $|h_{\underline{a},\underline{b}}| = |f_{\underline{a}+\underline{b}} | \leq \kappa \mathcal{N}^{|\underline{a}|} \mathcal{N}^{|\underline{b}|}$.
\medskip

\noindent{\underline{\it Proof of item $(ii)$}.} By Lemma \ref{lem:kappaNprod}, it follows that the function $(f(\psi,\zeta_{\psi}))^{p}$ satisfies $(\kappa^{p},p\mathcal{N},p\mathcal{M})$-bounds.
For simplicity, denote by $f^{(p)}_{\underline{a},\underline{b}}$ the coefficients of $(f(\psi,\zeta_{\psi}))^{p}$. Setting $h(\psi,\zeta_{\psi}):= (1 + f(\psi,\zeta_{\psi}))^{p}-1$, with coefficients $h_{\underline{a},\underline{b}}$, we notice that $h(\psi,\zeta_{\psi}) = \sum_{i =1}^{p} \, \left ( {p \atop i} \right ) (f(\psi,\zeta_{\psi}))^{p}$. Therefore:
\begin{equation}
|h_{\underline{a},\underline{b}} |= \left |\sum_{i=1}^{p} \, \left ( {p \atop i} \right ) f^{(p)}_{\underline{a},\underline{b}}\right |
\leq  p \kappa \mathcal{N}^{|\underline{a}|}\mathcal{M}^{|\underline{b}|} \sum_{i = 1} ^{\infty} \frac{i^{8}}{i!}  (p \kappa)^{i-1} \;.
\end{equation}
We denote the series by $K$ and notice that $K \leq e^{3} \exp(e^{3}p\kappa)$.
\end{proof}
Lemma \ref{lem:kappaNexp} will be used to prove the following statements on the $B_{\underline{a}, \underline{b}}^{(h)}$ functions.
\begin{proposition} \label{prop:B} Under the same assumptions of Theorem \ref{thm:effective_potential_flow}, the following is true. The functions $B^{(h)}_{\underline{a}, \underline{b}}(\phi/L,\zeta_{\phi})$  are analytic in $\phi\in \mathbb{C}^{4}$ and in $\zeta_{\phi} \in \mathbb{C}^{4}$, provided $\phi/L \pm \zeta_{\phi} \in \mathbb{S}^{(h)} \cup \mathbb{L}^{(h)}$. Moreover, $B^{(h)}_{\underline{0},\underline{0}}(0,0) = 1$. Also, there exists a universal constant $\widetilde{C}>0$ such that for $L$ large enough and for $\lambda$ small enough the following is true.
\begin{itemize}
\item[(i)] Let $\phi/L \pm \zeta_{\phi} \in \mathbb{S}^{(h)}$. Then:
\begin{eqnarray}\label{eq:B_00_1}
| B^{(h)}_{\underline{0}, \underline{0}}(\phi/L,\zeta_{\phi}) - 1| &\leq & 2 \,  \vert \lambda_{h} \vert^{\frac{1}{2}} L^{3}\;,\nonumber\\
| B^{(h)}_{\underline{a}, \underline{b}}(\phi/L,\zeta_{\phi}) | &\leq & \widetilde{C} \, \vert \lambda_{h} \vert^{\frac{1}{2} + \frac{|\underline{a}| + |\underline{b}|}{4}} L^{3} \qquad \text{for $|\underline{a}| + |\underline{b}| >0$.}\nonumber
\end{eqnarray}
\item[(ii)] Let $\phi/L + \zeta_{\phi} \in \mathbb{S}^{(h)}$ and $\phi/L - \zeta_{\phi} \in \mathbb{L}^{(h)}$. Then, for $\delta$ as in the assumptions of Theorem \ref{thm:effective_potential_flow}, see also Eq. (\ref{eq:large}):
\begin{equation}
|B^{(h)}_{\underline{a}, \underline{b}}(\phi/L,\zeta_{\phi})| \leq 4\delta^{\frac{L^{3}}{2} - \frac{|\underline{a}| + |\underline{b}|}{2}} L^{3(|\underline{a}| + |\underline{b}|)} \vert \lambda_{h} \vert^{\frac{|\underline{a}| + |\underline{b}|}{4}} e^{\frac{L^{3}}{2} c_{h}  \vert \lambda_{h} \vert \| \phi/L - \zeta_{\phi}\|^{4}}\;.
\end{equation}
%
%
\item[(iii)] Let $\phi/L \pm \zeta_{\phi} \in \mathbb{L}^{(h)}$. Then:
\begin{eqnarray}
&&|B^{(h)}_{\underline{a},\underline{b}}(\phi/L,\zeta_{\phi})| \nonumber\\
&&\leq 4 \delta^{L^{3} - \frac{|\underline{a}| + |\underline{b}|}{2}} L^{3(|\underline{a}| + |\underline{b}|)} \vert \lambda_{h} \vert^{\frac{|\underline{a}| + |\underline{b}|}{4}} e^{\frac{L^{3}}{2} c_{h}\vert \lambda_{h} \vert \| \phi/L - \zeta_{\phi}\|^{4}+\frac{L^{3}}{2} \vert \lambda_{h} \vert \| \phi/L + \zeta_{\phi}\|^{4}}\;.\nonumber
\end{eqnarray}
\end{itemize}
\end{proposition}
\begin{proof} 
To begin, notice that 
\begin{equation}
B_{\underline{0},\underline{0}}^{(h)}(\phi/L,\zeta_{\phi}) = \left ( R_{0}^{(h)}(\phi/L + \zeta_{\phi}) R_{0}^{(h)}(\phi/L - \zeta_{\phi})\right )^{\frac{L^{3}}{2}}
\end{equation}
and therefore $R_{0}^{(h)}(0) = 1$ implies $B_{\underline{0},\underline{0}}^{(h)}(0,0) = 1$. Analyticity of the functions $B^{(h)}_{\underline{a}, \underline{b}}(\phi/L,\zeta_{\phi})$ is an obvious consequence of the analyticity of $R_{n}^{(h)}(\phi)$. Next, we set:
\begin{eqnarray}\label{eq:gphi}
g^{(h)}(\phi/L,\zeta) &:=& [f^{(h)}(\Phi/L + \zeta) f^{(h)}(\Phi/L - \zeta)]^{\frac{L^{3}}{2}} \\
&=& \sum_{\underline{a}, \underline{b}} B^{(h)}_{\underline{a}, \underline{b}}(\phi/L,\zeta_{\phi}) L^{-|\underline{a}|}\psi^{\underline{a}} \zeta_{\psi}^{\underline{b}}\;.\nonumber
\end{eqnarray}
We also notice that for $\phi/L \in \mathbb{S}^{(h)}$ the inductive estimates \eqref{eq:small} imply that the functions $ f^{(h)}(\Phi/L)$ and $f^{(h)}(\Phi/L)-1$ satisfy respectively $(\kappa, \mathcal{N})$ and $(\kappa', \mathcal{N}')$ bounds, with:
\begin{equation}\label{eq:kappaNf}
(\kappa, \mathcal{N}) = \big((1+ |\lambda_{h}|^{\frac{1}{2}}), L^{-1}|\lambda_{h}|^{\frac{3}{8}}\big)\;,\qquad \big(\kappa', \mathcal{N}') = (|\lambda_{h}|^{\frac{1}{2}}, L^{-1}|\lambda_{h}|^{\frac{1}{4}}\big)\;.
\end{equation}
On the other hand, for $\phi \in \mathbb{L}^{(h)}$, the inductive estimates \eqref{eq:large} imply that $f(\Phi/L)$ satisfies $(\kappa, \mathcal{N})$ bounds with:
\begin{equation}\label{eq:kappaNflarge}
(\kappa, \mathcal{N}) = \Big(\delta e^{c_{h}|\lambda_{h}| \| \phi\|^{4}}, L^{-1}\delta^{-\frac{1}{2}} |\lambda_{h}|^{\frac{1}{4}}\Big)\;.
\end{equation}
\noindent{\underline{\it Proof of item $(i)$}.} By using the inductive assumption (\ref{eq:small}) on $R_{0}^{(h)}$ we get:
\begin{equation}
\Big|B_{\underline{0},\underline{0}}^{(h)}(\phi/L,\zeta_{\phi}) - 1\Big| \leq (1 + |\lambda_{h}|^{\frac{1}{2}})^{L^{3}} -1  \leq 2 |\lambda_{h}|^{\frac{1}{2}} L^{3}\;.
\end{equation}
For the case $n>0$, we proceed as follows. By Lemma \ref{lem:kappaNexp} part $(i)$ and by Eq. (\ref{eq:kappaNf}), we know that $\tilde f^{(h)}(\Phi/L \pm \zeta) := f^{(h)}(\Phi/L \pm \zeta) - 1$ satisfies $(\kappa, \mathcal{N}, \mathcal{M})$-bounds with:
\begin{equation}
\kappa = |\lambda_{h}|^{\frac{1}{2}}\;,\qquad \mathcal{N} = L^{-1} |\lambda_{h}|^{\frac{1}{4}}\;,\qquad \mathcal{M} = |\lambda_{h}|^{\frac{1}{4}}\;. 
\end{equation}
Therefore, by Lemma \ref{lem:kappaNexp} part $(ii)$, the function $f(\Phi/L \pm \zeta)^{\frac{L^{3}}{2}} - 1\equiv (1 + \tilde f(\Phi/L \pm \zeta))^{\frac{L^{3}}{2}} - 1$ satisfies $(\kappa', \mathcal{N}', \mathcal{M}')$-bounds with:
\begin{equation}\label{eq:Bitemi}
\kappa' = K \frac{L^{3}}{2} |\lambda_{h}|^{\frac{1}{2}}\;,\qquad \mathcal{N}' = L^{-1}|\lambda_{h}|^{\frac{1}{4}}\;,\qquad \mathcal{M}' =  |\lambda_{h}|^{\frac{1}{4}}\;.
\end{equation}
We then write:
\begin{eqnarray}\label{eq:gh1}
g^{(h)}(\Phi/L,\zeta)-1 &=& f^{(h)}(\Phi/L+\zeta)^{\frac{L^3}{2}} -1 + f^{(h)}(\Phi/L-\zeta)^{\frac{L^3}{2}} -1 \\
&& + \left (f^{(h)}(\Phi/L+\zeta)^{\frac{L^3}{2}} -1 \right )
\left ( f^{(h)}(\Phi/L-\zeta)^{\frac{L^3}{2}} -1 \right ) \;.\nonumber
\end{eqnarray}
The first two terms in the sum satisfy the $(\kappa', \mathcal{N}', \mathcal{M}')$-bound, (\ref{eq:Bitemi}). The last term satisfies a $(\kappa'', \mathcal{N}'', \mathcal{M}'')$-bound, with, using Lemma \ref{lem:kappaNprod}:
\begin{equation}
\kappa'' = \kappa'^{2}\;,\qquad \mathcal{N}'' = 2\mathcal{N}'\;,\qquad \mathcal{M}'' = 2\mathcal{M}'\;.
\end{equation}
The final $(\kappa''', \mathcal{N}''', \mathcal{M}''')$-bound for $g(\Phi/L,\zeta)-1$ follows setting:
\begin{equation}\label{eq:gh2}
\kappa''' = 2\kappa' + \kappa'' \leq 2K L^{3} |\lambda_{h}|^{\frac{1}{2}}\;,\qquad \mathcal{N}''' = L^{-1}\mathcal{M}''' = 2\mathcal{N}' + \mathcal{N}'' \leq 4 L^{-1}|\lambda_{h}|^{\frac{1}{4}}\;.
\end{equation}
We are now ready to estimate the coefficients of the function $g(\Phi/L,\zeta)$ with $n>1$. We have:
\begin{eqnarray}
|B^{(h)}_{\underline{a}, \underline{b}}(\phi/L,\zeta_{\phi})| &\leq& 2K L^{3} |\lambda_{h}|^{\frac{1}{2}} 4^{|\underline{a}| + |\underline{b}|} |\lambda_{h}|^{\frac{|\underline{a}| + |\underline{b}|}{4}}\nonumber\\
&\leq& 32 K L^{3} |\lambda_{h}|^{\frac{1}{2} + \frac{|\underline{a}| + |\underline{b}|}{4}}\;.
\end{eqnarray}
This concludes the proof of item $(i)$.
\medskip

\noindent{\underline{\it Proof of item $(ii)$}.} Here we write:
\begin{equation}
g^{(h)}(\Phi/L, \zeta) = f^{(h)}(\Phi/L + \zeta)^{\frac{L^{3}}{2}} f^{(h)}(\Phi/L - \zeta)^{\frac{L^{3}}{2}} 
\end{equation}
and we use that the $(\kappa, \mathcal{N}, \mathcal{M})$ bound of $g^{(h)}(\Phi/L, \zeta)$ is such that, by Lemma \ref{lem:kappaNprod}:
\begin{equation}
\kappa = \kappa_{+} \kappa_{-}\;,\qquad \mathcal{N} = \mathcal{N}_{+} + \mathcal{N}_{-}\;,\qquad \mathcal{M} = \mathcal{M}_{+} + \mathcal{M}_{-}
\end{equation}
where the quantities labelles by the sign $\eta = \pm$ correspond to $f^{(h)}(\Phi/L + \eta \zeta)^{\frac{L^{3}}{2}}$. Then, by Eqs. (\ref{eq:kappaNf}), (\ref{eq:kappaNflarge}), and using again Lemma \ref{lem:kappaNprod}:
\begin{eqnarray}\label{eq:kappaii}
&&\kappa_{+} = (1 + |\lambda_{h}|^{\frac{1}{2}})^{\frac{L^{3}}{2}}\;,\qquad \mathcal{N}_{+} = \frac{L^{2}}{2} |\lambda_{h}|^{\frac{3}{8}}\;,\qquad \mathcal{M}_{+} = \frac{L^{3}}{2} |\lambda_{h}|^{\frac{3}{8}} \\
&&\kappa_{-} = \delta^{\frac{L^{3}}{2}} e^{c_{h} \frac{L^{3}}{2}|\lambda_{h}| \| \phi/L -\zeta_{\phi}\|^{4}}\;,\quad \mathcal{N}_{-} = \frac{L^{2}}{2} \delta^{-\frac{1}{2}} |\lambda_{h}|^{\frac{1}{4}}\;,\quad \mathcal{M}_{-} = \frac{L^{3}}{2} \delta^{-\frac{1}{2}} |\lambda_{h}|^{\frac{1}{4}}\;.\nonumber
\end{eqnarray}
From this, the proof of item $(ii)$ easily follows.
\medskip

\noindent{\underline{\it Proof of item $(iii)$}.} We proceed as in the proof of item $(ii)$, except that now all functions depend on large fields. Hence we shall use the $(\kappa, \mathcal{N}, \mathcal{M})$-bounds for $g^{(h)}$ generated by Eq. (\ref{eq:kappaNflarge}). We easily get:
\begin{equation}\label{eq:kappaiii}
\kappa = \delta^{L^{3}} e^{\frac{L^{3}}{2} c_{h}\vert \lambda_{h} \vert \| \phi/L - \zeta_{\phi}\|^{4}+\frac{L^{3}}{2} \vert \lambda_{h} \vert \| \phi/L + \zeta_{\phi}\|^{4}}\;,\qquad \mathcal{N} = L^{-1}\mathcal{M} = L^{2} \delta^{-\frac{1}{2}} |\lambda_{h}|^{\frac{1}{4}}\;,
\end{equation}
which implies item $(iii)$.
\end{proof}
Recalling the expression (\ref{eq:UU}), (\ref{eq:widecheckV}), we rewrite:
\begin{equation}
\widecheck{V}^{(h)}(\Phi/L, \zeta) = \frac{\lambda_{h}}{L} (\Phi \cdot \Phi)^{2} + i\mu_{h} L (\Phi\cdot \Phi) + V^{(h)}_{\text{f}}(\Phi, \zeta) + V^{(h)}_{\text{b}}(\Phi, \zeta_{\phi})
\end{equation}
where:
\begin{eqnarray}\label{eq:VbVf}
V^{(h)}_{\text{b}}(\Phi, \zeta_{\phi}) &=& 4 \lambda_{h} L (\phi\cdot \zeta_{\phi})^{2} + 2\lambda_{h} L (\Phi\cdot \Phi) (\zeta_{\phi}\cdot \zeta_{\phi}) + \lambda_{h} L^{3} (\zeta_{\phi}\cdot \zeta_{\phi})^{2}\nonumber\\&& + i\mu_{h} L^{3} (\zeta_{\phi}\cdot \zeta_{\phi})\\
V_{\text{f}}^{(h)}(\Phi, \zeta) &=&  \lambda_{h} L^{3} (\zeta_{\psi}\cdot \zeta_{\psi})^{2} + 2\lambda_{h} L^{3} (\zeta_{\psi}\cdot \zeta_{\psi}) (\zeta_{\phi}\cdot \zeta_{\phi}) + 4 \lambda_{h} L (\psi\cdot \zeta_{\psi})^{2} \nonumber\\&& + 2\lambda_{h} L (\Phi\cdot \Phi) (\zeta_{\psi}\cdot \zeta_{\psi}) + 8 \lambda_{h} L (\psi\cdot \zeta_{\psi})(\phi\cdot \zeta_{\phi})\nonumber\\&& + i\mu_{h} L^{3} (\zeta_{\psi} \cdot \zeta_{\psi})\;;\nonumber
\end{eqnarray}
we rewrite:
\begin{eqnarray}\label{eq:tildeR}
U^{(h+1)}(\Phi) &=& e^{-\frac{\lambda_{h}}{L} (\Phi \cdot \Phi)^{2} - i L \mu_{h} (\Phi\cdot \Phi)} \int d\mu_{\phi}(\zeta_{\phi})\, e^{-V^{(h)}_{\text{b}}(\Phi, \zeta_{\phi})} \nonumber\\
&& \cdot \int d\mu_{\psi}(\zeta_{\psi})\, e^{-V^{(h)}_{\text{f}}(\Phi, \zeta)}\sum_{\underline{a}, \underline{b}} B^{(h)}_{\underline{a}, \underline{b}}(\phi/L,\zeta_{\phi}) L^{-|\underline{a}|} \psi^{\underline{a}} \zeta_{\psi}^{\underline{b}}\;. \nonumber
\end{eqnarray}
In the next section, we shall discuss the integration of the fermionic fluctuation field $\zeta_{\psi}$.
\subsubsection{Integration of the fermionic fluctuation}\label{sec:inductive_step}
The goal of this section is to compute:
\begin{equation}\label{eq:fermionic_int_B}
\int d\mu_{\psi}(\zeta_{\psi})\, e^{-V^{(h)}_{\text{f}}(\Phi, \zeta)}\sum_{\underline{a}, \underline{b}} B^{(h)}_{\underline{a}, \underline{b}}(\phi/L,\zeta_{\phi}) L^{-|\underline{a}|} \psi^{\underline{a}} \zeta_{\psi}^{\underline{b}}\;.
\end{equation}
The integration is performed by expanding the exponential, and by using the fermionic Wick's rule to integrate the field $\zeta_{\psi}$. By Corollary \ref{cor:symgra}, the outcome of the integration depends on the field $\psi$ only via the combination $(\psi \cdot \psi)$. That is:
\begin{eqnarray}\label{eq:hatR}
&& \int d\mu_{\psi}(\zeta_{\psi})\, e^{-V^{(h)}_{\text{f}}(\Phi, \zeta)}\sum_{\underline{a}, \underline{b}} B^{(h)}_{\underline{a}, \underline{b}}(\phi/L,\zeta_{\phi}) L^{-|\underline{a}|} \psi^{\underline{a}} \zeta_{\psi}^{\underline{b}} \nonumber\\&& = \sum_{n=0,1,2} C^{(h)}_{n}(\phi/L, \zeta_{\phi}) (\psi \cdot \psi)^{n} L^{-2n}\;,
\end{eqnarray}
for suitable functions $C^{(h)}_{n}$. The next proposition collects important properties of these functions.
\begin{proposition}\label{prp:C} Under the same assumptions of Theorem \ref{thm:effective_potential_flow}, the following is true. The functions $C^{(h)}_{n}(\phi/L, \zeta_{\phi})$ are analytic in $\phi, \zeta_{\phi}$, provided $\phi/L \pm \zeta_{\phi} \in \mathbb{S}^{(h)}\cup \mathbb{L}^{(h)}$. Let 
\begin{equation}\label{eq:alphahdef}
\alpha_{h}(\phi, \zeta_{\phi}) := 1 + |\lambda_{h}| L \|\phi\|^{2} + |\lambda_{h}| L^{3} \|\zeta_{\phi}\|^{2} + |\lambda_{h}|^{\frac{3}{2}} L^{4} \|\phi\|^{2} \|\zeta_{\phi}\|^{2}\;.
\end{equation}
Then, there exists a universal constant $K>0$ such that the following bounds hold true.
\begin{itemize}
\item[(i)] Let $\phi/L \pm \zeta_\phi \in \mathbb{S}^{(h)}$. Then,
\begin{eqnarray}\label{eq:Cbounds_zeta}
&&\qquad| C^{(h)}_{n}(\phi/L, \zeta_{\phi})  - \delta_{n,0}| \leq K \alpha_{h}(\phi, \zeta_{\phi}) |\lambda_{h}|^{\frac{n}{2}}\nonumber\\&&\cdot \Big(|\lambda_{h}|^{\frac{1}{2}} L^{3} + |\lambda_{h}| L \|\phi\|^{2} + |\lambda_{h}| L^{3} \|\zeta_{\phi}\|^{2} + |\lambda_{h}|^{\frac{3}{2}} L^{4} \|\phi\|^{2} \|\zeta_{\phi}\|^{2} \Big)\;. \nonumber
\end{eqnarray}
%
%
%
%
\item[(ii)] Let $\phi/L + \zeta_{\phi} \in \mathbb{S}^{(h)}$ and $\phi/L - \zeta_{\phi} \in \mathbb{L}^{(h)}$. Then,
\begin{equation}\label{eq:Chii}
| C^{(h)}_{n}(\phi/L, \zeta_{\phi}) | \leq  K \alpha_{h}(\phi, \zeta_{\phi})^{2} \delta^{\frac{L^{3}}{2} - n} L^{6n} |\lambda_{h}|^{\frac{n}{2}} e^{\frac{L^{3}}{2} c_{h}  \vert \lambda_{h} \vert \| \phi/L - \zeta_{\phi}\|^{4}}
\;.
\end{equation}
%
%
\item[(iii)] Let $\phi/L \pm \zeta_{\phi} \in \mathbb{L}^{(h)}$. Then,
\begin{eqnarray}\label{eq:Chiii}
&&| C^{(h)}_{n}(\phi/L, \zeta_{\phi}) |  
\\&&
\leq K \alpha_{h}(\phi, \zeta_{\phi})^{2} \delta^{L^{3} - n} L^{6n} |\lambda_{h}|^{\frac{n}{2}} e^{\frac{L^{3}}{2} c_{h}  \vert \lambda_{h} \vert \| \phi/L - \zeta_{\phi}\|^{4} + \frac{L^{3}}{2} c_{h}  \vert \lambda_{h} \vert \| \phi/L + \zeta_{\phi}\|^{4}}\;.\nonumber
\end{eqnarray}
\end{itemize}
\end{proposition}
\begin{proof} 
The functions $C_{n}^{(h)}(\phi/L,\zeta_{\phi})$ have the same analyticity domain of the functions $B_{\underline{a}, \underline{b}}^{(h)}(\phi/L,\zeta_{\phi})$ because $e^{-V_{\text{f}}^{(h)}(\Phi,\zeta)}$ has entire coefficients and because fermionic integration preserves analyticity.

In order to prove bounds for the functions $C^{(h)}_{n}$, we shall first derive $(\kappa, \mathcal{N}, \mathcal{M})$-bounds for the function $e^{-V^{(h)}_{\text{f}}(\Phi, \zeta)}$. These are proven as in the discussion of the integration of scale zero, see Eqs. (\ref{eq:sumexp0})-(\ref{eq:kappaN1}), with the only difference that now $\lambda$  and $\mu$ are replaced respectively by $\lambda_{h}$ and $\mu_{h}$. We get:
\begin{eqnarray}\label{eq:kappaN2}
&&\kappa = 1\;,\qquad \mathcal{N} = 3L^{-1}|\lambda_{h}|^{\frac{1}{4}}\;,\\
&&\mathcal{M} = 8|\lambda_{h}|^{\frac{1}{4}} L^{\frac{3}{2}} + (2\lambda L)^{\frac{1}{2}} \|\phi\| + (2|\lambda_{h}| L^{3})^{\frac{1}{2}} \|\zeta_{\phi}\| + 8 |\lambda_{h}|^{\frac{3}{4}} L^{2} \|\phi\| \|\zeta_{\phi}\|\;.\nonumber
\end{eqnarray}
Recall the definition of the function $g^{(h)}(\Phi/L, \zeta)$, Eq. (\ref{eq:gphi}).
\medskip

\noindent{\underline{\it Proof of item $(i)$}.} The goal is to obtain $(\kappa,\mathcal{N},\mathcal{M})$-bounds for 
\begin{eqnarray}\label{eq:separation_integration_C_coeff}
h^{(h)}(\psi; \psi, \zeta_{\phi}) &:=& \Big(\int d \mu _{\psi}(\zeta_{\psi})\, e^{-V_{\mathrm{f}}^{(h)}(\Phi,\zeta)} g^{(h)}(\Phi/L,\zeta)\Big)  -1\nonumber\\
&\equiv& h_{\text{I}}^{(h)}(\psi; \psi, \zeta_{\phi}) + h_{\text{II}}^{(h)}(\psi; \psi, \zeta_{\phi})\;,\nonumber\\
h_{\text{I}}^{(h)}(\psi; \psi, \zeta_{\phi}) &:=& \int d \mu _{\psi}(\zeta_{\psi})\, e^{-V_{\mathrm{f}}^{(h)}(\Phi,\zeta)} (g^{(h)}(\Phi/L,\zeta) -1)\nonumber\\ h_{\text{II}}^{(h)}(\psi; \psi, \zeta_{\phi}) &:=& \int d \mu _{\psi}(\zeta_{\psi})  \Big(e^{-V_{\mathrm{f}}^{(h)}(\Phi,\zeta)}  -1 \Big ) \;.\nonumber
\end{eqnarray}
Consider first the term $h^{(h)}_{\text{II}}$. Proceeding as in the proof of Proposition \ref{prp:Cn0} for the integration of scale zero (simply replacing $\lambda$ and $\mu$ with $\lambda_{h}$ and $\mu_{h}$), this function satisfies $(\kappa, \mathcal{N})$-bounds with, recall Eq. (\ref{eq:hest0}):
\begin{eqnarray}\label{eq:kesth}
&&\kappa =\nonumber\\
&& \widetilde{C}\alpha_{h}(\phi, \zeta_{\phi})\Big(|\lambda_{h}|^{\frac{1}{2}} L^{3} + |\lambda_{h}| L \|\phi\|^{2} + |\lambda_{h}| L^{3} \|\zeta_{\phi}\|^{2} + |\lambda_{h}|^{\frac{3}{2}} L^{4} \|\phi\|^{2} \|\zeta_{\phi}\|^{2} \Big) \nonumber\\
&&\alpha_{h}(\phi, \zeta_{\phi}) := 1 + |\lambda_{h}| L \|\phi\|^{2} + |\lambda_{h}| L^{3} \|\zeta_{\phi}\|^{2} + |\lambda_{h}|^{\frac{3}{2}} L^{4} \|\phi\|^{2} \|\zeta_{\phi}\|^{2}\nonumber\\
&&\mathcal{N} = 3 L^{-1} |\lambda_{h}|^{\frac{1}{4}}\;.
\end{eqnarray}
The condition $\phi/L \pm \zeta_{\phi} \in \mathbb{S}^{(h)}$ implies that $\|\zeta_{\phi}\| \leq |\lambda_{h}|^{-\frac{1}{4}}$, $\|\phi\|\leq L|\lambda_{h}|^{-\frac{1}{4}}$. Therefore:
\begin{equation}\label{eq:kesth2}
\kappa \leq 8 \widetilde{C}L^{6} |\lambda_{h}|^{\frac{1}{2}}\;.
\end{equation}
Consider now $h^{(h)}_{\text{I}}$. As proven in the proof of Proposition \ref{prp:C}, see Eqs. (\ref{eq:gh1})-(\ref{eq:gh2}), the function $g^{(h)}(\Phi/L,\zeta) -1$ satisfies $(\kappa', \mathcal{N}', \mathcal{M}')$-bounds with:
\begin{equation}
\kappa' = 2K L^{3} |\lambda_{h}|^{\frac{1}{2}}\;,\qquad \mathcal{N}' = 4 L^{-1}|\lambda_{h}|^{\frac{1}{4}}\;,\qquad \mathcal{M}' = 4 |\lambda_{h}|^{\frac{1}{4}}\;.
\end{equation}
To deduce bounds for $h^{(h)}_{\text{I}}(\psi; \phi, \zeta_{\phi})$, we use Lemma \ref{lem:kappaNint}. We get that $h^{(h)}_{\text{I}}(\psi; \phi, \zeta_{\phi})$ satisfies $(\kappa'', \mathcal{N}'')$-bounds, with:
\begin{eqnarray}
\kappa'' &=& \kappa( 1 + 12 \mathcal{M}'^{2} + 2\mathcal{M}'^{4} ) \leq 4K L^{3} |\lambda_{h}|^{\frac{1}{2}} \nonumber\\
\mathcal{N}'' &=& \mathcal{N}' = 4 L^{-1}|\lambda_{h}|^{\frac{1}{4}}\;.
\end{eqnarray}
The final $(\tilde{\kappa}, \tilde{\mathcal{N}})$-bounds for $h^{(h)} = h^{(h)}_{\text{I}} + h^{(h)}_{\text{II}}$ are obtained summing the estimates for the corresponding parameters of $h^{(h)}_{\text{I}}$ and $h^{(h)}_{\text{II}}$. The final claim follows using that $(C_{n}^{(h)} - \delta_{n,0}) L^{-2n}$ are the coefficients of the expansion of $h^{(h)}$ in $(\psi\cdot \psi)^{n}$.
\medskip

\noindent{\underline{\it Proof of item $(ii)$}.} By Proposition \ref{prop:B}, $g^{(h)}(\Phi/L,\zeta)$ satisfies $(\kappa, \mathcal{N}, \mathcal{M})$-bounds with:
\begin{equation}
\kappa = 4\delta^{\frac{L^{3}}{2}} e^{\frac{L^{3}}{2} c_{h}  \vert \lambda_{h} \vert \| \phi/L - \zeta_{\phi}\|^{4}}\;,\qquad \mathcal{N} = \delta^{-\frac{1}{2}} L^{2} | \lambda_{h} |^{\frac{1}{4}}\;,\qquad \mathcal{M} = \delta^{-\frac{1}{2}} L^{3} | \lambda_{h} |^{\frac{1}{4}}\;.
\end{equation}
On the other hand, the function $e^{-V_{\mathrm{f}}^{(h)}(\Phi,\zeta)}$ satisfies $(\kappa', \mathcal{N}', \mathcal{M}')$-bounds with, see Eq. (\ref{eq:kappaN1}) with $\lambda$ replaced by $|\lambda_{h}|$:
\begin{eqnarray}\label{eq:kappaNh}
&&\kappa' = 1\;,\qquad \mathcal{N}' = 3L^{-1}|\lambda_{h}|^{\frac{1}{4}}\;,\\
&&\mathcal{M}' = 8|\lambda_{h}|^{\frac{1}{4}} L^{\frac{3}{2}} + (2|\lambda_{h}| L)^{\frac{1}{2}} \|\phi\| + (2|\lambda_{h}| L^{3})^{\frac{1}{2}} \|\zeta_{\phi}\| + 8 |\lambda_{h}|^{\frac{3}{4}} L^{2} \|\phi\| \|\zeta_{\phi}\|\;.\nonumber
\end{eqnarray}
Therefore, by Lemma \ref{lem:kappaNprod}, $e^{-V_{\mathrm{f}}^{(h)}(\Phi,\zeta)} g^{(h)}(\Phi/L,\zeta)$ satisfies $(\kappa'', \mathcal{N}'', \mathcal{M}'')$-bounds with, for $L$ large enough:
\begin{eqnarray}
\kappa'' &=& \kappa\kappa' = 4\delta^{\frac{L^{3}}{2}} e^{\frac{L^{3}}{2} c_{h}  \vert \lambda_{h} \vert \| \phi/L - \zeta_{\phi}\|^{4}} \\
\mathcal{N}'' &=& \mathcal{N} + \mathcal{N}' \leq 2\delta^{-\frac{1}{2}} L^{2} | \lambda_{h} |^{\frac{1}{4}}\;,\nonumber\\
\mathcal{M}'' &=& \mathcal{M} + \mathcal{M}' \nonumber\\&\leq& 2\delta^{-\frac{1}{2}} L^{3} | \lambda_{h} |^{\frac{1}{4}} + (2|\lambda_{h}| L)^{\frac{1}{2}} \|\phi\| + (2|\lambda_{h}| L^{3})^{\frac{1}{2}} \|\zeta_{\phi}\| + 8 |\lambda_{h}|^{\frac{3}{4}} L^{2} \|\phi\| \|\zeta_{\phi}\|\;.\nonumber
\end{eqnarray}
Next, by Lemma \ref{lem:kappaNint}, the outcome of the Grassmann integration satisfies $(\kappa''', \mathcal{N}''')$-bounds with:
\begin{eqnarray}
\kappa''' &=& \kappa''( 1 + 12 \mathcal{M}''^{2} + 2\mathcal{M}''^{4} )\nonumber\\
&\leq& K\delta^{\frac{L^{3}}{2}} e^{\frac{L^{3}}{2} c_{h}  \vert \lambda_{h} \vert \| \phi/L - \zeta_{\phi}\|^{4}} \alpha_{h}(\phi, \zeta_{\phi})^{2}\nonumber\\
\mathcal{N}^{'''} &=& \mathcal{N}'' \nonumber\\&\leq& 2\delta^{-\frac{1}{2}} L^{2} | \lambda_{h} |^{\frac{1}{4}}\;,
\end{eqnarray}
from which the bound (\ref{eq:Chii}) easily follows.
\medskip

\noindent{\underline{\it Proof of item $(iii)$}.} By Proposition \ref{prop:B}, $g^{(h)}(\Phi/L,\zeta)$ satisfies $(\kappa, \mathcal{N}, \mathcal{M})$-bounds with:
\begin{eqnarray}
\kappa &=& 4\delta^{L^{3}} e^{\frac{L^{3}}{2} c_{h}\vert \lambda_{h} \vert \| \phi/L - \zeta_{\phi}\|^{4}+\frac{L^{3}}{2} \vert \lambda_{h} \vert \| \phi/L + \zeta_{\phi}\|^{4}} \nonumber\\
\mathcal{N} &=& \delta^{-\frac{1}{2}} L^{2} | \lambda_{h} |^{\frac{1}{4}}\;,\nonumber\\
\mathcal{M} &=& \delta^{-\frac{1}{2}} L^{3} | \lambda_{h} |^{\frac{1}{4}}\;,
\end{eqnarray}
while $e^{-V_{\mathrm{f}}^{(h)}(\Phi,\zeta)}$ satisfies $(\kappa', \mathcal{N}', \mathcal{M}')$-bounds (\ref{eq:kappaNh}). Thus, from the point of view of the estimates the only difference with respect to case $(ii)$ is the different $\kappa$. Hence, the final bound (\ref{eq:Chiii}) follows from the same argument used in item $(ii)$.
\end{proof}
This proposition concludes the discussion of the integration of the fermionic fluctuation field. In the next section, we shall discuss the integration of the bosonic fluctuation field.
\subsubsection{Integration of the bosonic fluctuation field}
\label{subsec:bosonic_int}
We now consider:
\begin{equation}\label{eq:intbos}
\int d\mu_{\phi}(\zeta_{\phi})\, e^{-V^{(h)}_{\text{b}}(\Phi, \zeta_{\phi})} \sum_{n = 0, 1, 2} C^{(h)}_{n}(\phi/L, \zeta_{\phi})(\psi\cdot \psi)^{n} L^{-2n}\;.
\end{equation}
To begin, we extract from $V_{\text{b}}^{(h)}(\Phi, \zeta_{\phi})$ the contribution due to the fermionic external field $\psi$. We proceed as for the integration of the scale zero, see Eqs. (\ref{eq:intbos0})-(\ref{eq:Dn}). We have:
\begin{eqnarray}\label{eq:Rnh}
&&e^{-V_{\text{b}}^{(0)}(\Phi, \zeta_{\phi})} \sum_{n=0,1,2} C^{(0)}_{n}(\phi/L, \zeta_{\phi}) (\psi\cdot \psi)^{n} L^{-2n} \\
&&\equiv  e^{-\widetilde V_{\text{b}}^{(0)}(\phi, \zeta_{\phi})} \sum_{n=0,1,2}  D^{(0)}_{n}(\phi/L, \zeta_{\phi})(\psi\cdot \psi)^{n} L^{-2n}\;,\nonumber
\end{eqnarray}
where:
\begin{eqnarray}\label{eq:Dnh}
\widetilde V_{\text{b}}^{(h)}(\phi, \zeta_{\phi}) &=& 4 \lambda_{h} L (\phi\cdot \zeta_{\phi})^{2} + 2\lambda_{h} L (\phi \cdot \phi)(\zeta_{\phi}\cdot \zeta_{\phi}) + \lambda_{h} L^{3} (\zeta_{\phi}\cdot \zeta_{\phi})^{2} \nonumber\\ && + i\mu_{h} L^{3} (\zeta_{\phi}\cdot \zeta_{\phi})\nonumber\\
D_{n}^{(h)}(\phi/L, \zeta_{\phi}) &=& \sum_{k = 0}^{n} \frac{(-2\lambda_{h} L^{3})^{k} (\zeta_{\phi}\cdot \zeta_{\phi})^{k}}{k!}   C_{n-k}^{(h)}(\phi/L, \zeta_{\phi})\;.
\end{eqnarray}
The next proposition collects important properties of the new coefficients $D^{(h)}_{n}$.
\begin{proposition}\label{prp:Dh} Under the same assumptions of Theorem \ref{thm:effective_potential_flow}, the following is true. The functions $D^{(h)}_{n}(\phi/L, \zeta_{\phi})$ are analytic in $\phi, \zeta_{\phi}$, provided $\phi/L \pm \zeta_{\phi} \in \mathbb{S}^{(h)}\cup \mathbb{L}^{(h)}$. Moreover, there exists a universal constant $K>0$ such that the following bounds hold true.
\begin{itemize}
\item[(i)] Let $\phi/L \pm \zeta_\phi \in \mathbb{S}^{(h)}$. Then,
\begin{eqnarray}\label{eq:Dhi}
&&|D_{n}^{(h)}(\phi/L, \zeta_{\phi}) - \delta_{n,0}| \leq K \tilde \alpha_{h}(\phi, \zeta_{\phi}) |\lambda_{h}|^{\frac{n}{2}} \nonumber\\&&\quad\cdot ( |\lambda_{h}|^{\frac{1}{2}} L^{3} + |\lambda_{h}| L \|\phi\|^{2} + |\lambda_{h}| L^{3} \|\zeta_{\phi}\|^{2} + |\lambda_{h}|^{\frac{3}{2}} L^{4} \|\phi\|^{2} \|\zeta_{\phi}\|^{2} ) \nonumber\\
&&\qquad + K|\lambda_{h}|^{n} L^{3n} \|\zeta_{\phi}\|^{2n} (1 - \delta_{n,0})\;,
\end{eqnarray}
with
\begin{equation}
\tilde \alpha_{h}(\phi, \zeta_{\phi}) = \alpha_{h}(\phi, \zeta_{\phi})(1 + \lambda^{\frac{1}{2}} L^{3} \| \zeta_{\phi} \|^{2} + \lambda L^{6} \| \zeta_{\phi} \|^{4} )\;.
\end{equation}
\item[(ii)] Let $\phi/L - \zeta_\phi \in \mathbb{S}^{(h)}$, $\phi/L + \zeta_\phi \in \mathbb{L}^{(h)}$. Then:
\begin{eqnarray}\label{eq:Dhii}
|D_{n}^{(h)}(\phi/L, \zeta_{\phi})| \leq K\alpha_{h}(\phi,\zeta_{\phi})^{4} \delta^{\frac{L^{3}}{2} - n} L^{6n} |\lambda_{h}|^{\frac{n}{2}}e^{\frac{L^{3}}{2}c_{h}|\lambda_{h}| \| \phi/L - \zeta_{\phi} \|^{4}}\;.
\end{eqnarray} 
\item[(iii)] Let $\phi/L \pm \zeta_\phi \in \mathbb{L}^{(h)}$. Then:
\begin{eqnarray}\label{eq:Dhiii}
&&|D_{n}^{(h)}(\phi/L, \zeta_{\phi})| \\
&&\leq K\alpha_{h}(\phi,\zeta_{\phi})^{4} \delta^{L^{3} - n} L^{6n} |\lambda_{h}|^{\frac{n}{2}}e^{\frac{L^{3}}{2}c_{h}|\lambda_{h}| \| \phi/L - \zeta_{\phi} \|^{4} + \frac{L^{3}}{2}c_{h}|\lambda_{h}| \| \phi/L + \zeta_{\phi} \|^{4}}\;.\nonumber
\end{eqnarray}
\end{itemize}
\end{proposition}
\begin{proof} The statement about analyticity follows immediately from the analyticity of the $C^{(h)}_{n}$ functions. Also, the proof of item $(i)$ is identical to the proof of Proposition \ref{prp:bdD}, since the bounds for the $C^{(h)}_{n}$ functions for $\phi/L \pm \zeta_\phi \in \mathbb{S}^{(h)}$ are identical to the bounds for the $C^{(0)}_{n}$ functions, after replacing $\lambda$ with $|\lambda_{h}|$; compare (\ref{eq:bdC}) with (\ref{eq:Cbounds_zeta}). Let us now discuss the proof of the remaining two items.
\medskip

\noindent{\underline{\it Proof of part $(ii)$}.} The statement for $n=0$ is trivial, since $D^{(h)}_{0} = C^{(h)}_{0}$. Suppose now that $n\neq 0$. We write:
\begin{eqnarray}
|D^{(h)}_{n}(\phi/L, \zeta_{\phi})| &\leq& K\sum_{k=0}^{n} |\lambda_{h}|^{k} L^{3k} \|\zeta_{\phi}\|^{2k} |C^{(h)}_{n-k}(\phi/L, \zeta_{\phi})|\;.
\end{eqnarray}
Plugging in the bound (\ref{eq:Chii}) for $C^{(h)}_{n-k}$, we get:
\begin{eqnarray}
|D^{(h)}_{n}(\phi/L, \zeta_{\phi})| &\leq& K\alpha_{h}(\phi,\zeta_{\phi})^{2} \delta^{\frac{L^{3}}{2} - n} L^{6n} |\lambda_{h}|^{\frac{n}{2}} e^{\frac{L^{3}}{2}c_{h}|\lambda_{h}| \| \phi/L - \zeta_{\phi} \|^{4}}\nonumber\\&&\cdot \sum_{k=0}^{n} |\lambda_{h}|^{k} L^{3k} \|\zeta_{\phi}\|^{2k} \delta^{k} L^{-6k} |\lambda_{h}|^{-\frac{k}{2}}\nonumber\\
&\leq& K\alpha_{h}(\phi,\zeta_{\phi})^{2} \delta^{\frac{L^{3}}{2} - n} L^{6n} |\lambda_{h}|^{\frac{n}{2}}e^{\frac{L^{3}}{2}c_{h}|\lambda_{h}| \| \phi/L - \zeta_{\phi} \|^{4}}\nonumber\\
&&\cdot ( 1 + |\lambda_{h}|^{\frac{1}{2}} \|\zeta_{\phi}\|^{2} + |\lambda_{h}| \|\zeta_{\phi}\|^{4} )
\end{eqnarray}
where in the last step we chose $L$ large enough. The final claim (\ref{eq:Dhii}) follows from
\begin{equation}
\alpha_{h}(\phi,\zeta_{\phi})^{2} ( 1 + |\lambda_{h}|^{\frac{1}{2}} \|\zeta_{\phi}\|^{2} + |\lambda_{h}| \|\zeta_{\phi}\|^{4} ) \leq \alpha_{h}(\phi,\zeta_{\phi})^{4}\;.
\end{equation}
\medskip

\noindent{\underline{\it Proof of item $(iii)$}.} The proof of item $(iii)$ is identical to the proof of item $(ii)$, the only difference being that one has to use the bound (\ref{eq:Chiii}) for the functions $C^{(h)}_{n-k}$. We omit the details.
\end{proof}
Proposition \ref{prp:Dh} implies that the functions $D^{(h)}_{n}$ satisfy the following (nonoptimal) bound, for all values of $\phi, \zeta_{\phi}$ such that $\phi/L \pm \zeta_{\phi} \in \mathbb{S}^{(h)} \cup \mathbb{L}^{(h)}$:
\begin{eqnarray}\label{eq:rough}
&&| D^{(h)}_{n}(\phi/L, \zeta_{\phi}) | \\
&&\leq C_{L} |\lambda_{h}|^{\frac{n}{2}} (1 + |\lambda_{h}| \|\zeta_{\phi}\|^{4})^{4}  e^{\frac{|\lambda_{h}|^{\frac{3}{2}}}{L} \|\phi\|^{4}} e^{\frac{L^{3}}{2}c_{h}|\lambda_{h}| \| \phi/L - \zeta_{\phi} \|^{4} + \frac{L^{3}}{2}c_{h}|\lambda_{h}| \| \phi/L + \zeta_{\phi} \|^{4}}\;,\nonumber
\end{eqnarray}
where we used that
\begin{equation}
\alpha_{h}(\phi,\zeta_{\phi}) \leq K(1 + |\lambda_{h}| \|\zeta_{\phi}\|^{4}) e^{\frac{|\lambda_{h}|^{\frac{3}{2}}}{L} \|\phi\|^{4}}\;,
\end{equation}
for $\lambda$ small enough. This bound will be used to estimate the remainder terms of the stationary phase expansion.

\subsubsection{Small field regime}\label{subsec:small_fields}
We define:
\begin{equation}\label{eq:Ehn0}
E^{(h)}_{n}(\phi) =  L^{-2n} \int d\mu_{\phi}(\zeta_{\phi})\, e^{-\widetilde V^{(h)}_{\text{b}}(\phi, \zeta_{\phi})} D^{(h)}_{n}(\phi/L, \zeta_{\phi})\;,
\end{equation}
in terms of which the effective potential can be rewritten as:
\begin{equation}\label{eq:Uh+1E}
U^{(h+1)}(\Phi) = e^{-\frac{\lambda_{h}}{L} (\Phi \cdot \Phi)^{2} - i L \mu_{h} (\Phi\cdot \Phi)} \sum_{n= 0,1,2}  E^{(h)}_{n}(\phi) (\psi \cdot \psi)^{n}\;.
\end{equation}
We proceed in a way analogous to the integration of scale zero, see Section \ref{sec:sfregime}. To begin, we have to prove that the bounds we derived on the functions $D^{(h)}_{n}$ are compatible with integrability. This is an immediate consequence of the following bound, whose proof is deferred to Appendix \ref{app:osc}.
\begin{proposition}\label{prp:mixed} Let $\varepsilon > 0$ as in the inductive assumptions (\ref{eq:indh}). Suppose that $\|\mathrm{Im}\, \zeta_{\phi}\| \leq |\lambda_{h}|^{-\frac{1}{4} + \eps}$, $\|\mathrm{Im}\, (\phi/L \pm \zeta_{\phi})\|\leq |\lambda_{h}|^{-\frac{1}{4}}$. Suppose that $c_{h}$ is as in Eq. (\ref{eq:indh}). Then:
\begin{eqnarray}\label{eq:bdmixed}
&&\Big| e^{-\widetilde{V}^{(h)}_{\text{b}}(\phi, \zeta_{\phi})} e^{c_{h}\frac{L^{3}}{2}\lambda_{h} \| \phi/L + \zeta_{\phi} \|^{4} + c_{h}\frac{L^{3}}{2}\lambda_{h} \| \phi/L - \zeta_{\phi} \|^{4} }  \Big| \nonumber\\
&&\qquad \leq C_{L} e^{-|\lambda_{h}| \frac{L^{3}}{2} \|\zeta_{\phi}\|^{4} + (c_{h} + |\lambda_{h}|^{4\varepsilon})\frac{|\lambda_{h}|}{L} \|\phi\|^{4}}\;.
\end{eqnarray}
\end{proposition}
Proposition \ref{prp:mixed} together with the bound (\ref{eq:rough}) immediately implies the following estimate, for $\|\text{Im}\, \zeta_{\phi}\| \leq |\lambda_{h}|^{-\frac{1}{4} + \eps}$, $\|\text{Im}\, (\phi/L \pm \zeta_{\phi})\|\leq |\lambda_{h}|^{-\frac{1}{4}}$:
\begin{equation}\label{eq:mixedD}
L^{-2n} \Big| e^{-\widetilde{V}^{(h)}_{\text{b}}(\phi, \zeta_{\phi})}  D^{(h)}_{n}(\phi/L, \zeta_{\phi})  \Big| \leq C_{L} |\lambda_{h}|^{\frac{n}{2}} e^{-|\lambda_{h}| \frac{L^{3}}{4} \|\zeta_{\phi}\|^{4} + \hat c_{h}\frac{|\lambda_{h}|}{L} \|\phi\|^{4}}\;,
\end{equation}
with $\hat c_{h} = c_{h} + |\lambda_{h}|^{\frac{1}{2}} + |\lambda_{h}|^{4\varepsilon}$. 

Next, recall that, by Proposition \ref{prp:Dh}, the integrand is analytic in $\phi$ and $\zeta_{\phi}$, provided $\phi/L \pm \zeta_{\phi} \in \mathbb{S}^{(h)} \cup \mathbb{L}^{(h)}$. Therefore, since in the integral $\zeta_{\phi} \in \mathbb{R}^{4}$, the integrand is analytic for $\phi \in L\mathbb{S}^{(h)}$. Moreover, thanks to Proposition \ref{prp:mixed}, for these values of $\phi$ the integrand is absolutely integrable in $\zeta_{\phi}$. In conclusion, by dominated convergence and by Morera's theorem, the function $E^{(h)}_{n}(\phi)$ is analytic in $\phi \in L\mathbb{S}^{(h)}$.

Let us now derive estimates for $E_{n}^{(h)}(\phi)$. The analysis is similar to the one performed for the integration of the zero scale, Section \ref{sec:sfregime}. For reasons that will be clear in a moment, in what follows we shall restrict $\phi$ to a smaller domain:
\begin{equation}\label{eq:tildeSh}
\phi \in L\widetilde{\mathbb{S}}^{(h)} := \Big\{ \phi \in L \mathbb{S}^{(h)} \mid \| \text{Im}\, \phi \| \leq \frac{L}{2} |\lambda_{h}|^{-\frac{1}{4}} \Big\}\;.
\end{equation}
We write, by stationary phase expansion:
\begin{eqnarray}\label{eq:Ehn}
E^{(h)}_{n}(\phi) &=& L^{-2n}D^{(h)}_{n}(\phi/L, 0) + L^{-2n}d_{1}\Big(\Delta e^{-\widetilde V^{(h)}_{\text{b}}(\phi, \cdot)} D^{(h)}_{n}(\phi/L, \cdot)\Big)(0)\nonumber\\&& + L^{-2n}\mathcal{E}_{2}\Big(e^{-\widetilde V_{\text{b}}^{(h)}(\phi, \cdot )} D^{(h)}_{n}(\phi/L, \cdot)\Big)\;.
\end{eqnarray}
By Eq. (\ref{eq:Dhi}):
\begin{equation}\label{eq:compact_D_bound}
L^{-2n}|D^{(h)}_{n}(\phi/L, 0) - \delta_{n,0}|\leq K L^{3 - 2n} |\lambda_{h}|^{\frac{1}{2} + \frac{n}{2}}\;. 
\end{equation}
Consider now the second term in Eq. (\ref{eq:Ehn}). We have:
\begin{eqnarray}\label{eq:Dhn2}
&&L^{-2n} \Big(\Delta e^{-\widetilde V_{\text{b}}^{(h)}(\phi, \cdot)} D^{(h)}_{n}(\phi/L, \cdot)\Big)(0) \\
&&= L^{-2n}\Big(\Delta e^{-\widetilde V_{\text{b}}^{(h)}(\phi, \cdot)}\Big)(0) D^{(h)}_{n}(\phi/L,0) + L^{-2n}\Big(\Delta D^{(h)}_{n}(\phi/L, \cdot)\Big)(0)\;;\nonumber
\end{eqnarray}
the first term is estimated using the analogue of the bound (\ref{eq:D003}). We have:
\begin{eqnarray}\label{eq:D0h3}
\Big|\Big(\Delta e^{-\widetilde V_{\text{b}}^{(h)}(\phi, \cdot)}\Big)(0)\Big| = \Big| (\Delta \widetilde V_{\text{b}}^{(h)}(\phi, \cdot))(0)  \Big| &\leq& K(|\lambda_{h}| L \|\phi\|^{2} + |\mu_{h}| L^{3}) \nonumber\\
&\leq& K|\lambda_{h}|^{\frac{1}{2}} L^{3}\;.
\end{eqnarray}
Hence we, using Proposition \ref{prp:Dh} part $(i)$, together with the bound (\ref{eq:D0h3}):
\begin{equation}
L^{-2n}\Big| \Big(\Delta e^{-\widetilde V_{\text{b}}^{(h)}(\phi, \cdot)}\Big)(0) D^{(h)}_{n}(\phi/L,0)\Big| \leq \widetilde K L^{6 - 2n} |\lambda_{h}|^{\frac{n}{2} + 1}\;.
\end{equation}
The second term in Eq. (\ref{eq:Dhn2}) is bounded using a Cauchy estimate, for $\|\zeta_{\phi}\| \leq (1/2) |\lambda_{h}|^{-1/4}$, in order to make sure that $\phi/L \pm \zeta_{\phi} \in \mathbb{S}^{(h)}$. We have, proceeding as in Eqs. (\ref{eq:cauDh0})-(\ref{eq:L2n}):
\begin{equation}\label{eq:compact_Delta_D_bound}
L^{-2n}\Big|\Big(\Delta D^{(h)}_{n}(\phi/L, \cdot)\Big)(0) \leq C_{L} |\lambda_{h}|^{\frac{n}{2} + 1}\;.
\end{equation}
Finally, consider the third term in Eq. (\ref{eq:Ehn}). This term will be estimated using Lemma \ref{lem:cau}, which requires analyticity in $\zeta_{\phi}$ in $\mathbb{R}^{4}_{W}$. We shall choose $W = |\lambda_{h}|^{-\frac{1}{4} + \varepsilon}$; due to the restriction to the smaller set of fields $\phi \in L\widetilde{\mathbb{S}}^{(h)}$ in Eq. (\ref{eq:tildeSh}), the condition $\zeta_{\phi} \in \mathbb{R}^{4}_{W}$ implies that $\phi/L\pm \zeta_{\phi} \in \mathbb{S}^{(h)}\cup \mathbb{L}^{(h)}$; hence, the argument of the integral in Eq. (\ref{eq:Ehn0}) is analytic in $\phi \in L \widetilde{\mathbb{S}}^{(h)}$ and in $\zeta_{\phi} \in \mathbb{R}^{4}_{W}$. Thanks to Lemma \ref{lem:cau} and the bound (\ref{eq:mixedD}), we get:
\begin{equation}\label{eq:compact_error_bound}
L^{-2n}\Big|\mathcal{E}_{2}\Big(e^{-\widetilde V_{\text{b}}^{(h)}(\phi, \cdot )} D^{(h)}_{n}(\phi/L, \cdot)\Big)\Big|\leq K_{L} |\lambda_{h}|^{\frac{n}{2} + 1 - 8\varepsilon}\;.
\end{equation}
In conclusion, for $\lambda$ small enough uniformly in $h$, for all $\phi \in L\widetilde{\mathbb{S}}^{(h)}$:
\begin{equation}\label{eq:bdgenh}
|E^{(h)}_{n}(\phi) - \delta_{n,0}|\leq 2 KL^{3 - 2n} |\lambda_{h}|^{\frac{n}{2} + \frac{1}{2}}\;.
\end{equation}
Notice that this bound is identical to the corresponding one obtained on scale zero, Eqs. (\ref{eq:bdgen}), (\ref{eq:bdgen2}), except that now $\lambda$ is replaced by $|\lambda_{h}|$. Notice also that that, by supersymmetry, $E^{(h)}_{0}(0) = 1$; see Remark \ref{rem:Nh} in Appendix \ref{app:SUSY}.
\medskip

\noindent{\underline{\it Localization and renormalization.}} Next, we define a localization and a renormalization procedure, restricting the set of allowed values of $\phi$. Given $\lambda_{h+1} \in \mathbb{C}$ (to be determined later) such that $| L\lambda_{h+1} - \lambda_{h} | \leq C|\lambda_{h}|^{\frac{3}{2}}$, we define:
\begin{equation}\label{eq:Sh+1}
\mathbb{S}^{(h+1)} := \{ \phi \in \mathbb{C}^{4} \mid \|\phi\| \leq |\lambda_{h+1}|^{-1/4} \} \subset L\widetilde{\mathbb{S}}^{(h)}\;.
\end{equation}
Notice that, for some universal constant $c>0$:
\begin{equation}
\text{dist}(\mathbb{S}^{(h+1)}, L\widetilde{\mathbb{S}}^{(h)c}) \geq cL |\lambda_{h}|^{-\frac{1}{4}}\;.
\end{equation}
We set $E_{n}^{(h)}(\phi) = \delta_{n,0} + \mathcal{L} E_{n}^{(h)}(\phi) + \mathcal{R} E_{n}^{(h)}(\phi)$ with:
\begin{equation}\label{eq:tayh}
\mathcal{L} E^{(h)}_{n}(\phi) := \left\{ \begin{array}{ccc} E^{(h)}_{2}(0) & \text{if $n=2$} \\ E^{(h)}_{1}(0) + \frac{1}{2}(\phi\cdot \phi) \partial^{2}_{\|\phi\|} E^{(h)}_{1}(0) & \text{if $n=1$} \\ \frac{1}{2}(\phi\cdot \phi)\partial^{2}_{\|\phi\|}E^{(h)}_{0}(0) + \frac{1}{4!} (\phi\cdot \phi)^{2} \partial_{\|\phi\|}^{4}E^{(h)}_{0}(0)  & \text{if $n=0$.} \end{array} \right.
\end{equation}
Arrived at this point, the discussion of the small field regime is identical to the one for scale zero, see Eqs. (\ref{eq:RE0})-(\ref{eq:smallfin0}), except that in all estimates one has to replace $\lambda$ with $|\lambda_{h}|$. We get, by the bound (\ref{eq:bdgenh}), using Cauchy estimates for $\phi \in \mathbb{S}^{(h+1)}$:
\begin{eqnarray}\label{eq:REh}
&&\big| \mathcal{R} E^{(h)}_{2}(\phi) \big| \leq K L^{-3} |\lambda_{h}|^{2} \|\phi\|^{2}\;,\qquad \big|\mathcal{R} E^{(h)}_{1}(\phi)\big| \leq K L^{-3} |\lambda_{h}|^{2} \|\phi\|^{4}\;,\nonumber\\
&&\qquad\qquad\qquad \quad \big|\mathcal{R} E^{(h)}_{0}(\phi) \big| \leq K L^{-3} |\lambda_{h}|^{2} \|\phi\|^{6}\;.
\end{eqnarray}
Concerning the contribution of the $\mathcal{L}$-term, we set:
\begin{eqnarray}\label{eq:Ebetah}
&&\gamma_{\psi,2}^{(h)} := E^{(h)}_{1}(0)\;,\qquad \gamma_{\phi,2}^{(h)} := \frac{1}{2}\partial^{2}_{\|\phi\|} E^{(h)}_{0}(0) \\
&&\gamma_{\psi\psi,4}^{(h)} := E^{(h)}_{2}(0)\;,\quad \gamma_{\phi\psi,4}^{(h)} := \frac{1}{2} \partial^{2}_{\|\phi\|} E^{(h)}_{1}(0)\;,\quad \gamma_{\phi\phi,4}^{(h)} := \frac{1}{4!} \partial_{\|\phi\|}^{4} E^{(h)}_{0}(0)\;.\nonumber
\end{eqnarray}
By supersymmetry, see Corollary \ref{cor:cons}, Appendix \ref{app:SUSY}:
\begin{equation}\label{eq:gammaeq}
\gamma_{\psi,2}^{(h)} = \gamma_{\phi,2}^{(h)} =: \gamma^{(h)}_{2}\;,\qquad \gamma_{\psi\psi,4}^{(h)} = \frac{1}{2} \, \gamma_{\phi\psi,4}^{(h)} = \gamma_{\phi\phi,4}^{(h)} =: \gamma^{(h)}_{4}\;.
\end{equation}
Also, by Cauchy estimates:
\begin{equation}\label{eq:gammahest}
|\gamma^{(h)}_{4}| \leq K L^{-1} |\lambda_{h}|^{\frac{3}{2}}\;,\qquad |\gamma_{2}^{(h)}| \leq K L|\lambda_{h}|\;.
\end{equation}
This concludes the small field analysis.
\subsubsection{Large field regime}\label{sec:largefieldh}
To begin, notice that the function $E^{(h)}_{n}(\phi)$ is analytic in $\phi \in L\mathbb{L}^{(h)}$. This follows from the analyticity in $\phi$, $\zeta_{\phi}$ of the argument of the integral in Eq. (\ref{eq:Ehn0}), which holds provided $\phi/L \pm \zeta_{\phi} \in \mathbb{S}^{(h)} \cup \mathbb{L}^{(h)}$, and from the bound (\ref{eq:mixedD}), which ensures integrability in $\zeta_{\phi}$. As for the small field regime, analyticity in $\phi \in L\mathbb{L}^{(h)}$ follows from dominated convergence, and from Morera's theorem. We shall now prove bounds for $E^{(h)}_{n}(\phi)$, in the domain:
\begin{equation}
\mathbb{L}^{(h+1)} := \{ \phi \in \mathbb{C}^{4} \mid \|\phi\| > |\lambda_{h+1}|^{-\frac{1}{4}}\;,\qquad \|\text{Im}\, \phi\| \leq |\lambda_{h+1}|^{-\frac{1}{4}} \}\;,
\end{equation}
with $\lambda_{h+1}$ as in the small field regime (to be determined later). We shall discuss separately the regions $\phi \in \mathbb{L}^{(h+1)} \cap L \mathbb{S}^{(h)}$ and $\phi \in \mathbb{L}^{(h+1)} \cap L \mathbb{S}^{(h)c}$.
\medskip

\noindent{\underline{\it Case $\phi\in \mathbb{L}^{(h+1)} \cap L \mathbb{S}^{(h)}$}.} In this region we rely on the nonrenormalized bounds (\ref{eq:bdgenh}), using that, for $L$ large enough:
\begin{equation}
\mathbb{L}^{(h+1)} \cap L \mathbb{S}^{(h)} = \mathbb{L}^{(h+1)} \cap L \widetilde{\mathbb{S}}^{(h)}\;.
\end{equation}
We have, proceeding as for the bound (\ref{eq:R01}), for $\lambda$ small enough:
\begin{eqnarray}\label{eq:tildedh}
|E^{(h)}_{0}(\phi)| &\leq& 1 + K L^{3} |\lambda_{h}|^{\frac{1}{2}} \nonumber\\
&\leq& \tilde\delta e^{\frac{|\lambda_{h+1}|}{4} \|\phi\|^{4}}\;,
\end{eqnarray} 
with $\tilde \delta < \delta$, say $\tilde \delta = e^{-\frac{1}{6}}$. Similarly, for $n=1,2$:
\begin{eqnarray}
|E^{(h)}_{n}(\phi)| &\leq& KL^{3 - 2n} |\lambda_{h}|^{\frac{n}{2} + \frac{1}{2}} \nonumber\\
&\leq& C_{L} |\lambda_{h}|^{\frac{n}{2} + \frac{1}{2}} e^{\frac{|\lambda_{h+1}|}{4} \|\phi\|^{4}}\;.
\end{eqnarray}
\noindent\underline{\it Case $\phi \in \mathbb{L}^{(h+1)} \cap L \mathbb{S}^{(h)c}$.} We compute $E_{n}^{(h)}(\phi)$ performing one step of stationary phase expansion. We have:
\begin{eqnarray}\label{eq:spalargeh}
E_{n}^{(h)}(\phi) &=& L^{-2n} \int d\mu_{\phi}(\zeta_{\phi})\, e^{-\widetilde V^{(h)}_{\text{b}}(\phi, \zeta_{\phi})} D^{(h)}_{n}(\phi/L, \zeta_{\phi})\\
&=& L^{-2n} e^{-\widetilde V^{(h)}_{\text{b}}(\phi, 0)} D^{(h)}_{n}(\phi/L, 0) + L^{-2n} \mathcal{E}_{1}\Big( e^{-\widetilde V^{(h)}_{\text{b}}(\phi, \cdot)} D^{(h)}_{n}(\phi/L, \cdot ) \Big)\;.\nonumber
\end{eqnarray}
Consider the first term. To estimate it, we use Proposition \ref{prp:Dh} item $(iii)$. We have, for $L$ large enough:
\begin{eqnarray}
L^{-2n} \Big| e^{-\widetilde V^{(h)}_{\text{b}}(\phi, 0)} D^{(h)}_{n}(\phi/L, 0)\Big| &\leq& K \alpha_{h}(\phi,0)^{4} \delta^{L^{3} - n} L^{6n} |\lambda_{h}|^{\frac{n}{2}}e^{c_{h} \frac{|\lambda_{h}|}{L} \| \phi \|^{4}}\nonumber\\
&\leq& \delta^{\frac{L^{3}}{2}} |\lambda_{h}|^{\frac{n}{2}} e^{\hat c_{h} \frac{|\lambda_{h}|}{L} \| \phi \|^{4}}\;,
\end{eqnarray}
with $\hat c_{h}$ as in Eq. (\ref{eq:mixedD}). We used that $\alpha_{h}(\phi, 0) = 1 + |\lambda_{h}| L \|\phi\|^{2} \leq e^{|\lambda_{h}| L \|\phi\|^{2}} \leq \widetilde{K} e^{\frac{|\lambda_{h}|^{\frac{3}{2}}}{L} \|\phi\|^{4}}$ for a universal constant $\widetilde{K}$, for $\lambda$ small enough. Consider now the remainder term of the stationary phase expansion in Eq. (\ref{eq:spalargeh}). By the bound (\ref{eq:mixedD}), together with Lemma \ref{lem:cau}, for $W = |\lambda_{h}|^{-\frac{1}{4} + \varepsilon}$:
\begin{equation}
L^{-2n} \Big|\mathcal{E}_{1}\Big( e^{-\widetilde V^{(h)}_{\text{b}}(\phi, \cdot)} D^{(h)}_{n}(\phi/L, \cdot ) \Big)\Big|\leq C_{L} |\lambda_{h}|^{\frac{n}{2} + \frac{1}{2} - 6\varepsilon} e^{\hat c_{h}\frac{|\lambda_{h}|}{L} \|\phi\|^{4}}\;.
\end{equation}
In conclusion, for $\tilde \delta < \delta$ as in Eq. (\ref{eq:tildedh}), and for $\tilde c_{h} = \hat c_{h} + |\lambda_{h}|$, for $\lambda$ small enough:
\begin{equation}\label{eq:bdgenla}
|E^{(h)}_{0}(\phi)| \leq \tilde\delta e^{\tilde c_{h} |\lambda_{h+1}| \|\phi\|^{4}}\;,\qquad |E^{(h)}_{n}(\phi)| \leq L^{-2n} |\lambda_{h}|^{\frac{n}{2}} e^{\tilde c_{h}\frac{|\lambda_{h}|}{L} \|\phi\|^{4}}\;\quad (n=1,2)\;.
\end{equation}
This concludes the discussion of the large field regime.
\subsubsection{The effective potential on scale $h+1$}\label{sec:scaleh}
We obtained:
\begin{equation}
U^{(h+1)}(\Phi) = e^{-\frac{\lambda_{h}}{L} (\Phi \cdot \Phi)^{2} - i \mu_{h} L (\Phi\cdot \Phi)} \sum_{n=0,1,2} E^{(h)}_{n}(\phi) (\psi \cdot \psi)\;.
\end{equation}
The functions $E^{(h)}_{n}(\phi)$ are analytic for $\phi \in L \mathbb{S}^{(h)} \cup L\mathbb{L}^{(h+1)}$. Moreover, they satisfy the bounds (\ref{eq:bdgenh}) for $\phi \in L\widetilde{\mathbb{S}}^{(h)}$ and the bounds (\ref{eq:bdgenla}) for $\phi \in \mathbb{L}^{(h+1)}$. Instead, the renormalized functions $\mathcal{R} E^{(h)}_{n}$ satisfy the bounds (\ref{eq:REh}) in the domain  $\mathbb{S}^{(h+1)}$.

As we did after the integration of the scale zero, we now redefine the effective coupling constant and the chemical potential, taking into account the terms extracted in the renormalization procedure. 
\medskip

\noindent{\underline{\it Small field bounds.}} Here we proceed exactly as for the corresponding discussion on scale zero, Eqs. (\ref{eq:ULUR})-(\ref{eq:small1}). The only difference is that now $\lambda$ is replaced by $\lambda_{h}$. We get, for $\phi \in \mathbb{S}^{(h+1)}$:
\begin{equation}\label{eq:URh}
U^{(h+1)}(\Phi) = e^{-\lambda_{h+1} (\Phi\cdot \Phi)^{2} - i\mu_{h+1} (\Phi\cdot \Phi)} \sum_{n=0,1,2} R_{n}^{(h+1)}(\phi) (\psi \cdot \psi)^{n}\;, 
\end{equation}
with:
\begin{eqnarray}\label{eq:beta_functions_scale_h}
&&\lambda_{h+1} := L^{-1} \lambda_{h} + \beta^{(h)}_{4}\;,\qquad \mu_{h+1} := L \mu_{h} + \beta_{2}^{(h)} \nonumber\\
&&\beta^{(h)}_{2} = i\gamma^{(h)}_{2}\;,\qquad \beta^{(h)}_{4} = -\gamma^{(h)}_{4} - \frac{\gamma_{2}^{(h)2}}{2}\;,
\end{eqnarray}
hence $|\beta^{(h)}_{2}| \leq K L|\lambda_{h}|$ and $|\beta^{(h)}_{4}| \leq K L^{-1}|\lambda_{h}|^{\frac{3}{2}}$, by the bounds (\ref{eq:gammahest}). The functions $R^{(h+1)}_{n}(\phi)$ are analytic in $\phi \in \mathbb{S}^{(h+1)}$ and satisfy the bounds:
\begin{eqnarray}\label{eq:Rsmall}
&&|R^{(h+1)}_{2}(\phi)| \leq |\lambda_{h+1}|^{2}\|\phi\|^{2}\;,\qquad |R^{(h+1)}_{1}(\phi)|\leq |\lambda_{h+1}|^{2}  \|\phi\|^{4}\;,\nonumber\\&&
\quad\quad\quad\qquad | R^{(h+1)}_{0}(\phi) - 1 | \leq |\lambda_{h+1}|^{2} \|\phi\|^{6}\;.
\end{eqnarray}
\noindent{\underline{\it Large field bounds.}} Proceeding as for the scale zero, Eqs. (\ref{eq:large11})-(\ref{eq:large12}), we rewrite the effective potential as in Eq. (\ref{eq:URh}), where now the functions $R^{(h+1)}_{n}(\phi)$ satisfy the bounds, for $\phi \in \mathbb{L}^{(h+1)}$:
\begin{eqnarray}\label{eq:Rlarge}
|R^{(h+1)}_{0}(\phi)| &\leq& \delta e^{c_{h+1} \vert \lambda_{h+1} \vert \|\phi\|^{4}}\;,\nonumber\\
|R^{(h+1)}_{n}(\phi)| &\leq& \vert \lambda_{h+1} \vert^{\frac{n}{2}} e^{c_{h+1} \vert \lambda_{h+1} \vert \|\phi\|^{4}}\qquad n=1,2\;,
\end{eqnarray}
with $c_{h+1} = \tilde c_{h} + |\lambda_{h}|^{\frac{1}{2}}$. This concludes the check of the inductive assumptions for the effective potential on scale $h+1$, and concludes the proof of Theorem \ref{thm:effective_potential_flow}.\qed 
\section{Proof of Theorem \ref{thm:main}}\label{sec:2pt}
\subsection{Setting up the multiscale analysis}
In this section we shall adapt the method developed in Section \ref{sec:RGpart} to the computation of the two-point correlation function, in order to prove Theorem \ref{thm:main}. The same method could be applied to the evaluation of higher correlations, with a larger number of internal degrees of freedom. We will omit this extension.

By supersymmetry, it will be enough to study the bosonic two-point function, see Eq.~\eqref{eq: bosonic_fermionic_two-point_function} below.
%
%
Following Eq. (\ref{eq:decompPHI}), we rewrite the fields $\phi^{+}_{x}$, $\phi^{-}_{y}$ as:
\begin{eqnarray}\label{eq:splitphi}
\phi^{+}_{x} &\equiv& \phi^{(\geq 0)+}_{x} = \frac{1}{L}\phi^{(\geq 1)+}_{\fl{L^{-1} x}} + A_{x} \zeta^{(0)+}_{\phi,\fl{ L^{-1} x} }\nonumber\\  
\phi^{-}_{y} &\equiv& \phi^{(\geq 0)-}_{y} = \frac{1}{L}\phi^{(\geq 1)-}_{\fl{L^{-1} y}} + A_{y} \zeta^{(0)-}_{\phi,\fl{ L^{-1} y}}\;.
\end{eqnarray}
Let $\fl{L^{-1}x } \neq \fl{ L^{-1}y }$. For the sake of notation, in this section we shall drop the spin label, unless otherwise stated. We compute the two-point correlation function with equal spins; by spin symmetry, the two-point correlation function with different spins is trivially zero. Plugging the decomposition (\ref{eq:splitphi}) in $\langle \phi^{+}_x  \phi^{-}_y\rangle$, and using that:
\begin{equation}\label{eq:int_z}
\int d\mu(\zeta^{(0)})\, [U^{(0)}(L^{-1}\Phi+\zeta^{(0)}) U^{(0)}(L^{-1}\Phi-\zeta^{(0)})]^{\frac{L^3}{2}} \zeta^{(0)\pm}_{\phi,\fl{ L^{-1}x }}=0\qquad \forall x\;,
\end{equation}
we get:
\begin{equation}
\langle \phi^{+}_x \phi^{-}_y\rangle_{N} = L^{-2}\langle \phi^{(\geq 1)+}_{\fl{L^{-1}x}} \phi^{(\geq 1)-}_{\fl{L^{-1}y}}\rangle_{N}\;.
\end{equation}
This procedure can be iterated. Let $k\in \mathbb{N}$ be the first integer such that $\fl{ L^{-k}x} = \fl{ L^{-k}y}$. Then, for all $h\leq k$:
\begin{eqnarray}\label{eq:startcor}
\langle \phi^{+}_x \phi^{-}_y\rangle_{N} &=& L^{-2(h-1)} \langle \phi^{(\geq h-1)+}_{\fl{L^{-h+1}x}} \phi^{(\geq h-1)-}_{\fl{L^{-h+1}y}}\rangle_{N}\nonumber\\
&\equiv& L^{-2(k-1)} \langle \phi^{(\geq k-1)+}_{\fl{L^{-k+1}x}} \phi^{(\geq k-1)-}_{\fl{L^{-k+1}y}}\rangle_{N}\;.
\end{eqnarray}
We are left with discussing the evaluation of $\big\langle \phi^{(\geq k-1)+}_{\fl{L^{-k+1}x}} \phi^{(\geq k-1)-}_{\fl{L^{-k+1}y}}\big\rangle_{N}$. To begin:
\begin{eqnarray}\label{eq:corrk0}
&&\langle \phi^{(\geq k-1)+}_{\fl{L^{-k+1}x}} \phi^{(\geq k-1)-}_{\fl{L^{-k+1}y}}\rangle_{N} \nonumber\\
&&\quad = \int d\mu(\Phi^{(\geq 0)}) \Big[ \prod_{z\in \Lambda^{(0)}} U^{(0)}(\Phi^{(\geq 0)}_{z}) \Big]  \phi^{(\geq k-1)+}_{\fl{L^{-k+1}x}} \phi^{(\geq k-1)-}_{\fl{L^{-k+1}y}} \\
&&\quad = \int d\mu(\Phi^{(\geq k-1)}) \Big[ \prod_{ x\in \Lambda^{(k-1)} } U^{(k-1)}(\Phi_{x}^{(\geq k-1)}) \Big] \phi^{(\geq k-1)+}_{\fl{L^{-k+1}x}} \phi^{(\geq k-1)-}_{\fl{L^{-k+1}y}}\nonumber
\end{eqnarray}
where the functions $U^{(k-1)}(\cdot)$ are the outcome of the construction of Section \ref{sec:RGpart} (we again used that $N^{(h)} = 1$ for all $h$, by SUSY). Again by SUSY, see Remark \ref{rem:Nh}:
\begin{equation}\label{eq:SUSYcorr}
\langle (\phi^{(\geq k-1)+}_{\fl{L^{-k+1}x}} \phi^{(\geq k-1)-}_{\fl{L^{-k+1}y}} + \psi^{(\geq k-1)+}_{\fl{L^{-k+1}x}} \psi^{(\geq k-1)-}_{\fl{L^{-k+1}y}})\rangle_{N} = 0\;.
\end{equation}
It therefore suffices to compute the bosonic two-point function, since Eq.~(\ref{eq:SUSYcorr}) together with the previous discussion implies
\begin{equation}\label{eq: bosonic_fermionic_two-point_function}
\langle \phi^{+}_{x} \phi^{-}_{y} \rangle_{N} = - \langle \psi^{+}_{x} \psi^{-}_{y} \rangle_{N}\;.
\end{equation}

\subsection{Integration of the nontrivial scales}
We write:
\begin{equation}
\phi^{(\geq k-1)\pm}_{z} = \frac{1}{L}\phi^{(\geq k)\pm}_{\fl{L^{-1}z}} + A_{ z } \zeta^{(k-1)\pm}_{\phi,\fl{ L^{-1} z} }\;;
\end{equation}
plugging this identity in Eq. (\ref{eq:corrk0}), and using again that the average of odd functions is zero:
\begin{eqnarray}\label{eq:phiphi}
\langle \phi^{(\geq k-1)+}_{\fl{L^{-k+1} x}} \phi^{(\geq k-1)-}_{\fl{L^{-k+1}y}}  \rangle_{N} &=& \frac{1}{L^{2}}\langle \phi^{(\geq k)+}_{\fl{L^{-k}x}} \phi^{(\geq k)-}_{\fl{L^{-k}y}} \rangle_{N}\\&& + A_{\fl{ L^{-k+1} x}} A_{ \fl{L^{-k+1}y} }\langle \zeta^{(k-1)+}_{\phi,\fl{ L^{-k} x} } \zeta^{(k-1)-}_{\phi,\fl{ L^{-k} y} }  \rangle_{N}\nonumber\\
&\equiv& \sum_{h = k}^{N+1} \frac{A_{\fl{ L^{-h+1} x}} A_{ \fl{L^{-h+1}y} }}{L^{2(h-k)}} \langle \zeta^{(h-1)+}_{\phi,\fl{ L^{-h}x}} \zeta^{(h-1)-}_{\phi,\fl{ L^{-h}y}}  \rangle_{N}\;.\nonumber
\end{eqnarray}
Notice that, by definition of scale $k$, $\fl{ L^{-h}x } = \fl{ L^{-h}y }$ for all $h\geq k$. Moreover, the average in the sum does not depend on the location of the fields.
\subsubsection{Integration of the scale $k-1$}
In this section we discuss the integration of the first nontrivial scale, corresponding to the entry $h=k-1$ in the sum in Eq. (\ref{eq:phiphi}). To evaluate $\langle \zeta^{(k-1)+}_{y} \zeta^{(k-1)-}_{y}  \rangle_{N}$, $y \in \Lambda^{(k)}$ we proceed as follows. By definition of effective potential:
\begin{eqnarray}
&&\langle \zeta^{(k-1)+}_{\phi,y} \zeta^{(k-1)-}_{\phi,y}  \rangle_{N} \\&&\qquad = \int d\mu(\Phi^{(\geq k-1)}) \Big[ \prod_{ z\in \Lambda^{(k)}} \prod_{x \in \b^{(k)}_{z}}U^{(k-1)}(\Phi_{x}^{(\geq k-1)}) \Big] \zeta^{(k-1)+}_{\phi,y} \zeta^{(k-1)-}_{\phi,y}\nonumber
\end{eqnarray}
which we rewrite as:
\begin{eqnarray}\label{eq:corrk}
&&\langle \zeta^{(k-1)+}_{\phi,y} \zeta^{(k-1)-}_{\phi,y}  \rangle_{N} \nonumber\\
&&= \int d\mu(\Phi^{(\geq k)}) d\mu(\zeta^{(k-1)}) \Big[ \prod_{ z\in \Lambda^{(k)}} \prod_{x\in \b^{(k)}_{z}} U^{(k-1)}(\Phi_{z}^{(\geq k)}/L + A_{x}\zeta^{(k-1)}_{z }) \Big]\nonumber\\&&\qquad\qquad\qquad\qquad\qquad\qquad\cdot  \zeta^{(k-1)+}_{\phi,y} \zeta^{(k-1)-}_{\phi,y}\nonumber\\
&& = \int d\mu(\Phi^{(\geq k)}) d\mu(\zeta^{(k-1)}) \Big[ \prod_{ \substack{ z\in \Lambda^{(k)} \\ z\neq y}} \prod_{x\in \b^{(k)}_{z}} U^{(k-1)}(\Phi_{z}^{(\geq k)}/L + A_{ x}\zeta^{(k-1)}_{z}) \Big]\nonumber\\&&\qquad\qquad \cdot\Big[ \prod_{x\in \b^{(k)}_{y}} U^{(k-1)}(\Phi_{y}^{(\geq k)}/L + A_{x}\zeta^{(k-1)}_{y})\Big] \zeta^{(k-1)+}_{\phi,y} \zeta^{(k-1)-}_{\phi,y}\nonumber\\
&&\equiv \int d\mu(\Phi^{(\geq k)}) \Big[\prod_{\substack{ x \in \Lambda^{(k)} \\ x\neq y}} U^{(k)}(\Phi^{(\geq k)}_{x})\Big] F_{k-1}^{(k)}(\Phi^{(\geq k)}_{y})\;,
\end{eqnarray}
where we introduced:
\begin{eqnarray}\label{eq:defF}
F^{(k)}_{k-1}(\Phi) &:=& \int d\mu(\zeta)\, [U^{(k-1)}( \Phi/L + \zeta )U^{(k-1)}( \Phi/L - \zeta )]^{\frac{L^{3}}{2}} \zeta^{+}_{\phi} \zeta^{-}_{\phi}\nonumber\\
&=& e^{-\frac{\lambda_{k-1}}{L}(\Phi \cdot \Phi)^{2} - i L\mu_{k-1} (\Phi\cdot \Phi)} \nonumber\\&&\cdot\sum_{n=0,1,2} (\psi\cdot \psi)^{n} L^{-2n} \int d\mu_{\phi}(\zeta_{\phi})\, e^{-\widetilde{V}^{(k-1)}_{\text{b}}(\phi,\zeta_{\phi})} \zeta^{+}_{\phi} \zeta^{-}_{\phi} D^{(k-1)}_{n}(\phi, \zeta_{\phi})\;; \nonumber
\end{eqnarray}
the last step follows from the integration of the fermionic fields, discussed in Sections \ref{sec:seth}, \ref{sec:inductive_step}, and from the definition of the $D^{(k-1)}$ functions, recall Eq. (\ref{eq:Rnh}). We now discuss the integration of the bosonic fluctuation field, which we shall perform via a stationary phase expansion. We have:
\begin{eqnarray}\label{eq:correxp}
&&\sum_{n=0,1,2} (\psi\cdot \psi)^{n} L^{-2n}\int d\mu_{\phi}(\zeta_{\phi})\, e^{-\widetilde V^{(k-1)}_{\text{b}}(\phi,\zeta_{\phi})} \zeta^{+}_{\phi} \zeta^{-}_{\phi} D^{(k-1)}_{n}(\phi, \zeta_{\phi})\\
&&\quad = -i\sum_{n=0,1,2} (\psi\cdot \psi)^{n} L^{-2n} D^{(k-1)}_{n}(\phi, 0) +\sum_{n=0,1,2} (\psi\cdot \psi)^{n} \mathcal{E}_{2}(Y^{(k-1)}_{n}(\phi, \cdot))\;,\nonumber
\end{eqnarray}
where we defined:
\begin{eqnarray}\label{eq:tildeD2}
Y^{(k-1)}_{n}(\phi, \zeta_{\phi}) &:=& L^{-2n} e^{-\widetilde V^{(k-1)}_{\text{b}}(\phi,\zeta_{\phi})} \zeta^{+}_{\phi} \zeta^{-}_{\phi} D^{(k-1)}_{n}(\phi, \zeta_{\phi})\;.
\end{eqnarray}
Plugging the expansion (\ref{eq:correxp}) in (\ref{eq:defF}), we get:
\begin{eqnarray}\label{eq:corrFk+1}
F^{(k)}_{k-1}(\Phi) &=& -i U^{(k-1)}(\Phi/L)^{L^{3}} \\ && + e^{-\frac{\lambda_{k-1}}{L}(\Phi \cdot \Phi)^{2} - i L\mu_{k-1} (\Phi\cdot \Phi)} \sum_{n=0,1,2} (\psi\cdot \psi)^{n} \mathcal{E}_{2}(Y^{(k-1)}_{n}(\phi, \cdot))\nonumber\\
&\equiv& -i U^{(k-1)}(\Phi/L)^{L^{3}} + \widetilde F^{(k)}_{k-1}(\Phi)\label{eq:Ffin}\;,\nonumber
\end{eqnarray}
where we used that, by the definition of the $D^{(k-1)}$ functions:
\begin{equation}
U^{(k-1)}(\Phi/L)^{L^{3}} = e^{-\frac{\lambda_{k-1}}{L}(\Phi \cdot \Phi)^{2} - i L\mu_{k-1} (\Phi\cdot \Phi)} \sum_{n=0,1,2} (\psi\cdot \psi)^{n} L^{-2n} D^{(k-1)}_{n}(\phi, 0)
\end{equation}
and we defined:
\begin{equation}\label{eq:widetildeF}
\widetilde F^{(k)}_{k-1}(\Phi) := e^{-\frac{\lambda_{k-1}}{L}(\Phi \cdot \Phi)^{2} - iL\mu_{k-1} (\Phi\cdot \Phi)} \sum_{n=0,1,2} (\psi\cdot \psi)^{n} \mathcal{E}_{2}(Y^{(k-1)}_{n}(\phi, \cdot))\;.
\end{equation}
Let us now estimate the right-hand side of Eq. (\ref{eq:widetildeF}). We use the bound (\ref{eq:mixedD}), together with Lemma \ref{lem:cau}. We get, recalling the definition (\ref{eq:tildeD2}), for $\|\text{Im}\, \zeta_{\phi}\| \leq |\lambda_{k-1}|^{-\frac{1}{4} + \eps}$, $\|\text{Im}\, (\phi/L \pm \zeta_{\phi})\|\leq |\lambda_{k-1}|^{-\frac{1}{4}}$:
\begin{equation}
|Y^{(k-1)}_{n}(\phi, \zeta_{\phi})| \leq C_{L} |\lambda_{k-1}|^{\frac{n}{2} - \frac{1}{2}} e^{-|\lambda_{k-1}| \frac{L^{3}}{4} \|\zeta_{\phi}\|^{4} + \hat c_{k-1}\frac{|\lambda_{k-1}|}{L} \|\phi\|^{4}}\;.
\end{equation}
The extra factor $|\lambda_{k-1}|^{-\frac{1}{2}}$ is due to the presence of $\zeta^{+}_{\phi} \zeta^{-}_{\phi}$ in the definition of $Y^{(k-1)}_{n}$. Thus, $Y_{n}^{(k-1)}(\phi, \cdot)$ satisfies the bound (\ref{eq:FW}) with $W = |\lambda_{k-1}|^{-\frac{1}{4} + \eps}$:
\begin{equation}
F_{W}(Y^{(k-1)}_{n}(\phi, \cdot)) = K_{L} |\lambda_{k-1}|^{\frac{n}{2} - \frac{3}{2}} e^{\hat c_{k-1}\frac{|\lambda_{k-1}|}{L} \|\phi\|^{4}}\;.
\end{equation}
Hence, by Lemma \ref{lem:cau}:
\begin{equation}\label{eq:Gk}
| \mathcal{E}_{2}(Y^{(k-1)}_{n}(\phi, \cdot))| \leq \widetilde K_{L} |\lambda_{k-1}|^{\frac{n}{2} + \frac{1}{2} - 8\eps} e^{\hat c_{k-1}\frac{|\lambda_{k-1}|}{L} \|\phi\|^{4}}\;.
\end{equation}
\subsubsection{Iterative integration}
Plugging Eq. (\ref{eq:corrFk+1}) into Eq. (\ref{eq:corrk}), we get, by Theorem \ref{prp:SUSY}:
\begin{equation}\label{eq:zetazeta}
\langle \zeta^{(k-1)+}_{y} \zeta^{(k-1)-}_{y}  \rangle_{k-1} = -i + \int d\mu(\Phi^{(\geq k)}) \Big[\prod_{\substack{ x\in \Lambda^{(k)} \\ x\neq y}} U^{(k)}(\Phi^{(\geq k)}_{x})\Big] \widetilde F^{(k)}_{k-1}(\Phi^{(\geq k)}_{y})\;.
\end{equation}
We shall compute the integral in a multiscale fashion. We have, for all $\ell$ such that $k+\ell \leq N$:
\begin{eqnarray}\label{eq:itF}
&&\int d\mu(\Phi^{(\geq k)}) \Big[\prod_{\substack{ x\in \Lambda^{(k)} \\ z\neq y}} U^{(k)}(\Phi^{(\geq k)}_{x})\Big] \widetilde F^{(k)}_{k-1}(\Phi^{(\geq k)}_{y})\nonumber\\
&& = \int d\mu(\Phi^{(\geq k+\ell)}) \Big[\prod_{\substack{ x\in \Lambda^{(k+\ell)} \\ x\neq \fl{L^{-\ell}y}}} U^{(k+\ell)}(\Phi^{(\geq k+\ell)}_{x})\Big] \widetilde F^{(k+\ell)}_{k-1}(\Phi^{(\geq k+\ell)}_{\fl{L^{-\ell}y}})\;,
\end{eqnarray}
%
where, denoting by $ x_{h}(y):=\fl{L^{-h+k}y} \in  \Lambda^{(h)}$ for $h>k-1$:
\begin{eqnarray}\label{eq:widetildeFh+1}
\widetilde F^{(h+1)}_{k-1}(\Phi^{(\geq h+1)}_{x_{h+1}(y)}) &:=& \int d\mu(\zeta)\, \Big[\prod_{\substack{x\in \b^{(h+1)}_{x_{h+1}(y)} \\ x\neq x_{h}(y)}} U^{(h)}(\Phi^{(\geq h+1)}_{x_{h+1}(y)}/L + A_{x} \zeta)\Big] \nonumber\\&&\qquad\cdot\widetilde{F}^{(h)}_{k-1}( \Phi^{(\geq h+1)}_{x_{h+1}(y)}/L + A_{x_{h}(y)}\zeta)\;.
\end{eqnarray}
%
In particular, Eqs. (\ref{eq:zetazeta}), (\ref{eq:itF}) imply:
\begin{equation}\label{eq:zetazeta2}
\langle \zeta^{(k-1)+}_{y} \zeta^{(k-1)-}_{y}  \rangle_{k-1} = -i + \widetilde F^{(N)}_{k-1}(0)\;.
\end{equation}
That is,
\begin{eqnarray}\label{eq:phiphi2}
\langle \phi^{+}_x \phi^{-}_y\rangle_{N} &=& \frac{1}{L^{2(k-1)}} \langle \phi^{(\geq k-1)+}_{\fl{L^{-k+1}x}} \phi^{(\geq k-1)-}_{\fl{L^{-k+1}y}}\rangle_{N} \nonumber\\
&=& \frac{1}{L^{2(k-1)}}\sum_{h = k}^{N+1} \frac{A_{\fl{ L^{-h+1} x}} A_{ \fl{L^{-h+1}y} }}{L^{2(h-k)}} \langle \zeta^{(h-1)+}_{\phi,\fl{ L^{-h}x}} \zeta^{(h-1)-}_{\phi,\fl{ L^{-h}y}}  \rangle_{N}\nonumber\\
&=& \frac{-i}{L^{2(k-1)}} \sum_{h = k}^{N+1} \Big( \frac{A_{\fl{ L^{-h+1} x}} A_{ \fl{L^{-h+1}y} }}{L^{2(h-k)}} + i \widetilde F^{(N)}_{h-1}(0)\Big)\;.
\end{eqnarray}
Let us rewrite the map (\ref{eq:widetildeFh+1}) in a more symmetric way. Recalling that $A_{z} = \pm$ and that $\sum_{x\in \b^{(h+1)}_{y}} A_{x} = 0$:
\begin{eqnarray}\label{eq:indcorrstart}
&&\widetilde F^{(h+1)}_{k-1}(\Phi^{(\geq h+1)}_{x_{h+1}(y)})\nonumber\\
&&\quad = \int d\mu(\zeta) \Big[ U^{(h)}(\Phi^{(\geq h+1)}_{x_{h+1}(y)}/L + \zeta)U^{(h)}(\Phi^{(\geq h+1)}_{x_{h+1}(y)}/L - \zeta) \Big]^{\frac{L^{3}}{2}-1} \nonumber\\&&\qquad \cdot U^{(h)}(\Phi^{(\geq h+1)}_{x_{h+1}(y)}/L - \zeta) \widetilde{F}^{(h)}_{k-1}(\Phi^{(\geq h+1)}_{x_{h+1}(y)}/L + \zeta)\;.
\end{eqnarray}
To prove this equality we assumed that $A_{x_{h}(y)} = 1$. If not, we can reduce the discussion to this case by performing a $\zeta\to -\zeta$ change of variable. Let us now assume inductively that:
\begin{equation}\label{eq:Fhind}
\widetilde F^{(h)}_{k-1}(\Phi) = e^{-\lambda_{h}(\Phi \cdot \Phi)^{2} - i\mu_{h} (\Phi\cdot \Phi)} \sum_{n=0,1,2} (\psi \cdot \psi)^{n} G^{(h)}_{k-1;n}(\phi)\;,
\end{equation}
with $G^{(h)}_{k-1;n}(\phi)$ analytic in $\phi \in \mathbb{S}^{(h)}\cup \mathbb{L}^{(h)}$, and such that:
\begin{equation}\label{eq:Ghind}
| G^{(h)}_{k-1;n}(\phi) | \leq C_{h}\vert \lambda_{k-1} \vert^{\frac{1}{2}- 8\varepsilon} |\lambda_{h}|^{\frac{n}{8}} e^{c_{h} |\lambda_{h}| \|\phi\|^{4}}\;,
\end{equation}
for some $0< C_{h}\leq 2\widetilde{K}_{L}$. These assumptions are true for $h=k$, see Eq. (\ref{eq:Gk}) (there, $C_{k} = \widetilde{K}_{L}$). Our goal will be to show that these bounds propagate to scale $h+1$. 

After integrating the fermionic fluctuation field, we get:
\begin{eqnarray}\label{eq:tildeFh+10}
\widetilde F^{(h+1)}_{k-1}(\Phi) &=& e^{-\frac{\lambda_{h}}{L}(\Phi\cdot \Phi)^{2} - i\mu_{h} L (\Phi\cdot \Phi)} \\&&\cdot \int d\mu_{\phi}(\zeta_{\phi})\, e^{-\widetilde V_{\text{b}}^{(h)}(\phi, \zeta_{\phi})} \sum_{n=0,1,2} (\psi \cdot \psi)^{n} L^{-2n} \Gamma^{(h)}_{k-1;n}(\phi/L, \zeta_{\phi})\;,\nonumber
\end{eqnarray}
where $\Gamma^{(h)}_{k-1;n}(\phi/L, \zeta_{\phi})$ is analytic for $\phi/L\pm \zeta_{\phi} \in \mathbb{S}^{(h)}\cup \mathbb{L}^{(h)}$. The $(\psi\cdot \psi)$-dependence of $\widetilde F^{(h+1)}_{k-1}(\Phi)$ follows from (\ref{eq:indcorrstart}) together with the fact that both $U^{(h)}$, $\widetilde{F}^{(h)}_{k-1}$ depend on $\psi$ either via $(\psi \cdot \psi)$, or via $(\psi \cdot \zeta_{\psi})$, and the fermionic covariance of the Grassmann Gaussian integration $d\mu_{\psi}(\zeta_{\psi})$ is diagonal. The next proposition collects important bounds on the functions $\Gamma^{(h)}_{k-1;n}$.
\begin{proposition}\label{prp:41} Let $\phi/L \pm \zeta_{\phi} \in \mathbb{S}^{(h)} \cup \mathbb{L}^{(h)}$. Then, the following bounds hold true, for $\lambda$ small enough and for a universal constant $C>0$:
\begin{eqnarray}\label{eq:Gammahbd}
&&|\Gamma^{(h)}_{k-1;n}(\phi/L, \zeta_{\phi})| \leq \\
&&C_{h} \hat \beta_{h}(\phi, \zeta_{\phi})^{4} (1 + 24 |\lambda_{h}|^{\frac{1}{16}}) |\lambda_{k-1}|^{\frac{1}{2} - 8\eps} |\lambda_{h}|^{\frac{n}{8}} e^{c_{h} |\lambda_{h}| \frac{L^{3}}{2} (\| \phi + \zeta_{\phi} \|^{4} + c_{h} |\lambda_{h}| \| \phi - \zeta_{\phi} \|^{4})}\;, \nonumber
\end{eqnarray}
with:
\begin{equation}
\hat \beta_{h}(\phi, \zeta_{\phi}) := (1 + |\lambda_{h}|^{\frac{3}{8}} \|\phi\| + |\lambda_{h}|^{\frac{3}{8}} \|\zeta_{\phi}\| + |\lambda_{h}|^{\frac{5}{8}} \|\phi\| \|\zeta_{\phi}\|) (1 + |\lambda_{h}|^{\frac{3}{8}} \| \zeta_{\phi} \|)\;.
\end{equation}
\end{proposition}
\begin{proof}
The proof follows the same strategy of the proofs of Propositions \ref{prop:B}, \ref{prp:C}, \ref{prp:Dh}. Let us set:
\begin{eqnarray}
&&f^{(h)}(\Phi):= \sum_{n} R_{n}^{(h)}(\phi) (\psi\cdot \psi)^{n}\;,\qquad \tilde{f}^{(h)}_{k-1}(\Phi):= \sum_{n} G_{k-1;n}^{(h)}(\phi)(\psi\cdot \psi)^{n}\nonumber\\
&&\qquad\qquad\qquad \tilde g^{(h)}(\Phi,\zeta): = f^{(h)}(\Phi + \zeta)^{\frac{L^{3}}{2}-1} f^{(h)}(\Phi - \zeta)^{\frac{L^{3}}{2}-1}\;.
\end{eqnarray}
In terms of these functions, $\widetilde F^{(h+1)}_{k-1}(\Phi)$ reads:
\begin{eqnarray}\label{eq:tildeFh+1}
&&\widetilde F^{(h+1)}_{k-1}(\Phi) = e^{-\frac{\lambda_{h}}{L}(\Phi\cdot \Phi)^{2} - i\mu_{h} L (\Phi\cdot \Phi)} \\&&\cdot \int d\mu(\zeta)\, e^{- V^{(h)}_{\text{b}}(\Phi,\zeta_{\phi}) - V^{(h)}_{\text{f}}(\Phi,\zeta)} \tilde g^{(h)}(\Phi/L,\zeta) f^{(h)}(\Phi/L - \zeta) \tilde{f}^{(h)}_{k-1}(\Phi/L + \zeta)\;,\nonumber
\end{eqnarray}
recall Eq. (\ref{eq:VbVf}) for the definition of $V^{(h)}_{\text{b}}$, $V^{(h)}_{\text{f}}$. The functions $\Gamma^{(h)}_{k-1;n}$ in Eq. (\ref{eq:tildeFh+1}) arise after the integration of the fermionic fluctuation field. To bound them, we will use $(\kappa, \mathcal{N}, \mathcal{M})$-bounds for the argument of the fermionic integration. 

To begin, the bounds (\ref{eq:Rsmall}), (\ref{eq:Rlarge}), together with Lemma \ref{lem:kappaNexp}, implies that for $\phi - \zeta_{\phi} \in \mathbb{S}^{(h)} \cup \mathbb{L}^{(h)}$ the function $f^{(h)}(\Phi - \zeta)$ satisfies $(\kappa', \mathcal{N}', \mathcal{M}')$-bounds with $\mathcal{M}' = \mathcal{N}'$ and:
\begin{equation}
\kappa' = \delta e^{c_{h} \vert \lambda_{h} \vert \|\phi - \zeta\|^{4}}\;,\qquad \mathcal{N}' = |\lambda_{h}|^{\frac{1}{4}}\;.
\end{equation}
Next, consider $\tilde g^{(h)}(\Phi, \zeta)$. Notice that the function $\tilde g^{(h)}$ is almost identical to the function $g^{(h)}$ appearing in the proof of Proposition \ref{prop:B}, see Eq. (\ref{eq:gphi}), the only difference being that the function $f^{(h)}(\Phi \pm \zeta)^{L^{3}/2}$ in the definition of $g^{(h)}$ is replaced by $f^{(h)}(\Phi \pm \zeta)^{L^{3}/2 - 1}$ in the definition of $\tilde g^{(h)}$. It is then easy to see that for $\phi \pm \zeta_{\phi} \in \mathbb{S}^{(h)} \cup \mathbb{L}^{(h)}$ the function $\tilde g^{(h)}(\Phi)$ satisfies $(\kappa'', \mathcal{N}'', \mathcal{M}'')$-bounds, with $\mathcal{N}'' = \mathcal{M}'' = |\lambda_{h}|^{\frac{1}{8}}$ and $\kappa''$ given by:
\begin{align}\label{eq:123corr}
(i) & \; \,  \kappa'' = 1+ 2 K L^{3} |\lambda_{h}|^{\frac{1}{2}}, & 
\phi \pm \zeta_{\phi} \in \mathbb{S}^{(h)},
\nonumber \\ 
(ii) &  \; \, \kappa'' = K \delta^{\frac{L^{3}}{2}-2} e^{ (\frac{L^{3}}{2} - 1) c_{h} |\lambda_{h}| \| \phi - \zeta_{\phi}\|^{4}},  & \phi + \zeta_{\phi} \in \mathbb{S}^{(h)}, \phi - \zeta_{\phi} \in \mathbb{L}^{(h)},
\nonumber \\
& \;\, \kappa'' = K \delta^{\frac{L^{3}}{2}-2} e^{ (\frac{L^{3}}{2}-1) c_{h} |\lambda_{h}| \| \phi + \zeta_{\phi}\|^{4}}  & \phi - \zeta_{\phi} \in \mathbb{S}^{(h)}, \phi + \zeta_{\phi} \in \mathbb{L}^{(h)},
\nonumber \\
(iii) & \; \, \kappa'' = K \delta^{L^{3}-2} e^{(\frac{L^{3}}{2} - 1) c_{h} |\lambda_{h}|( \| \phi + \zeta_{\phi}\|^{4} + \| \phi - \zeta_{\phi}\|^{4})}, &
\phi \pm \zeta_{\phi} \in \mathbb{L}^{(h)} \;.
\end{align}
These statements are proven as in the proof of Proposition \ref{prop:B}. The values of $\mathcal{N}''$ and $\mathcal{M}''$ correspond to the worst case of the corresponding values in the proof of Proposition \ref{prop:B}, taking $\lambda$ small enough to get rid of multiplicative constants. Item $(i)$ in (\ref{eq:123corr}) follows from (\ref{eq:gh2}). Item $(ii)$ follows from (\ref{eq:kappaii}). Item $(iii)$ follows from (\ref{eq:kappaiii}).

Consider now $\tilde{f}^{(h)}_{k-1}$. The bound (\ref{eq:Ghind}) together with Lemma \ref{lem:kappaNexp} implies that, for $\phi + \zeta_{\phi} \in \mathbb{S}^{(h)} \cup \mathbb{L}^{(h)}$, the function $\tilde f^{(h)}_{k-1}(\Phi + \zeta)$ satisfies $(\kappa''', \mathcal{N}''', \mathcal{M}''')$-bounds with $\mathcal{N}''' = \mathcal{M}'''$ and:
\begin{equation}
\kappa''' = C_{h} |\lambda_{k-1}|^{\frac{1}{2} - 8\eps} e^{c_{h} |\lambda_{h}| \| \phi + \zeta_{\phi} \|^{4}}\;,\qquad \mathcal{N}''' = |\lambda_{h}|^{\frac{1}{16}}\;.
\end{equation}
Next, consider the $\zeta_{\psi}$-dependent part of the integrand in Eq. (\ref{eq:tildeFh+1}):
\begin{equation}\label{eq:fullprod}
e^{- V^{(h)}_{\text{f}}(\Phi,\zeta)} \tilde g^{(h)}(\Phi/L,\zeta) f^{(h)}(\Phi/L - \zeta) \tilde{f}^{(h)}_{k-1}(\Phi/L + \zeta)\;.
\end{equation}
By Lemma \ref{lem:kappaNprod}, the product $\tilde g^{(h)}(\Phi/L,\zeta) f^{(h)}(\Phi/L - \zeta) \tilde{f}^{(h)}_{k-1}(\Phi/L + \zeta)$ appearing at the argument of the integral in (\ref{eq:tildeFh+1}) satisfies $(\kappa_{1}, \mathcal{N}_{1}, \mathcal{M}_{1})$-bounds, with
\begin{eqnarray}\label{eq:fullprod2}
\kappa_{1} &=& \kappa'\kappa''\kappa'''\leq (1 + |\lambda_{h}|^{\frac{1}{4}})C_{h} |\lambda_{k-1}|^{\frac{1}{2} - 8\eps} e^{c_{h} |\lambda_{h}| \frac{L^{3}}{2} (\| \phi + \zeta_{\phi} \|^{4} + c_{h} |\lambda_{h}| \| \phi - \zeta_{\phi} \|^{4})}\nonumber\\
\mathcal{N}_{1} &=&  \mathcal{N}' + \mathcal{N}'' + \mathcal{N}''' \leq |\lambda_{h}|^{\frac{1}{16}}(1 + 2 |\lambda_{h}|^{\frac{1}{16}})\;,\qquad \mathcal{M}_{1} = \mathcal{N}_{1}\;.
\end{eqnarray}
Consider now the expression in (\ref{eq:fullprod}). Recall the definition (\ref{eq:VbVf}) for $V^{(h)}_{\text{f}}(\Phi,\zeta)$ defined in Eq. (\ref{eq:VbVf}), and the $(\kappa, \mathcal{N}, \mathcal{M})$-bounds of (\ref{eq:kappaN2}). Lemma \ref{lem:kappaNprod}, and (\ref{eq:fullprod2}), (\ref{eq:kappaN2}) imply that (\ref{eq:fullprod}) satisfies $(\kappa_{2}, \mathcal{N}_{2}, \mathcal{M}_{2})$-bounds with:
\begin{eqnarray}\label{eq:hatkappa}
\kappa_{2} &=& \kappa_{1}\nonumber\\
\mathcal{N}_{2} &\leq& |\lambda_{h}|^{\frac{1}{16}}(1 + 3 |\lambda_{h}|^{\frac{1}{16}})\nonumber\\
\mathcal{M}_{2} &\leq& |\lambda_{h}|^{\frac{1}{16}} \beta_{h}(\phi, \zeta_{\phi})
\end{eqnarray}
with:
\begin{equation}
\beta_{h}(\phi, \zeta_{\phi}) = 1 + |\lambda_{h}|^{\frac{3}{8}} \|\phi\| + |\lambda_{h}|^{\frac{3}{8}} \|\zeta_{\phi}\| + |\lambda_{h}|^{\frac{5}{8}} \|\phi\| \|\zeta_{\phi}\|\;.
\end{equation}
We shall now compute the $(\kappa, \mathcal{N})$-bounds for the function obtained after integrating (\ref{eq:fullprod}) in $\zeta_{\psi}$:
\begin{equation}
\tilde h^{(h)}(\Phi, \zeta_{\phi}) :=  \int d \mu _{\psi}(\zeta_{\psi})\, e^{-V^{(h)}_{\text{f}}(\Phi,\zeta)} \tilde g^{(h)}(\Phi/L,\zeta) f^{(h)}(\Phi/L - \zeta) \tilde{f}^{(h)}_{k-1}(\Phi/L + \zeta)\;.
\end{equation}
Using Lemma \ref{lem:kappaNint} together with (\ref{eq:hatkappa}) we have that $\tilde h^{(h)}(\Phi, \zeta_{\phi})$ satisfies $(\kappa_{3}, \mathcal{N}_{3})$-bounds with:
\begin{eqnarray}
\kappa_{3} &=& \kappa_{2} (1 + 12 \mathcal{M}_{2}^{2} + 2 \mathcal{M}_{2}^{4})\nonumber\\
&\leq& \kappa_{1} (1 + 24 |\lambda_{h}|^{\frac{1}{8}}) \beta_{h}(\phi, \zeta_{\phi})^{4}\nonumber\\
\mathcal{N}_{3} &=& \mathcal{N}_{2}\;. 
\end{eqnarray}
To conclude, we write $V_{\text{b}}^{(h)}(\Phi, \zeta_{\phi}) = \widetilde{V}_{\text{b}}^{(h)}(\phi, \zeta_{\phi}) + \hat V_{\text{b}}^{(h)}(\Phi, \zeta_{\phi})$, with $\hat V_{\text{b}}^{(h)}(\Phi, \zeta_{\phi}) = 2\lambda_{h} L (\psi \cdot \psi) (\zeta_{\phi}\cdot \zeta_{\phi})$, recall (\ref{eq:tildeV0}). The final claim will follow from estimating the coefficients in $(\psi \cdot \psi)$ of:
\begin{equation}\label{eq:prodfin}
e^{-\hat V^{(h)}_{\text{b}}(\Phi, \zeta_{\phi})} \tilde h^{(h)}(\Phi, \zeta_{\phi})\;.
\end{equation}
To do so, notice that the function $e^{-\hat V^{(h)}_{\text{b}}(\Phi, \zeta_{\phi})}$ satisfies $(\kappa_{4}, \mathcal{N}_{4})$-bounds, with:
\begin{equation}
\kappa_{4} = 1\;,\qquad \mathcal{N}_{4} = 2^{\frac{1}{2}} L^{\frac{1}{2}} |\lambda_{h}|^{\frac{1}{2}} \|\zeta_{\phi}\|\;.
\end{equation}
Therefore, the function (\ref{eq:prodfin}) satisfies $(\kappa_{5}, \mathcal{N}_{5})$-bounds with:
\begin{eqnarray}
\kappa_{5} &=& \kappa_{4} \kappa_{3} \\ &\leq& C_{h} \beta_{h}(\phi, \zeta_{\phi})^{4} (1 + 48 |\lambda_{h}|^{\frac{1}{8}}) |\lambda_{k-1}|^{\frac{1}{2} - 8\eps} e^{c_{h} |\lambda_{h}| \frac{L^{3}}{2} (\| \phi + \zeta_{\phi} \|^{4} + c_{h} |\lambda_{h}| \| \phi - \zeta_{\phi} \|^{4})}\nonumber\\
\mathcal{N}_{5} &=& \mathcal{N}_{4} + \mathcal{N}_{3}\nonumber\\
&\leq& |\lambda_{h}|^{\frac{1}{16}}(1 + 3 |\lambda_{h}|^{\frac{1}{16}}) + 2^{\frac{1}{2}} L^{\frac{1}{2}} |\lambda_{h}|^{\frac{1}{2}} \|\zeta_{\phi}\|\nonumber\\
&\leq& |\lambda_{h}|^{\frac{1}{16}} ( 1+ 3|\lambda_{h}|^{\frac{1}{16}}) (1 + |\lambda_{h}|^{\frac{3}{8}} \| \zeta_{\phi} \| )\;.
\end{eqnarray}
Finally, we get rid of factor $( 1+ 3|\lambda_{h}|^{\frac{1}{16}})(1 + |\lambda_{h}|^{\frac{3}{8}} \| \zeta_{\phi} \| )$ in $\mathcal{N}_{5}$ by using that if (\ref{eq:prodfin}) satisfies $(\kappa_{5}, \mathcal{N}_{5})$ then it also satisfies $(\kappa, \mathcal{N})$ bounds with:
\begin{eqnarray}
\kappa &=& \kappa_{5} ( 1+ 3|\lambda_{h}|^{\frac{1}{16}})^{4} (1 + |\lambda_{h}|^{\frac{3}{8}} \| \zeta_{\phi} \| )^{4}\;,\nonumber\\
&\leq&  C_{h} \hat \beta_{h}(\phi, \zeta_{\phi})^{4} (1 + 24 |\lambda_{h}|^{\frac{1}{16}}) |\lambda_{k-1}|^{\frac{1}{2} - 8\eps} e^{c_{h} |\lambda_{h}| \frac{L^{3}}{2} (\| \phi + \zeta_{\phi} \|^{4} + c_{h} |\lambda_{h}| \| \phi - \zeta_{\phi} \|^{4})} \nonumber\\
\mathcal{N} &=& |\lambda_{h}|^{\frac{1}{16}}\;.
\end{eqnarray}
with $\hat \beta_{h}(\phi, \zeta_{\phi}) = \beta_{h}(\phi, \zeta_{\phi}) (1 + |\lambda_{h}|^{\frac{3}{8}} \| \zeta_{\phi} \|)$. This concludes the proof of Proposition \ref{prp:41}.
\end{proof}
To complete the evaluation of $\widetilde{F}^{(h+1)}_{k-1}(\Phi)$, we are now left with the integration of the bosonic fluctuation field in Eq. (\ref{eq:tildeFh+10}). By stationary phase expansion, Lemma \ref{lem:osc}:
\begin{eqnarray}\label{eq:tildeFh+12}
\widetilde F^{(h+1)}_{k-1}(\Phi) &=& e^{-\frac{\lambda_{h}}{L}(\Phi\cdot \Phi)^{2} - i\mu_{h} L (\Phi\cdot \Phi)}\sum_{n=0,1,2} (\psi \cdot \psi)^{n} L^{-2n} \Gamma^{(h)}_{k-1;n}(\phi/L, 0) \nonumber\\
&& + e^{-\frac{\lambda_{h}}{L}(\Phi\cdot \Phi)^{2} - i\mu_{h} L (\Phi\cdot \Phi)} \sum_{n=0,1,2} (\psi \cdot \psi)^{n} \mathcal{E}_{1}(  Y^{(h)}_{k-1;n}(\phi/L, \cdot) )\;,\nonumber
\end{eqnarray}
where:
\begin{equation}
Y^{(h)}_{k-1;n}(\phi/L, \zeta_{\phi}) := L^{-2n} e^{-\widetilde V_{\text{b}}^{(h)}(\phi, \zeta_{\phi})}\Gamma^{(h)}_{k-1;n}(\phi/L, \zeta_{\phi})\;.
\end{equation}
Consider the first term in Eq. (\ref{eq:tildeFh+12}). We have, thanks to Proposition \ref{prp:41}, for $\phi/L \in \mathbb{S}^{(h)} \cup \mathbb{L}^{(h)}$:
\begin{eqnarray}\label{eq:Gh0}
&&|L^{-2n} \Gamma^{(h)}_{k-1;n}(\phi/L, 0)| \\&&\qquad\qquad\qquad\; 
\leq L^{-2n} C_{h} \hat \beta_{h}(\phi, 0)^{4} (1 + 24 |\lambda_{h}|^{\frac{1}{16}}) |\lambda_{k-1}|^{\frac{1}{2} - 8\eps} |\lambda_{h}|^{\frac{n}{8}} e^{c_{h} \frac{|\lambda_{h}|}{L} \|\phi\|^{4}}\nonumber\\
&&\qquad\qquad\qquad\; \leq L^{-2n} C_{h} (1 + 48 |\lambda_{h}|^{\frac{1}{16}})|\lambda_{k-1}|^{\frac{1}{2} - 8\eps} |\lambda_{h}|^{\frac{n}{8}} e^{(c_{h} + |\lambda_{h}|^{\frac{1}{4}}) \frac{|\lambda_{h}|}{L} \|\phi\|^{4}}\;,\nonumber
\end{eqnarray}
where in the last step we used that:
\begin{eqnarray}
\hat \beta_{h}(\phi, 0) &=& 1 + |\lambda_{h}|^{\frac{3}{8}} \|\phi \| \nonumber\\
&\leq& e^{|\lambda_{h}|^{\frac{3}{8}} \|\phi\|} \nonumber\\
&\leq& e^{C L |\lambda_{h}|^{\frac{1}{4}}}  e^{|\lambda_{h}|^{\frac{5}{4}} \frac{1}{4 L}\|\phi\|^{4}}
\end{eqnarray}
and we took into account the first exponential by $(1 + 24 |\lambda_{h}|^{\frac{1}{16}}) e^{4 C L |\lambda_{h}|^{\frac{1}{4}}}  \leq (1 + 48 |\lambda_{h}|^{\frac{1}{16}})$. Consider now the remainder term in the stationary phase expansion. In order to bound the error, we need the analogue of (\ref{eq:mixedD}) for the argument of $\mathcal{E}_{1}$. We shall rely on Proposition \ref{prp:mixed}, together with the estimate:
\begin{eqnarray}\label{eq:hatbeta}
\hat \beta_{h}(\phi, \zeta_{\phi}) &\leq& (1 + CL|\lambda_{h}|^{\frac{1}{4}} + C |\lambda_{h}|^{\frac{5}{4}} \|\phi\|^{4} + C |\lambda_{h}|^{\frac{5}{4}} \|\zeta_\phi\|^{4}) (1 + |\lambda_{h}|^{\frac{3}{8}} \| \zeta_{\phi} \|)\nonumber\\
&\leq& e^{2 CL|\lambda_{h}|^{\frac{1}{4}}} e^{|\lambda_{h}|^{\frac{5}{4}} ( \frac{1}{4L} + C) \|\phi\|^{4}} e^{ C |\lambda_{h}|^{\frac{5}{4}} \|\zeta_{\phi}\|^{4}}\;. 
\end{eqnarray}
Therefore, Eq. (\ref{eq:hatbeta}) and Propositions \ref{prp:41}, \ref{prp:mixed} imply, for $\|\text{Im}\, \zeta_{\phi}\| \leq |\lambda_{h}|^{-\frac{1}{4} + 2\eps}$ and $\phi/L \pm \zeta_{\phi} \in \mathbb{S}^{(h)} \cup \mathbb{L}^{(h)}$:
\begin{eqnarray}
&&| Y^{(h)}_{k-1;n}(\phi/L, \zeta_{\phi}) | \nonumber\\
&& \qquad \leq C_{L} C_{h} |\lambda_{k-1}|^{\frac{1}{2} - 8\eps} |\lambda_{h}|^{\frac{n}{8}} e^{-|\lambda_{h}| \frac{L^{3}}{4} \|\zeta_{\phi}\|^{4} + (c_{h} + |\lambda_{h}|^{8\varepsilon} + |\lambda_{h}|^{\frac{1}{8}})\frac{|\lambda_{h}|}{L} \|\phi\|^{4}}\nonumber\\
&&\qquad  \leq C_{L} C_{h} |\lambda_{k-1}|^{\frac{1}{2} - 8\eps} |\lambda_{h}|^{\frac{n}{8}} e^{-|\lambda_{h}| \frac{L^{3}}{4} \|\zeta_{\phi}\|^{4} + (\hat c_{h} - \frac{1}{2} |\lambda_{h}|^{4\eps})|\frac{|\lambda_{h}|}{L} \|\phi\|^{4}}
\end{eqnarray}
where we used that, for $\lambda$ small enough, $c_{h} + |\lambda_{h}|^{8\varepsilon} + |\lambda_{h}|^{\frac{1}{8}}\leq \hat c_{h} - (1/2) |\lambda_{h}|^{4\eps}$, since $\hat c_{h} = c_{h} + |\lambda_{h}|^{\frac{1}{2}} + |\lambda_{h}|^{4\varepsilon}$, recall the definition after Eq. (\ref{eq:mixedD}). The factor $-|\lambda_{h}|^{4\eps}$ will be useful later on. Hence, by Lemma \ref{lem:cau}:
\begin{equation}\label{eq:E1Gammah}
|\mathcal{E}_{1}(  Y^{(h)}_{k-1;n}(\phi/L, \cdot) )|\leq K_{L} C_{h} |\lambda_{k-1}|^{\frac{1}{2} - 8\eps} |\lambda_{h}|^{\frac{n}{8} + \frac{1}{2} - 12\eps}e^{(\hat c_{h} - \frac{1}{2} |\lambda_{h}|^{4\eps}) \frac{|\lambda_{h}|}{L} \|\phi\|^{4}}\;.
\end{equation}
All in all, putting together Eqs. (\ref{eq:tildeFh+12}), (\ref{eq:Gh0}), (\ref{eq:E1Gammah}), we get:
\begin{equation}\label{eq:almost}
\widetilde{F}^{(h+1)}_{k-1}(\Phi) = e^{-\frac{\lambda_{h}}{L}(\Phi\cdot \Phi)^{2} - i\mu_{h} L (\Phi\cdot \Phi)} \sum_{n=0,1,2} (\psi \cdot \psi)^{n} \widetilde G^{(h+1)}_{k-1;n}(\phi)\;,
\end{equation}
for some new functions $\widetilde G^{(h+1)}_{k-1;n}(\phi)$ analytic in $\phi \in \mathbb{S}^{(h+1)} \cup \mathbb{L}^{(h+1)}$, satisfying the bound:
\begin{equation}\label{eq:tildeGh+1}
|\widetilde G^{(h+1)}_{k-1;n}(\phi)| \leq
C_{h} (1 + 48 |\lambda_{h}|^{\frac{1}{16}})|\lambda_{k-1}|^{\frac{1}{2} - 8\eps} |\lambda_{h+1}|^{\frac{n}{8}} e^{(\hat c_{h} - \frac{1}{2} |\lambda_{h}|^{4\eps}) \frac{|\lambda_{h}|}{L} \|\phi\|^{4}} \;.
\end{equation}
Notice that we used the factor $L^{-2n}$ in Eq. (\ref{eq:Gh0}) to update the prefactor $|\lambda_{h}|^{\frac{n}{8}}$ to $|\lambda_{h+1}|^{\frac{n}{8}}$. Let us now update the running coupling constant appearing in the explicit exponential prefactor in Eq. (\ref{eq:almost}). We write:
\begin{equation}
\begin{split}
\widetilde{F}^{(h+1)}_{k-1}(\Phi) &  = e^{-\frac{\lambda_{h}}{L}(\Phi\cdot \Phi)^{2} - i\mu_{h} L (\Phi\cdot \Phi)} \sum_{n=0,1,2} (\psi \cdot \psi)^{n} \widetilde G^{(h+1)}_{k-1;n}(\phi)
\\
& \equiv  e^{-\lambda_{h+1}(\Phi\cdot \Phi)^{2} - i\mu_{h+1} (\Phi\cdot \Phi)} G^{(h+1)}_{k-1}(\Phi) 
\;,
\\
G^{(h+1)}_{k-1}(\Phi) &:= e^{+\beta^{(h)}_{4}(\Phi\cdot \Phi)^{2} + i\beta^{(h)}_{2} (\Phi\cdot \Phi)}\sum_{n=0,1,2} (\psi \cdot \psi)^{n} \widetilde G^{(h+1)}_{k-1;n}(\phi)
\\
& =: \sum_{n=0,1,2} (\psi \cdot \psi)^{n} G^{(h+1)}_{k-1;n}(\phi)\;,
\end{split}
\end{equation}
where in the last step we expanded the overall exponential as a polynomial in $(\psi\cdot \psi)$, and collected terms of the same powers. We shall now prove bounds for the new functions $G^{(h+1)}_{k-1;n}(\phi)$, and check the inductive assumption (\ref{eq:Ghind}). We notice that the function $e^{\beta^{(h)}_{4}(\Phi\cdot \Phi)^{2} + i\beta^{(h)}_{2} (\Phi\cdot \Phi)}$ satisfies $(\kappa,\mathcal{N})$-bounds with:
\begin{equation}\label{eq:kappaNfin}
\begin{split}
\kappa & = e^{|\beta_{2}^{(h)}| \| \phi\|^{2} + |\beta_{4}^{(h)}| \| \phi\|^{4}} \leq (1 + 2 C^{2} L^{3} |\lambda_{h}|^{\frac{1}{2}}) \, e^{2C|\lambda_{h}|^{\frac{3}{2}} \frac{\| \phi \|^{4}}{L}}
\;,
\\
\mathcal{N} & =  |\beta_{2}^{(h)}|^{\frac{1}{2}} + 2^{\frac{1}{2}}|\beta_{4}^{(h)}|^{\frac{1}{2}}  \| \phi\| + |\beta_{4}^{(h)}|^{\frac{1}{4}} 
\leq |\lambda_{h+1}|^{\frac{1}{4}} \, e^{C |\lambda_{h}|^{\frac{3}{2}} \frac{\| \phi \|^{4}}{L}} \;,
\end{split}
\end{equation}
where we used the bounds $|\beta^{(h)}_{2}| \leq CL|\lambda_{h}|$ and $|\beta^{(h)}_{4}| \leq CL^{-1}|\lambda_{h}|^{\frac{3}{2}}$ (see after Eqs. \eqref{eq:beta_functions_scale_h}) and Lemma \ref{lem:kappaNprod}. All in all, Eqs. (\ref{eq:kappaNfin}), together with the bound (\ref{eq:tildeGh+1}) give, again using Lemma \ref{lem:kappaNprod}:
\begin{eqnarray}
|G^{(h+1)}_{k-1;n}(\phi)| &\leq&
C_{h} (1 + K |\lambda_{h}|^{\frac{1}{16}})|\lambda_{k-1}|^{\frac{1}{2} - 8\eps} |\lambda_{h+1}|^{\frac{n}{8}} e^{(c_{h} - \frac{1}{4} |\lambda_{h}|^{\eps}) \frac{|\lambda_{h}|}{L} \|\phi\|^{4}} \nonumber\\
&\leq& C_{h} (1 + K |\lambda_{h}|^{\frac{1}{16}})|\lambda_{k-1}|^{\frac{1}{2} - 8\eps} |\lambda_{h+1}|^{\frac{n}{8}} e^{c_{h+1} |\lambda_{h+1}| \|\phi\|^{4}} 
\end{eqnarray}
where $c_{h+1} = \hat c_{h} + |\lambda_{h}|^{\frac{1}{2}} + |\lambda_{h}|$, see Eqs. (\ref{eq:Rlarge}), (\ref{eq:bdgenla}), and we used that $\frac{|\lambda_{h}|}{L} \leq |\lambda_{h+1}| (1 + 2 C |\lambda_{h}|^{\frac{1}{2}})$. Also, by construction, the new functions $G^{(h+1)}_{k-1;n}(\phi)$ are analytic in $\phi \in \mathbb{S}^{(h+1)} \cup \mathbb{L}^{(h+1)}$. In conclusion, the inductive assumption Eq. (\ref{eq:Ghind}) is true on scale $h+1$, with:
\begin{eqnarray}
C_{h+1} &=& C_{h} (1 + K | \lambda_{h}|^{\frac{1}{16}}) \nonumber\\
&\leq& C_{k} \prod_{j=k}^{N} (1 + K | \lambda_{j}|^{\frac{1}{16}})\nonumber\\
&\leq& 2C_{k}\;.
\end{eqnarray}
\subsubsection{Conclusion}
We are now ready to compute the two-point correlation function, recall Eq. (\ref{eq:phiphi2}). We have:
\begin{eqnarray}
\langle \phi^{+}_{x,\sigma} \phi^{-}_{y,\sigma'} \rangle &=& \frac{-i\delta_{\sigma, \sigma'}}{L^{2(k-1)}} \Big(\sum_{h=k}^{N+1} \frac{A_{\fl{ L^{-h+1} x}} A_{\fl{ L^{-h+1}y }}}{L^{2(h-k)}} + \mathcal{E}_{N}(x,y)\Big) \nonumber\\
\mathcal{E}_{N}(x,y) &=& i \sum_{h=k}^{N+1} \widetilde F^{(N)}_{h-1}(0)\nonumber\\
| \mathcal{E}_{N}(x,y) | &\leq& \sum_{h=k}^{N+1} 2\widetilde{K}_{L} \vert \lambda_{h-1} \vert^{\frac{1}{2} - 8\varepsilon} \leq K_{L} \Big(\frac{\lambda}{L^{2(k-1)}}\Big)^{\frac{1}{2} - 8\eps}\;,
\end{eqnarray}
where the last bound follows from $\widetilde{F}^{(N)}_{k-1}(0) = G^{(N)}_{k-1;0}(0)$ together with (recall (\ref{eq:Ghind})):
\begin{equation}\label{eq:FNfin}
|G^{(N)}_{k-1;0}(0)| \leq C_{N} \vert \lambda_{k} \vert^{\frac{1}{2} - 8\varepsilon}\leq 2\widetilde{K}_{L} \vert \lambda_{k-1} \vert^{\frac{1}{2} - 8\varepsilon}\;.
\end{equation}
This concludes the proof of Theorem \ref{thm:main}.\qed
\medskip

\noindent{\bf Acknowledgements.} Our work has been supported by the Swiss National Science Foundation via the grant ``Mathematical Aspects of Many-Body Quantum Systems''. M.P. and L.F. acknowledge financial support from the ERC Starting Grant MaMBoQ, grant agreement n.802901. We thank the Isaac Newton Institute (Cambridge) and the organizers of the INI thematic period ``Scaling limits, rough paths, quantum field theory'', where part of this paper was written. We thank the anonymous referees for useful comments on a previous version of the paper.

\appendix
\section{The flow of the chemical potential}\label{app:mu}

In this section we shall control the flow of the chemical potential $\mu_{h}$. Before starting, it is important to recall that the induction of Section \ref{subsec: General_integration_step}, that allowed to contruct the effective potential on all scales, works {\it provided} the sequence of chemical potentials $\{ \mu_{k} \}_{k=0}^{N}$ is bounded as $|\mu_{k}| \leq C|\lambda_{k}|$. In particular, it is important to notice that the constant $C$ can be chosen arbitrarily large, provided $\lambda$ is small enough. For these sequences of chemical potentials, the effective quartic coupling behaves as desired, namely $| \lambda_{k} - \frac{\lambda}{L^{k}} | \leq \frac{K \lambda^{\frac{3}{2}}}{L^{k}}$, with $K$ dependent on $C$, in general.

 Here we shall show that there exists a unique choice of $\mu$ such that indeed the sequence $\{\mu_{k}\}$ satisfies the desired bounds. This is the content of the next proposition.
\begin{proposition}\label{prp:count} For $\overline{C} > 0$ large enough, there exists a unique $\mu \in \mathbb{C}$, $\mu \equiv \mu(\lambda)$, $|\mu(\lambda)|\leq 2\overline{C} \lambda$ such that $|\mu_{h}|\leq 2\overline C|\lambda_{h}|$ for all $h$.
\end{proposition}
\begin{proof} Recall that we are in the context of Theorem~\ref{thm:effective_potential_flow} and that we shall prove the statement by induction. To begin, we shall find a more precise estimate on the beta function of the chemical potential. Recall the flow equation for the chemical potential \eqref{eq:beta_functions_scale_h} and its beta function~\eqref{eq:Ebetah}, \eqref{eq:beta_functions_scale_h}:
\begin{equation}\label{eq:flowmu}
\begin{split}
& \mu_{h+1}  = L\mu_{h} + \beta_{2}^{(h)}\;,
\\
& \beta_{2}^{(h)} :=i \gamma_{2}^{(h)} \equiv \frac{i}{2} \, \partial^{2}_{\| \phi \|} E_{0}^{(h)}(0) \;.
\end{split}
\end{equation}
The function $E^{(h)}_{0}(\phi)$ can be computed via a stationary phase expansion, compare with Eq.~\eqref{eq:Ehn}:
\begin{eqnarray*}
E^{(h)}_{0}(\phi) &=& D^{(h)}_{0}(\phi/L, 0) + d_{1}\Big(\Delta e^{-\widetilde V^{(h)}_{\text{b}}(\phi, \cdot)}\Big)(0) D^{(h)}_{0}(\phi/L, 0)\nonumber
\\
&& + d_{1}\Big(\Delta  D^{(h)}_{0}(\phi/L, \cdot)\Big)(0) + \mathcal{E}_{2}\Big(e^{-\widetilde V_{\text{b}}^{(h)}(\phi, \cdot )} D^{(h)}_{0}(\phi/L, \cdot)\Big)\;;
\end{eqnarray*}
the various terms admit the following bounds:
\begin{equation}
\begin{split}
& |D^{(h)}_{0}(\phi/L, 0) -1 | \leq K L^{3} |\lambda_{h}|^{\frac{1}{2}}\;, 
\qquad 
\Big|\Big(\Delta D^{(h)}_{0}(\phi/L, \cdot)\Big)(0) \Big| \leq C_{L} |\lambda_{h}|\;,
\\
& 
\qquad \qquad 
\Big|\mathcal{E}_{2}\Big(e^{-\widetilde V_{\text{b}}^{(h)}(\phi, \cdot )} D^{(h)}_{0}(\phi/L, \cdot)\Big)\Big|\leq K_{L} |\lambda_{h}|^{ 1 - 8\varepsilon}\;.
\end{split}
\end{equation}
In order to improve the bounds we already obtained on the beta function, we furthermore notice that $\Big(\Delta e^{-\widetilde V^{(h)}_{\text{b}}(\phi, \cdot)}\Big)(0) = 24 \lambda_{h} L (\phi \cdot \phi) + 8 i \mu _{h} L$; hence, provided that $|\mu_{j}|\leq 2\overline C|\lambda_{j}|$, for all $j\leq h$,  the following expression is attained:
\begin{equation}
\beta_{2}^{(h)} = 24 d_{1} i \lambda_{h} L + \frac{i}{2} \partial^{2}_{\|\phi \|} D_{0}^{(h)}(\phi/L,0) +
\tilde{\beta}_{2}^{(h)}\;,
\end{equation}
with
\begin{equation}
\big|\partial^{2}_{\|\phi \|} D_{0}^{(h)}(\phi/L,0) \big| \leq  K L |\lambda_{h}|,
\qquad
|\tilde{\beta}_{2}^{(h)}| \leq \tilde{K}_{L}  |\lambda_{h}|^{ \frac{3}{2} - 8\varepsilon}\;,
\end{equation}
where the first inequality follows from Cauchy estimates in the ball of radius $R = L |\lambda_{h}|^{-\frac{1}{4}}$, and $\tilde K_{L} \equiv \tilde K_{L}(\overline{C})$. Therefore, for $\lambda$ is small enough: 
\begin{equation}
|\beta_{2}^{(h)}| \leq K L |\lambda_{h}|\;,
\end{equation}
for a for a universal constant $K > 0$, which does not depend on $\overline{C}$. In particular, by taking $\overline{C}$ large enough, we have:
\begin{equation}\label{eq:CCbar}
L(\overline{C} - K) > 4\overline{C} L^{-1}\;.
\end{equation}
\noindent{\underline{\it Existence.}} We shall show that there exists a solution $\{ \mu_{h} \}_{h=0}^{\infty}$ of Eq. (\ref{eq:flowmu}) with the desired properties. Later, we will comment on uniqueness. Our discussion closely follows \cite{GK}, see also \cite{BBS3}. Let $\mu \equiv \mu_{0} \in \{z \in \mathbb{C} \;| \;|z| < 2 \overline{C} |\lambda_0|\}=: I_{0}$. By construction, $\beta_{2}^{(0)}$ is a continuous function of $\mu_{0}\in I_{0}$, and $|\beta^{(0)}_{2}| \leq K L|\lambda_{0}|$. Thus, Eq. (\ref{eq:flowmu}) implies that $I_{0}\ni \mu_{0}\mapsto \mu_{1}(\mu_{0})$ is continuous, and that:
\begin{equation}
\mu_{1}(I_{0}) \supset 
\{|z| < 2\overline{C} L |\lambda_{0}| - K L|\lambda_{0}|\}
\supset \;\{|z| < 2 \overline{C} |\lambda_1|\},
\end{equation}
where the last step follows from Eq. (\ref{eq:CCbar}). Thus, by continuity there exists $I_{1}\subset I_{0}$ such that:
\begin{equation}
\{|z| < \overline{C} |\lambda_1|\} \subset \mu_{1}(I_{1}) \subset \{|z| < 2 \overline{C} |\lambda_1|\}\;.
\end{equation}
This shows in particular that, for $\mu_{0} \in I_{1}$, $|\mu_{1}|\leq 2\overline{C}|\lambda_{1}|$. Now, suppose inductively that there exists $I_{h} \subset I_{0}$ such that $I_{h}\ni \mu_{0}\mapsto \mu_{k}(\mu_{0})$ is continuous for all $k\leq h$, and that:
\begin{equation}\label{eq:muindk}
\{|z| < \overline{C} |\lambda_k|\} \subset \mu_{k}(I_{h}) \subset \{|z| < 2 \overline{C} |\lambda_k|\}\;,\qquad \forall k\leq h\;.
\end{equation}
In particular, Eq. (\ref{eq:muindk}) implies that $|\mu_{k}|\leq 2\overline{C}| \lambda_{k}|$ for all $k\leq h$. These assumptions are true for $h=1$, as we just proved.

Let us check the inductive assumptions on scale $h+1$. By the RG construction $|\beta^{(h)}_{2}|\leq K L|\lambda_{h}|$. Also, $\beta^{(h)}_{2}$ is continuous in $\mu_{k}$, $k\leq h$, and hence in $\mu_{0}\in I_{h}$. Eq. (\ref{eq:flowmu}) implies $\mu_{h+1}(\mu_{0})$ is continuous in $\mu_{0} \in I_{h}$, and that:
\begin{equation}
\mu_{h+1}(I_{h}) \supset \{|z| < 2\overline{C} L |\lambda_{h}| - K L|\lambda_{h}|\} \supset \;\{|z| < 2 \overline{C} |\lambda_{h+1}|\}
\end{equation}
where the last step follows from Eq. (\ref{eq:CCbar}). Thus, by continuity there exists $I_{h+1}\subset I_{h}$ such that:
\begin{equation}
\{|z| < \overline{C} |\lambda_{h+1}|\} \subset \mu_{h+1}(I_{h+1}) \subset \{|z| < 2 \overline{C} |\lambda_{h+1}|\}\;.
\end{equation}
This shows in particular that $|\mu_{k}|\leq 2\overline{C} |\lambda_{k}|$ for all $k\leq h+1$, which is what we wanted to prove. 

\medskip

\noindent{\underline{\it Uniqueness.}} Here we shall prove the uniqueness of $\mu(\lambda)$: we shall show that the set $I_{h}$ shrinks to a point as $h\to \infty$. To do this, we rely on the Lipschitz continuity of the beta function of the effective chemical potential, as function of the effective chemical potentials on all the previous scales. This will be proven via a Cauchy estimate, which in turn relies on the analyticity properties of the effective potential $U^{(h)}$ as function on $\{\mu_{j}\}_{j=0}^{h}$. In fact, in the next proposition, we shall regard the effective potentials $U^{(h)}$ as functions of the sequence $\{\mu_{k}\}_{k=0}^{h}$: as it is clear from the induction of Section \ref{subsec: General_integration_step}, the only information about the sequence of effective potentials that we required is that $|\mu_{k}| \leq C|\lambda_{k}|$. Also, as already pointed out at the beginning of the section, the constant $C$ can be taken arbitrarily large, provided $\lambda$ is small enough.
\begin{lemma}\label{lem:anal_effective_potential_mu}
Let $U^{(h)}(\Phi) = \sum_{n=0,1,2} U^{(h)}_{n}(\phi) (\psi\cdot \psi)^{n}$, and let $C>0$. Then, for $\lambda$ small enough, for all $h\in \mathbb{N}$, the functions $U_{n}^{(h)}(\Phi)$ are analytic in $\{\mu _{j}\}_{j = 0}^{h}$ with $|\mu _{j}| < C |\lambda_{j}|$ for any $j$, and $\phi \in \mathbb{S}^{(h)} \cup \mathbb{L}^{(h)}$.
\end{lemma}
\begin{proof}
The proof is by induction. The case $h = 0$ is true by inspection.

Next, we assume that Theorem~\ref{thm:effective_potential_flow} and Lemma~\ref{lem:anal_effective_potential_mu} are true on scales $j<h$, and we shall prove that they hold on scale $h$. 
Consider: $f^{(h-1)}(\Phi,\zeta): = [U^{(h-1)}(\Phi/L + \zeta)U^{(h-1)}(\Phi/L - \zeta)]^{\frac{L^{3}}{2}}$. By the inductive assumptions, $f^{(h-1)}$ is analytic in $\{\mu_{j}\}_{j=0}^{h-1}$, provided $|\mu_{j}| < C |\lambda_{j}|$ and $\phi/L \pm \zeta_{\phi} \in \mathbb{S}^{(h)} \cup \mathbb{L}^{(h)}$. The new effective potential $U^{(h)}$ is obtained after integrating the fluctuation field $\zeta$. It is clear that the analyticity domain is left unaffected by the integration on the fermionic fluctuation field $\zeta_{\psi}$. Consider now the integration of the bosonic fluctuation field $\zeta_{\phi}$. By the bounds proven in the analysis of the effective potential in Section \ref{subsec: General_integration_step}, the coefficients of the Grassmann expansion of $f^{(h-1)}$ are absolutely integrable in $\zeta_{\phi}$, for $\{\mu_{j}\}_{j=0}^{h-1}$ in the assumed range and for $\phi \in \mathbb{S}^{(h)} \cup \mathbb{L}^{(h)}$. Thus, analyticity of $\int d\mu(\zeta)\, f^{(h-1)}(\Phi,\zeta)$ in $\{ \mu_{j} \}_{j=0}^{h-1}$ follows from dominated convergence and from Morera's theorem. This implies analyticity of $U_{n}^{(h)}$ in $\{ \mu_{j} \}_{j=0}^{h-1}$. Finally, analyticity of $U_{n}^{(h)}$ in $\mu_{h}$ follows from the fact that the function $e^{i\mu_{h}(\Phi \cdot \Phi)}$ is entire in $\mu_{h}$
%
%
\end{proof}
This lemma easily implies analyticity of the beta function of the chemical potential.
\begin{corollary}\label{cor:anal_beta_mu}
For any $h \in \mathbb{N}$, $\beta^{(h)}_{2}$ is an analytic function in $\{\mu _{j}\}_{j = 0}^{h}$ provided that $|\mu _{j}| < C |\lambda_{j}|$ for all $j$.
\end{corollary}
\begin{proof}
In Section~\ref{subsec: General_integration_step} we proved that $U_{0}^{(h)}(\phi)$ is analytic in $\phi \in \mathbb{S}^{(h)}$   and that
\begin{equation}\label{eq:U_h+1}
U^{(h+1)}_{0}(\phi) = e^{-\lambda_{h+1} (\phi\cdot \phi)^{2} - i\mu_{h+1} (\phi\cdot \phi)}  R^{(h+1)}_{0}(\phi)\;,
\end{equation}
with $\mu_{h+1} = L \mu _{h} + \beta_{2}^{(h)}$, and where $R_{0}^{(h+1)}(\phi)$ is analytic in $\phi \in \mathbb{S}^{(h+1)}$ and satisfies the bound $| R^{(h+1)}_{0}(\phi) - 1 | \leq C\vert \lambda_{h+1} \vert^{2} \|\phi\|^{6}$ for these values of the field $\phi$. Consider the restriction of $U^{(h+1)}_{0}$ and of $R_{0}^{(h+1)}(\phi)$ to $\phi\in \mathbb{R}^{4}$, and recall that both $U^{(h+1)}_{0}$, $R_{0}^{(h+1)}(\phi)$ depend on $\phi\in \mathbb{R}^{4}$ only via $\|\phi\|$ (see Appendix~\ref{append: symmetries}). Therefore,
\begin{equation}\label{eq:muder}
\frac{1}{2}\partial^{2}_{\| \phi\|} U_{0}^{(h+1)}(0) = 
- i \mu _{h+1} \, R_{0}^{(h+1)}(0) + \frac{1}{2}\partial^{2}_{\| \phi \|} R_{0}^{(h+1)} (0)  = -i \mu _{h+1} \;;
\end{equation}
the last identity follows from $R_{0}^{(h+1)}(0) = 1$, implied by SUSY (see Appendix \ref{append: symmetries}), and by the fact that $R_{0}^{(h+1)} (\phi) - 1$ is at least of order $\|\phi\|^{6}$. 

By Lemma~\ref{lem:anal_effective_potential_mu} we know that $U_{0}^{(h+1)}(\phi)$ is analytic in $\{\mu _{j}\}_{j = 0}^{h}$ uniformly in $\phi\in \mathbb{S}^{(h)}\cup \mathbb{L}^{(h)}$, provided that $|\mu _{j}| < C |\lambda_{j}|$ for any $j$. Hence, so is $\partial^{2}_{\| \phi\|} U_{0}^{(h+1)}(0)$; this together with the identity $\beta_{2}^{(h)} = i (1/2)\partial^{2}_{\| \phi\|} U_{0}^{(h+1)}(0) - L\mu _{h}$, implied by Eq. (\ref{eq:muder}) and by the definition of beta function, shows the analyticity of $\beta_{2}^{(h)}$.
\end{proof}
We are now ready to prove uniqueness of $\mu(\lambda)$. We shall proceed by contradiction. Suppose that the function $\mu(\lambda)$ is not unique: there exist two sequences $\underline{\mu} = \{ \mu_{k} \}_{k=0}^{\infty}$ and $\underline{\mu}' = \{ \mu'_{k} \}_{k=0}^{\infty}$ such that $|\mu_{k}| \leq 2\overline{C}|\lambda_{k}|$ and $|\mu'_{k}| \leq 2\overline{C}|\lambda'_{k}|$, solving Eq. (\ref{eq:flowmu}), with $\overline{C}$ as in the statement of Proposition \ref{prp:count}. We are denoting by $\{\lambda'_{k}\}$ the effective quartic couplings associated to the sequence $\{\mu'_{k}\}$. By the assumptions on the sequence $\{\mu_{k}\}$, it is true that $| \lambda'_{k} - \frac{\lambda}{L^{k}} | \leq \frac{K \lambda^{\frac{3}{2}}}{L^{k}}$, $| \lambda_{k} - \frac{\lambda}{L^{k}} | \leq \frac{K \lambda^{\frac{3}{2}}}{L^{k}}$ for all $k$. We have:
\begin{equation}\label{eq:deltamu}
\mu_{k+1} - \mu_{k+1}' = L(\mu_{k} - \mu'_{k}) + \beta_{2}^{(k)}(\mu_{k}, \ldots, \mu_0) - \beta_{2}^{(k)}(\mu_{k}',\ldots, \mu_{0}')\;;
\end{equation}
by Corollary~\ref{cor:anal_beta_mu}, the beta functions are analytic in $\{\mu_{j}\}$, $\{\mu'_{j}\}$, respectively, if $|\mu_{j}| < C|\lambda_{j}|$ and $|\mu'_{j}| < C|\lambda'_{j}|$. Notice that we are free to assume that $C>4\overline{C}$, if $\lambda$ is small enough. We rewrite Eq. (\ref{eq:deltamu}) as:
\begin{equation}\label{eq:difflow}
\mu_{k} - \mu_{k}' = -\sum_{j=k}^{\infty} L^{k-j-1} ( \beta_{2}^{(j)}(\mu_{j}, \ldots, \mu_0) - \beta_{2}^{(j)}(\mu_{j}',\ldots, \mu_{0}') )\;,
\end{equation}
where we used that the sequences vanish at infinity. In order to estimate the difference of the beta functions, we proceed as follows. We use the telescopic representation:
\begin{eqnarray}
&&\beta_{2}^{(j)}(\mu_{j}, \ldots, \mu_0) - \beta_{2}^{(j)}(\mu_{j}',\ldots, \mu_{0}') = \\
&& \sum_{r=0}^{j} \Big( \beta_{2}^{(j)}(\mu'_{j}, \ldots, \mu'_{r+1}, \mu_{r}, \ldots, \mu_0) - \beta_{2}^{(j)}(\mu_{j}',\ldots, \mu'_{r}, \mu_{r-1}, \ldots, \mu_{0}) \Big)\;.\nonumber
\end{eqnarray}
We estimate every addend in the sum via a Cauchy estimate. We interpolate:
\begin{eqnarray}
&&\beta_{2}^{(j)}(\mu'_{j}, \ldots, \mu'_{r+1}, \mu_{r}, \ldots, \mu_0) - \beta_{2}^{(j)}(\mu_{j}',\ldots, \mu'_{r}, \mu_{r-1}, \ldots, \mu_{0}) \nonumber\\
&&\quad = \int_{\mu'_{r}}^{\mu_{r}} d\nu\, \frac{\partial}{\partial \nu} \beta_{2}^{(j)}(\mu'_{j}, \ldots, \mu'_{r+1}, \nu, \ldots, \mu_0)
\end{eqnarray}
and we use that, in the integral, the distance between $\nu$ and the boundary of the analyticity domain of $\beta_{2}^{(j)}$ in its $r$-th variable is bounded below by $\overline{C} \lambda L^{-r}$. Also, for these values of $\nu$, the beta function satisfies the bound $|\beta_{2}^{(j)}(\mu'_{j}, \ldots, \mu'_{r+1}, \nu, \ldots, \mu_0)| \leq K \lambda L^{-j}$. Therefore, by a Cauchy estimate:
\begin{eqnarray}
|\beta_{2}^{(j)}(\mu_{j}, \ldots, \mu_0) - \beta_{2}^{(j)}(\mu_{j}',\ldots, \mu_{0}')| &\leq& \sum_{r = 0}^{j} \widetilde{K} L^{r-j} |\mu_{r} - \mu'_{r}| \nonumber\\
&\leq& 2 \widetilde{K} \| \underline{\mu} - \underline{\mu}' \|_{\infty}\;,
\end{eqnarray}
with $\|\underline{\mu}\|_{\infty} = \sup_{k} | \mu_{k} |$ and where $\widetilde{K}$ is proportional to $K/\overline{C}$. Plugging this bound in Eq. (\ref{eq:difflow}) we get:
\begin{equation}
|\mu_{k} - \mu'_{k}| \leq 2\widetilde{K} \sum_{j=k}^{\infty} L^{k-j-1} \| \underline{\mu} - \underline{\mu}' \|_{\infty}\qquad \forall k\in \mathbb{N}\;,
\end{equation}
which implies:
\begin{equation}\label{eq:mumu'}
\| \underline{\mu} - \underline{\mu}' \|_{\infty} \leq 4 \widetilde{K} L^{-1} \| \underline{\mu} - \underline{\mu}' \|_{\infty}\;.
\end{equation}
For $L^{-1} \widetilde{K}$ small enough, Eq. (\ref{eq:mumu'}) implies $\underline{\mu} = \underline{\mu}'$. This concludes the proof of Proposition \ref{prp:count}.
\end{proof}

\section{Proof of technical lemmas}\label{app:osc}
In this section we collect the proofs of three key results, namely Lemma \ref{lem:osc}, Lemma \ref{lem:cau} and Proposition \ref{prp:mixed}.

\subsection{Proof of Lemma \ref{lem:osc}}
Let $\mathcal{S}(\mathbb{R}^{n})$ be the Schwartz space, and let $f\in \mathcal{S}(\mathbb{R}^{n})$. Then:
\begin{equation}\label{eq:pl0}
\int_{\mathbb{R}^{n}} dx\, e^{-i\|x\|^{2}} f(x) =  \frac{1}{2^{\frac{n}{2}}} \int_{\mathbb{R}^{n}} dp\, e^{\frac{i}{4} \|p\|^{2}} \hat f(p)\;.
\end{equation}
The proof of (\ref{eq:pl0}) goes as follows. Since $f\in L^{1}(\mathbb{R}^{n})$ we have, by dominated convergence:
%
%
\begin{eqnarray}
\int_{\mathbb{R}^{n}} dx\, e^{-i\|x\|^{2}} f(x) &=& \int_{\mathbb{R}^{n}} dx\, \lim_{\varepsilon\to 0^{+}} e^{-(i + \eps)\|x\|^{2}} f(x)\nonumber\\
&=& \lim_{\varepsilon\to 0^{+}} \int_{\mathbb{R}^{n}} dx\, e^{-(i + \eps)\|x\|^{2}} f(x)\;.
\end{eqnarray}
Then, since $e^{-(i + \eps)\|x\|^{2}}\in L^{2}(\mathbb{R}^{n})$, $f\in L^{2}(\mathbb{R}^{n})$ we have, by Plancherel's theorem:
\begin{equation}
\int_{\mathbb{R}^{n}} dx\, e^{-i\|x\|^{2}} f(x) = \lim_{\varepsilon\to 0^{+}} \int dp\, \overline{\hat g_{\varepsilon}(p)} \hat f(p)\;,
\end{equation}
where $\hat f \in \mathcal{S}(\mathbb{R}^{n})$, $\hat g_{\varepsilon}\in \mathcal{S}(\mathbb{R}^{n})$, given by:
\begin{eqnarray}
\overline{\hat g_{\varepsilon}(p)} &:=& \frac{1}{(2\pi)^{\frac{n}{2}}}\int dx\, e^{ip\cdot x} e^{-(i+\eps)\|x\|^{2}} \nonumber\\
&=& \frac{1}{2^{\frac{n}{2}}} e^{-\frac{1}{4 (i + \eps)} \|p\|^{2}}\;.
\end{eqnarray}
Hence, applying again dominated convergence,
\begin{eqnarray}
\int dx\, e^{-i\|x\|^{2}} f(x) &=& \int dp\, \lim_{\varepsilon\to 0^{+}} \overline{\hat g_{\varepsilon}(p)} \hat f(p)\nonumber\\
&=&  \frac{1}{2^{\frac{n}{2}}} \int dp\, e^{\frac{i}{4}\|p\|^{2}} \hat f(p)\;,
\end{eqnarray}
which concludes the proof of Eq. (\ref{eq:pl0}). Lemma \ref{lem:osc} follows from Taylor expanding $e^{\frac{i}{4}\|p\|^{2}}$.
\qed
\medskip

\subsection{Proof of Lemma \ref{lem:cau}.} 
We shall only prove Lemma \ref{lem:cau} part $(b)$, part $(a)$ being the well-known Cauchy estimate. The proof of part $(b)$ is a simple application of Cauchy formula and integration by parts. Consider:
\begin{eqnarray}\label{eq:parts}
\Big| p_{i}^{m} \hat f(p)  \Big| &=& \Big| \frac{1}{(2\pi)^{\frac{n}{2}}}\int dx\, e^{-ip\cdot x} \partial_{x_{i}}^{m} f(x_{1},\ldots, x_{n}) \Big|\nonumber\\
&\leq& \frac{1}{(2\pi)^{\frac{n}{2}}}\int dx\, | \partial_{x_{i}}^{m} f(x_{1}, \ldots, x_{n}) |\;.
\end{eqnarray}
In the first step we used that $\partial_{x_{i}}^{m} f$ vanishes as $|x_{i}|\to \infty$ for all $m\geq 0$. This follows from the analyticity of $f$, from the representation of $\partial_{x_{i}}^{m} f$ via Cauchy formula, and from Eq. (\ref{eq:FW}), which implies the vanishing of $f$ as $|\text{Re}\, z_{i}|\to \infty$, in a strip around the real axis.

Now, Cauchy theorem, combined with the analyticity of the function $z_{i} \mapsto f(x_{1}, \ldots, z_{i}, \ldots, x_{n})$ in $\mathbb{R}_{W}$, implies:
\begin{equation}
| \partial_{x_{i}}^{m} f(x_{1}, \ldots, x_{n}) | \leq K_{m} \Big| \int_{\mathcal{C}(x_{i})} dz_{i}\, \frac{f(x_{1}, \ldots, z_{i}, \ldots, x_{n})}{(z_{i} - x_{i})^{m+1}} \Big|\;,
\end{equation}
where $\mathcal{C}(x_{i}) := \{ z_{i} \mid | z_{i} - x_{i} | = W \}$. Changing variable:
\begin{eqnarray}\label{eq:caucau}
| \partial_{x_{i}}^{m} f(x_{1}, \ldots, x_{n}) | &\leq& K_{m} \Big| \int_{\mathcal{C}(0)} dw_{i}\, \frac{f(x_{1}, \ldots, w_{i} + x_{i}, \ldots, x_{n})}{w_{i}^{m+1}} \Big|\\
& \leq & \frac{K_{m}}{{W}^{m+1}} \int_{\mathcal{C}(0)} |dw_{i}|\, |f(x_{1}, \ldots, w_{i} + x_{i}, \ldots, x_{n})|\;,\nonumber
\end{eqnarray}
where we used the notation $|dw_{i}| := \Big|\frac{dw_{i}(t)}{dt} \Big| dt$, for any parametrization $w_{i}(t)$ of $\mathcal{C}(x_{i})$. 
Plugging (\ref{eq:caucau}) this into (\ref{eq:parts}) we get:
\begin{eqnarray}\label{eq:pnf}
\Big| p_{i}^{m} \hat f(p)  \Big| &\leq& \frac{\widetilde K_{m}}{{W}^{m+1}} \int dx \int_{\mathcal{C}(0)} |dw_{i}|\, |f(x_{1}, \ldots, w_{i} + x_{i}, \ldots, x_{n})|\nonumber\\
&=&  \frac{\widetilde K_{m}}{{W}^{m+1}} \int_{\mathcal{C}(0)} |dw_{i}| \int dx\, |f(x_{1}, \ldots, w_{i} + x_{i}, \ldots, x_{n})|\nonumber\\
&\leq& \frac{2\pi \widetilde K_{m}}{{W}^{m}} F_{W}(f)\;,
\end{eqnarray}
where we used the assumption (\ref{eq:FW}). Thus, from Eq. (\ref{eq:pnf}) we easily get that, for any $p\in \mathbb{R}^{n}$, $m\in \mathbb{N}$, and for some $C_{m}>0$:
\begin{equation}\label{eq:1+W}
(1 + (W |p|)^{m}) |\hat f(p)| \leq C_{m} F_{W}(f)\;,
\end{equation}
which implies Eq. (\ref{eq:estW}). The final statement, Eq. (\ref{eq:rembd}), follows immediately from the bound (\ref{eq:1+W}) together with the formula (\ref{eq:remcau}) for the remainder of the stationary phase expansion. This concludes the proof. \qed
\subsection{Proof of Proposition \ref{prp:mixed}.} 
We start by writing:
\begin{eqnarray}\label{eq:mixedproof1}
&&e^{-\widetilde{V}^{(h)}_{\text{b}}(\phi, \zeta_{\phi})} e^{c_{h}\frac{L^{3}}{2}|\lambda_{h}| \| \phi/L + \zeta_{\phi} \|^{4} + c_{h}\frac{L^{3}}{2}|\lambda_{h}| \| \phi/L - \zeta_{\phi} \|^{4} } \nonumber\\
&&\qquad = e^{-\widetilde{V}^{(h)}_{\text{b}}(0, \zeta_{\phi})} e^{c_{h} \frac{|\lambda_{h}|}{L} \|\phi\|^{4} + c_{h} |\lambda_{h}|L^{3} \|\zeta_{\phi}\|^{4}} e^{A(\phi, \zeta_{\phi})}
\end{eqnarray}
with:
\begin{eqnarray}
A(\phi, \zeta_{\phi}) &:=& -4 \lambda_{h} L (\phi\cdot \zeta_{\phi})^{2} - 2\lambda_{h} L (\phi\cdot \phi) (\zeta_{\phi} \cdot \zeta_{\phi})\nonumber\\&& + 4c_{h}|\lambda_{h}|L(\text{Re}\,\langle\phi, \zeta_{\phi}\rangle)^{2} + 2c_{h}L |\lambda_{h}| \|\phi\|^{2} \|\zeta_{\phi}\|^{2}\;.
\end{eqnarray}
and
\begin{eqnarray}
\widetilde V^{(h)}_{\text{b}}(0,\zeta_{\phi}) &=& (\text{Re}\lambda_{h}) L^{3} (\zeta_{\phi}\cdot \zeta_{\phi})^{2} + i(\text{Im} \lambda_{h}) L^{3} (\zeta_{\phi}\cdot \zeta_{\phi})^{2}\nonumber\\&& + i\mu_{h} L^{3} (\zeta_{\phi}\cdot \zeta_{\phi})\;.
\end{eqnarray}
Using that:
\begin{equation}
(\zeta_{\phi} \cdot \zeta_{\phi}) = (\text{Re}\, \zeta_{\phi} \cdot \text{Re}\, \zeta_{\phi}) - (\text{Im}\, \zeta_{\phi}\cdot \text{Im}\, \zeta_{\phi}) + 2i (\text{Re}\, \zeta_{\phi}\cdot \text{Im}\, \zeta_{\phi})
\end{equation}
we have, for any $\eta>0$:
\begin{eqnarray}
\text{Re}\, (\zeta_{\phi} \cdot \zeta_{\phi})^{2} &=& \Big( (\text{Re}\, \zeta_{\phi} \cdot \text{Re}\, \zeta_{\phi}) - (\text{Im}\, \zeta_{\phi}\cdot \text{Im}\, \zeta_{\phi}) \Big)^{2} - 4 (\text{Re}\, \zeta_{\phi}\cdot \text{Im}\, \zeta_{\phi})^{2} \nonumber\\
&\geq& \|\text{Re}\, \zeta_{\phi}\|^{4} - 6 \|\text{Re}\, \zeta_{\phi}\|^{2} \|\text{Im}\, \zeta_{\phi}\|^{2}\nonumber\\
&\geq& (1 - 3\eta) \|\text{Re}\, \zeta_{\phi}\|^{4} - 3\eta^{-1} \|\text{Im}\, \zeta_{\phi}\|^{4}\;.
\end{eqnarray}
Then, since $\|\zeta_{\phi}\|^{2} = \| \text{Re}\, \zeta_{\phi} \|^{2} + \| \text{Im}\, \zeta_{\phi} \|^{2}$ we have, using that by assumption $\|\text{Im}\, \zeta_{\phi}\| \leq |\lambda_{h}|^{-\frac{1}{4} + \eps}$, $\| \text{Re}\, \zeta_{\phi} \|^{2} \geq \|\zeta_{\phi}\|^{2} - |\lambda_{h}|^{-\frac{1}{2} + 2\eps}$; this implies:
\begin{eqnarray}
\text{Re}\, (\zeta_{\phi} \cdot \zeta_{\phi})^{2} \geq (1 - 4\eta) \|\zeta_{\phi}\|^{4} - C\eta^{-1} |\lambda_{h}|^{-1 + 4\eps}\;.
\end{eqnarray}
This, together with $|\mu_{h}| \leq C|\lambda_{h}|$ and $|\text{Im}\, \lambda_{h}| \leq C\lambda |\lambda_{h}|$ (recall  (\ref{eq:lambdah})) gives, for $\lambda$ small enough:
\begin{equation} \label{eq:reVest}
\text{Re}\, \widetilde V^{(h)}_{\text{b}}(0,\zeta_{\phi})\geq (1 - 5\eta) |\lambda_{h}| L^{3} \|\zeta_{\phi}\|^{4} - K \eta^{-1} L^{3} |\lambda_{h}|^{4\eps}\;.
\end{equation}
Next, consider $A(\phi, \zeta_{\phi})$. We write:
\begin{eqnarray}\label{eq:Aphi}
A(\phi, \zeta_{\phi}) &=& -4 |\lambda_{h}| L (\phi\cdot \zeta_{\phi})^{2} - 2|\lambda_{h}| L (\phi\cdot \phi) (\zeta_{\phi} \cdot \zeta_{\phi})\\&& + 4c_{h}|\lambda_{h}|L(\text{Re}\,\langle\phi, \zeta_{\phi}\rangle)^{2} + 2c_{h}L |\lambda_{h}| \|\phi\|^{2} \|\zeta_{\phi}\|^{2} + A_{1}(\phi, \zeta_{\phi})\nonumber
\end{eqnarray}
where $A_{1}(\phi, \zeta_{\phi})$ takes into account the replacement of $\lambda_{h}$ with $|\lambda_{h}|$ in the first two terms. Using that $|\text{Im}\, \lambda_{h}| \leq C\lambda |\lambda_{h}|$, we have that $|\lambda_{h} - |\lambda_{h}|| \leq \widetilde{C}\lambda |\lambda_{h}|$, and hence:
\begin{equation}\label{eq:A1est}
| A_{1}(\phi, \zeta_{\phi})|\leq K \lambda |\lambda_{h}| L \|\phi\|^{2} \|\zeta_{\phi}\|^{2}\;.
\end{equation}
To exhibit a cancellation in the various terms appearing in the right-hand side of Eq. (\ref{eq:Aphi}), we can to show that $(\phi\cdot \zeta_{\phi})^{2} - (\text{Re}\,\langle\phi, \zeta_{\phi}\rangle)^{2}$ and $(\phi\cdot \phi) (\zeta_{\phi} \cdot \zeta_{\phi}) - \|\phi\|^{2} \|\zeta_{\phi}\|^{2}$ are small. We have:
\begin{eqnarray}\label{eq:canc1}
&&\text{Re}\, (\phi \cdot \zeta_{\phi})^{2} - (\text{Re}\, \langle \phi, \zeta_{\phi} \rangle)^{2} = (\text{Re}\, (\phi \cdot \zeta_{\phi}))^{2} -(\text{Re}\, \langle \phi, \zeta_{\phi} \rangle)^{2} - (\text{Im}\, ( \phi \cdot \zeta_{\phi} ))^{2}\nonumber\\
&& = \Big( \text{Re}\, (\phi \cdot \zeta_{\phi}) - \text{Re}\, \langle \phi, \zeta_{\phi} \rangle \Big)\Big( \text{Re}\, (\phi \cdot \zeta_{\phi}) + \text{Re}\, \langle \phi, \zeta_{\phi} \rangle \Big) - (\text{Im}\, ( \phi \cdot \zeta_{\phi}))^{2}\nonumber\\
&& = -2 ( \text{Im}\phi \cdot \text{Im}\zeta_{\phi} )\Big( \text{Re}\, (\phi \cdot \zeta_{\phi}) + \text{Re}\, \langle \phi, \zeta_{\phi} \rangle \Big)  - (\text{Im}\, ( \phi \cdot \zeta_{\phi} ))^{2}\;.
\end{eqnarray}
Therefore,
\begin{eqnarray}
| \text{Re}\, (\phi \cdot \zeta_{\phi})^{2} - (\text{Re}\, \langle \phi, \zeta_{\phi} \rangle)^{2} | &\leq& C[ \| \text{Im}\, \phi \|\|\text{Im}\, \zeta_{\phi}\| \|\phi\| \|\zeta_{\phi}\| \\&& + \| \text{Im}\,\phi \|^{2} \| \text{Re}\, \zeta_{\phi} \|^{2} + \| \text{Re}\, \phi \|^{2} \| \text{Im}\, \zeta_{\phi} \|^{2}]\;.\nonumber
\end{eqnarray}
Using that, by assumption, $\| \text{Im}\, \phi \|\leq L |\lambda_{h}|^{-\frac{1}{4}}$ and that $\| \text{Im}\, \zeta_{\phi} \| \leq |\lambda_{h}|^{-\frac{1}{4} + \epsilon}$, we get:
\begin{equation}\label{eq:canc1est}
| \text{Re}\, (\phi \cdot \zeta_{\phi})^{2} - (\text{Re}\, \langle \phi, \zeta_{\phi} \rangle)^{2} | \leq K L^{2}  |\lambda_{h}|^{-\frac{1}{2}} \|\zeta_{\phi}\|^{2} + K L |\lambda_{h}|^{-\frac{1}{2} + 2\varepsilon} \|\phi\|^{2}\;.
\end{equation}
Similarly,
\begin{eqnarray}\label{eq:canc2}
&&\text{Re} ((\phi \cdot \phi) (\zeta_{\phi}\cdot \zeta_{\phi})) - \|\phi\|^{2} \|\zeta_{\phi}\|^{2} \\
&&= \text{Re} (\phi \cdot \phi) \text{Re}(\zeta_{\phi}\cdot \zeta_{\phi}) - \text{Im} (\phi \cdot \phi) \text{Im}(\zeta_{\phi}\cdot \zeta_{\phi}) - \|\phi\|^{2} \|\zeta_{\phi}\|^{2}\nonumber\\
&& = \text{Re}\,(\phi\cdot \phi)\Big( \text{Re}\,(\zeta_{\phi}\cdot \zeta_{\phi}) - \|\zeta_{\phi}\|^{2} \Big) + \| \zeta_{\phi} \|^{2} \Big( \text{Re}\,(\phi\cdot\phi) - \|\phi\|^{2} \Big)\nonumber\\
&& \quad - \text{Im} (\phi \cdot \phi) \text{Im}(\zeta_{\phi}\cdot \zeta_{\phi})\nonumber\\
&& = -2 \text{Re}(\phi\cdot \phi) \| \text{Im}\zeta_{\phi} \|^{2}  - 2\|\zeta_{\phi}\|^{2} \| \text{Im}\phi \|^{2} - \text{Im} (\phi \cdot \phi) \text{Im}(\zeta_{\phi}\cdot \zeta_{\phi})\;.\nonumber
\end{eqnarray}
Therefore,
\begin{eqnarray}
&&|\text{Re} ((\phi \cdot \phi) (\zeta_{\phi}\cdot \zeta_{\phi})) - \|\phi\|^{2} \|\zeta_{\phi}\|^{2}| \leq C[ \| \phi \|^{2} \|\text{Im}\zeta_{\phi}\|^{2} + \|\zeta_{\phi}\|^{2} \|\text{Im}\phi\|^{2}\nonumber\\&&\qquad + \|\text{Im}\zeta_{\phi}\| \|\text{Im}\phi\|\|\text{Re}\phi\|\|\text{Re}\zeta_{\phi}\| ]\;.
\end{eqnarray}
Using again the assumptions on $\phi$ and on $\zeta_{\phi}$:
\begin{equation}\label{eq:canc2est}
|\text{Re} ((\phi \cdot \phi) (\zeta_{\phi}\cdot \zeta_{\phi})) - \|\phi\|^{2} \|\zeta_{\phi}\|^{2}| \leq K L^{2} |\lambda_{h}|^{-\frac{1}{2}} \|\zeta_{\phi}\|^{2} + K L |\lambda_{h}|^{-\frac{1}{2} + 2\varepsilon} \|\phi\|^{2}\;.
\end{equation}
Therefore, we rewrite the real part of Eq. (\ref{eq:Aphi}) as:
\begin{eqnarray}\label{eq:ReA}
\text{Re}\, A(\phi, \zeta_{\phi}) &=& -4 |\lambda_{h}| L \text{Re}\, (\phi\cdot \zeta_{\phi})^{2} - 2|\lambda_{h}| L \text{Re}\,((\phi\cdot \phi) (\zeta_{\phi} \cdot \zeta_{\phi}))\nonumber\\&& + 4c_{h}|\lambda_{h}|L(\text{Re}\,\langle\phi, \zeta_{\phi}\rangle)^{2} + 2c_{h}L |\lambda_{h}| \|\phi\|^{2} \|\zeta_{\phi}\|^{2} + \text{Re}\,A_{1}(\phi, \zeta_{\phi})\;,\nonumber\\
&\equiv& 4 |\lambda_{h}| (c_{h} - 1) L (\text{Re}\, \langle \phi, \zeta_{\phi} \rangle)^{2} + 2|\lambda_{h}| (c_{h} - 1) L \|\phi\|^{2} \|\zeta_{\phi}\|^{2} \nonumber\\
&& + A_{2}(\phi, \zeta_{\phi})\;,
\end{eqnarray}
where in the last step we used (\ref{eq:canc1}), (\ref{eq:canc2}); hence we have, using the bounds (\ref{eq:A1est}), (\ref{eq:canc1est}), (\ref{eq:canc2est}):
\begin{eqnarray}\label{eq:A2est}
|A_{2}(\phi, \zeta_{\phi})| &\leq& K \lambda |\lambda_{h}| L \|\phi\|^{2} \|\zeta_{\phi}\|^{2}\nonumber\\&& + \widetilde K L^{3} |\lambda_{h}|^{\frac{1}{2}} \|\zeta_{\phi}\|^{2} + \widetilde K L |\lambda_{h}|^{\frac{1}{2} + 2\varepsilon} \|\phi\|^{2}\;.
\end{eqnarray}
The first two terms in the right-hand side of Eq. (\ref{eq:ReA}) are negative, thanks to $|c_{h} - 1/6| \leq C \lambda^{\varepsilon}$, recall Eq. (\ref{eq:indh}). Also, the first term in the right-hand side of Eq. (\ref{eq:A2est}) can be controlled using that $2 (c_{h} - 1) + K \lambda < 0$ for $\lambda$ small enough. Hence,
\begin{eqnarray}\label{eq:ReAest}
\text{Re}\, A(\phi, \zeta_{\phi}) &\leq& \widetilde K L^{3} |\lambda_{h}|^{\frac{1}{2}} \|\zeta_{\phi}\|^{2} + \widetilde K L |\lambda_{h}|^{\frac{1}{2} + 2\varepsilon} \|\phi\|^{2} \\
&\leq& \eta |\lambda_{h}| L^{3} \|\zeta_{\phi}\|^{4} + \frac{1}{L}|\lambda_{h}|^{1 + 4\eps} \| \phi\|^{4} + \widetilde K^{2} L^{3} + \eta^{-3} \widetilde{K}^{2} L^{5}\;.\nonumber
\end{eqnarray}
Finally, thanks to (\ref{eq:mixedproof1}), (\ref{eq:reVest}), (\ref{eq:ReAest}):
\begin{eqnarray}
&&\Big| e^{-\widetilde{V}^{(h)}_{\text{b}}(\phi, \zeta_{\phi})} e^{c_{h}\frac{L^{3}}{2}|\lambda_{h}| \| \phi/L + \zeta_{\phi} \|^{4} + c_{h}\frac{L^{3}}{2}|\lambda_{h}| \| \phi/L - \zeta_{\phi} \|^{4} }\Big| \\
&&\qquad \leq e^{-\text{Re}\, \widetilde{V}^{(h)}_{\text{b}}(0, \zeta_{\phi})}  e^{ \text{Re}\, A(\phi, \zeta_{\phi})} e^{c_{h} \frac{|\lambda_{h}|}{L} \|\phi\|^{4} + c_{h} |\lambda_{h}|L^{3} \|\zeta_{\phi}\|^{4}} \nonumber\\
&&\qquad \leq C_{L, \eta} e^{-|\lambda_{h}| (1 - 6\eta) L^{3} \|\zeta_{\phi}\|^{4} + \frac{1}{L}|\lambda_{h}|^{1 + 4\eps} \| \phi\|^{4}} e^{c_{h} \frac{|\lambda_{h}|}{L} \|\phi\|^{4} + c_{h} |\lambda_{h}|L^{3} \|\zeta_{\phi}\|^{4}}\nonumber
\end{eqnarray}
which gives, for $\eta$ such that $1 - 6\eta - c_{h} \geq 1/2$ and for $\tilde c_{h} := c_{h} + |\lambda_{h}|^{4\eps}$:
\begin{eqnarray}
&&\Big| e^{-\widetilde{V}^{(h)}_{\text{b}}(\phi, \zeta_{\phi})} e^{c_{h}\frac{L^{3}}{2}|\lambda_{h}| \| \phi/L + \zeta_{\phi} \|^{4} + c_{h}\frac{L^{3}}{2}|\lambda_{h}| \| \phi/L - \zeta_{\phi} \|^{4} }\Big| \nonumber\\
&&\qquad \leq C_{L} e^{-|\lambda_{h}| \frac{L^{3}}{2} \|\zeta_{\phi}\|^{4} + \tilde c_{h}\frac{|\lambda_{h}|}{L} \|\phi\|^{4}}\;,
\end{eqnarray}
which is the final claim. \qed

\section{Symmetries}\label{app:SUSY}
\label{append: symmetries}

Here we shall discuss the symmetry properties of the model. Recall that the notations $\phi$, $\psi$ denote the four-component vectors $\phi =  (\phi_{1,\uparrow}, \phi_{2,\uparrow}, \phi_{1,\downarrow}, \phi_{2,\downarrow})^{T}$, $\psi = (\psi^{+}_{\uparrow}, \psi^{-}_{\uparrow}, \psi^{+}_{\downarrow}, \psi^{-}_{\downarrow})^{T}$. In the following, we shall denote by $\sigma_{y}$ the second Pauli matrix:
\begin{equation}
\sigma_{y} = \begin{pmatrix} 0 & i \\ -i & 0 \end{pmatrix}\;.
\end{equation}
\begin{proposition}\label{prp:sym} Let $\phi \in \mathbb{C}^{4}$ and $\psi = (\psi^{+}_{\uparrow}, \psi^{-}_{\uparrow}, \psi^{+}_{\downarrow}, \psi^{-}_{\downarrow})^{T}$. Let us denote by  $O \Phi = (O_{b} \phi, O_{f}\psi)$ the transformation $O = (O_{b},O_{f}) \in \mathbb{R}^{4 \times 4} \times \mathbb{C}^{4 \times 4}$ such that 
\begin{equation}
O_{b}^{T}O_{b}= \mathbbm{1}_{4},
\qquad
O_{f}^{T} \left (i\sigma_{y} \otimes \mathbbm{1}_{2} \right ) O_{f} = i\sigma_{y} \otimes \mathbbm{1}_{2},
\qquad
\det O_{f} = 1 \;.
\end{equation}
%
Then, for all scales $h\geq 0$:
\begin{equation}\label{eq:O}
U^{(h)}(\Phi) = U^{(h)}(O\Phi)\;,
\end{equation}
\end{proposition}
\begin{proof} The proof is by induction. Consider Eq. (\ref{eq:O}). The statement is true for $h=0$, since $(\Phi\cdot \Phi) = (O\Phi\cdot O\Phi)$. Suppose it is true for all scales $k\leq h$. Let us prove it for the scale $k=h+1$. We have:
\begin{eqnarray}
U^{(h+1)}(O\Phi) &=& \int d\zeta\, e^{-i (\zeta \cdot \zeta)} [ U^{(h)}(O\Phi/L + \zeta) U^{(h)}(O\Phi/L - \zeta) ]^{\frac{L^{3}}{2}}\\
&=& \int d\zeta\, e^{-i (\zeta \cdot \zeta)} [ U^{(h)}(O(\Phi/L + O^{-1}\zeta)) U^{(h)}(O(\Phi/L - O^{-1}\zeta)) ]^{\frac{L^{3}}{2}}\nonumber\\
&=& \int d\zeta\, e^{-i (\zeta \cdot \zeta)} [ U^{(h)}(\Phi/L + O^{-1}\zeta) U^{(h)}(\Phi/L - O^{-1}\zeta) ]^{\frac{L^{3}}{2}}\nonumber
\end{eqnarray}
where in the last step we used the validity of the symmetry on scale $h$. Let us now perform the change of variable $O^{-1} \zeta\to \zeta'$, with $\zeta'_{\phi} \in \mathbb{R}^{4}$ thanks to the fact that $O_{b}\in \mathbb{R}^{4\times 4}$. Since  the Jacobian of the transformation is $|\det O_{b}| (\det O_{f})^{-1} = 1$, Eq. (\ref{eq:O}) on scale $h+1$ follows. 
\end{proof}
\begin{corollary}\label{cor:symgra} 
\begin{itemize}
\item[(i)] $U^{(h)}(\Phi)$ is a polynomial in $(\psi \cdot \psi)$.
\item[(ii)]
The following identities hold true, for all scales $h\geq 0$:
\begin{equation}\label{eq:corsym}
E^{(h)}_{n}(\phi) = E^{(h)}_{n}(O_{\text{b}}\phi)\;.
\end{equation}
\end{itemize}
\end{corollary}
\begin{remark}\label{rem:sym} In particular, for $\phi \in \mathbb{R}^{4}$, $E^{(h)}_{n}(\phi) \equiv E^{(h)}_{n}(\|\phi\|)$.
\end{remark}
\begin{proof} 
To prove item $(i)$, we proceed as follows. By construction, $U^{(h)}(\Phi)$ is a polynomial in the Grassmann variables $\psi^{+}_{\sigma}$, $\psi^{-}_{\sigma}$. By Eq. (\ref{eq:O}), $U^{(h)}(\Phi)$ is invariant under $\psi \to O_{f} \psi$, and the only Grassmann monomials invariant under this transformation are powers of $(\psi \cdot \psi)$.

To prove item $(ii)$, recall the expression (\ref{eq:U_h}):
\begin{equation}
U^{(h+1)}(\Phi) = e^{- \frac{\lambda_{h}}{L} (\Phi\cdot \Phi)^{2} - i L\mu_{h} (\Phi\cdot \Phi)} \sum_{n=0,1,2} E^{(h)}_{n}(\phi) (\psi\cdot \psi)^{n}\;.
\end{equation}
The exponential prefactor is manifestly invariant under the transformations $O = (O_{b}, O_{f})$ of Proposition \ref{prp:sym}. Therefore, since $(\psi \cdot \psi)$ is $O_{f}$-invariant, the claim (\ref{eq:corsym}) immediately follows.
\end{proof}

The next definition makes precise the notion of supersymmetry of our model.
\begin{definition}\label{def:SUSY} Let us define the differential operator:
\begin{equation}
Q^{\Phi} := \sum_{\sigma, \varepsilon} [\psi^{\varepsilon}_{\sigma} \frac{\partial }{\partial \phi^{\varepsilon}_{\sigma}} - \varepsilon \phi^{\varepsilon}_{\sigma} \frac{\partial}{\partial \psi^{\varepsilon}_{\sigma}}]\;.
\end{equation}
We say that a function $f(\Phi) \equiv f(\phi^{+}, \phi^{-}, \psi^{+}, \psi^{-})$ with $\phi^{+} = \overline{\phi^{-}}$, is {\em supersymmetric} if:
\begin{equation}\label{eq:SUSYdef}
Q^{\Phi} f(\Phi) = 0\;.
\end{equation}
\end{definition}
\begin{remark}
As an example, notice that the combination $\psi^{+}_{\sigma} \psi^{-}_{\sigma} + \phi^{+}_{\sigma} \phi^{-}_{\sigma}$ is a sumersymmetric function. More generally, all analytic functions of $\psi^{+}_{\sigma} \psi^{-}_{\sigma} + \phi^{+}_{\sigma} \phi^{-}_{\sigma}$ are supersymmetric.
\end{remark}
Supersymmetry of regular enough functions implies remarkable identities after integration over the superfields. Here we shall only consider Schwartz functions of the superfield $\Phi$, defined as follows.
\begin{definition}\label{def:schwartz_superfunction}
We say that a function $f(\zeta) = \sum_{\underline{b}} f_{\underline{b}}(\zeta_{\phi}) \,  \zeta_{\psi}^{\underline{b}}$ is of Schwartz type if $f_{\underline{b}}(\cdot)$ are Schwartz functions for all $\underline{b} \in \{0,1 \}^{\{\pm\} \times \{\uparrow,\downarrow\}}$.
\end{definition}
The next lemma will be useful later on, to perform integration by parts over the superfields. 
\begin{lemma}\label{lem:parts} Let $f(\zeta)$ be of Schwartz type. Then:
\begin{equation}
\int d\zeta\, Q^{\zeta} f(\zeta) = 0\;.
\end{equation}
\end{lemma}
\begin{proof}
We write:
\begin{eqnarray}
\int d\zeta\, Q^{\zeta} f(\zeta) &=& \int d\zeta\,  \sum_{\sigma, \varepsilon} [ \zeta^{\varepsilon}_{\psi, \sigma} \frac{\partial}{\partial \zeta^{\varepsilon}_{\phi,\sigma}} - \varepsilon \zeta^{\varepsilon}_{\phi,\sigma} \frac{\partial}{\partial \zeta^{\varepsilon}_{\psi, \sigma}} ] f(\zeta)\nonumber\\
&\equiv& \text{I} + \text{II}\;.
\end{eqnarray}
Consider $\mathrm{I} = \sum_{\sigma,\varepsilon} \int d \zeta_{\psi} \, \zeta_{\psi,\sigma}^{\varepsilon} \, \int d \zeta _{\phi} \,\frac{\partial}{\partial \zeta_{\phi,\sigma}^{\varepsilon}} f(\zeta)  $. Since $f(\zeta)$ is of Schwartz type, the boson integral of the boson derivative is well-defined and equal to zero, which proves that $\text{I} = 0$. 

Consider now $\text{II}$. After differentiation, the Grassmann variable $\zeta^{\varepsilon}_{\psi, \sigma}$ disappears from the integrand, by definition of Grassmann derivative. Therefore, using that $\int d\zeta^{\varepsilon}_{\phi,\sigma} = 0$, we get $\text{II} = 0$.
\end{proof}
This lemma can be used to prove that the effective potential $U^{(h)}(\Phi)$, restricted to $\phi_{i,\sigma} \in \mathbb{R}$, is a supersymmetric function.
\begin{proposition}\label{prp:QU} Let $\Phi = (\psi^{+}, \psi^{-}, \psi^{+}, \psi^{-})$ with $\phi^{+} = \overline{\phi^{-}}$. For all $h\geq 0$,
\begin{equation}\label{eq:USUSY}
Q^{\Phi} U^{(h)}(\Phi) = 0\;.
\end{equation}
\end{proposition}
\begin{proof}(of Proposition \ref{prp:QU}.) The proof goes by induction. The supersymmetry of $(\Phi\cdot \Phi)$ implies that Eq. (\ref{eq:USUSY}) holds true for $h=0$. Suppose it holds for $k<h$ and let us prove it for $k=h$. We have, setting for convenience $U^{(h)}_{L}(\cdot) := U^{(h)}(\cdot)^{\frac{L^{3}}{2}}$:
\begin{eqnarray}
&&Q^{\Phi} U^{(h+1)}(L\Phi) = Q^{\Phi}\int d\mu(\zeta)\, [ U^{(h)}(\Phi + \zeta) U^{(h)}(\Phi - \zeta) ]^{\frac{L^{3}}{2}}\nonumber\\
&& = Q^{\Phi} \int d\mu(\zeta)\, U^{(h)}_{L}(\Phi+\zeta) U^{(h)}_{L}(\Phi-\zeta) \\
&& = \int d\mu(\zeta)\, Q^{\Phi} U^{(h)}_{L}(\Phi+\zeta) U^{(h)}_{L}(\Phi-\zeta)\nonumber\\
&& \equiv \int d\mu(\zeta)\, [(Q^{\Phi} U^{(h)}_{L}(\Phi+\zeta)) U^{(h)}_{L}(\Phi-\zeta) + U^{(h)}_{L}(\Phi+\zeta) (Q^{\Phi} U^{(h)}_{L}(\Phi-\zeta))]  \;.\nonumber
\end{eqnarray}
By using that $Q^{\zeta} e^{-i(\zeta\cdot\zeta)} = 0$ and Lemma \ref{lem:parts} to ``integrate by parts'':
\begin{eqnarray}
&&-Q^{\Phi} U^{(h+1)}(L\Phi) = \\
&&\int d\mu(\zeta)\, [(Q^{\Phi,\zeta} U^{(h)}_{L}(\Phi+\zeta)) U^{(h)}_{L}(\Phi-\zeta) + U^{(h)}_{L}(\Phi+\zeta) (Q^{\Phi,\zeta} U^{(h)}_{L}(\Phi-\zeta))] \nonumber
\end{eqnarray}
with $Q^{\Phi,\zeta} := Q^{\Phi} + Q^{\zeta}$. We claim that:
\begin{equation}\label{eq:claimSUSY}
(Q^{\Phi} + Q^{\zeta}) U^{(h)}_{L}(\Phi\pm \zeta) = (Q^{\Phi} U^{(h)}_{L}) (\Phi \pm \zeta)\;.
\end{equation}
This together with our inductive assumption (\ref{eq:USUSY}) immediately implies that $Q^{\Phi} U^{(h+1)}(L\Phi) = 0$ and concludes the proof. Let us check the claim (\ref{eq:claimSUSY}). We have:
\begin{eqnarray}
&&(Q^{\Phi} + Q^{\zeta}) U^{(h)}_{L}(\Phi\pm \zeta)\nonumber\\
&& = \sum_{\varepsilon, \sigma}\Big( \psi^{\varepsilon}_{\sigma} \frac{\partial}{\partial \phi^{\varepsilon}_{\sigma}} + \zeta_{\psi,\sigma}^{\varepsilon} \frac{\partial}{\partial \zeta_{\phi, \sigma}^{\varepsilon}} - \varepsilon  \phi^{\varepsilon}_{\sigma} \frac{\partial}{\partial \psi^{\varepsilon}_{\sigma}} - \varepsilon\zeta^{\varepsilon}_{\phi, \sigma} \frac{\partial}{\partial \zeta^{\varepsilon}_{\psi, \sigma}} \Big) U^{(h)}_{L}(\Phi\pm \zeta) \nonumber\\
&& = \sum_{\varepsilon, \sigma} \Big( (\psi^{\varepsilon}_{\sigma} \pm \zeta_{\psi,\sigma}^{\varepsilon}) \frac{\partial}{\partial \phi^{\varepsilon}_{\sigma}} - \varepsilon(\phi^{\varepsilon}_{\sigma} \pm \zeta^{\varepsilon}_{\phi, \sigma})\frac{\partial}{\partial \psi^{\varepsilon}_{\sigma}} \Big) U^{(h)}_{L}(\Phi\pm \zeta) \nonumber\\
&&\equiv (Q^{\Phi} U^{(h)}_{L}) (\Phi \pm \zeta)\;,
\end{eqnarray}
which proves Eq. (\ref{eq:claimSUSY}).
\end{proof}
\begin{corollary}\label{cor:cons}
Eqs. (\ref{eq:gammaeq}) hold true.
\end{corollary}
\begin{proof} The proof is by induction. It is trivially true for $h=0$. Suppose it is true for $k<h$, and let us prove it for $k=h$. We rewrite $U^{(h)}$ as:
\begin{equation}
U^{(h)}(\Phi) = e^{-\frac{\lambda_{h}}{L}(\Phi\cdot \Phi)^{2} - i L\mu_{h}(\Phi\cdot \Phi)} \sum_{n=0,1,2} E_{n}^{(h)}(\phi) L^{-2n} (\psi\cdot \psi)^{n}\;.
\end{equation}
Let $\phi_{i,\sigma} \in \mathbb{R}$. Being the explicit factor $e^{-\frac{\lambda_{h}}{L}(\Phi\cdot \Phi)^{2} - iL\mu_{h}(\Phi\cdot \Phi)}$ supersymmetric, Eq. (\ref{eq:USUSY}) implies that:
\begin{equation}\label{eq:QE}
Q^{\Phi} \sum_{n=0,1,2} E_{n}^{(h)}(\phi) L^{-2n} (\psi\cdot \psi)^{n} = 0\;.
\end{equation}
Explicitly,
\begin{eqnarray}\label{eq:QE2}
&&Q^{\Phi} \sum_{n=0,1,2} E_{n}^{(h)}(\phi) L^{-2n} (\psi\cdot \psi)^{n} \\
&&= \sum_{\sigma, \varepsilon} \sum_{n=0,1,2} \Big[ \psi^{\varepsilon}_{\sigma} \Big(\frac{\partial}{\partial \phi^{\varepsilon}_{\sigma}} E_{n}^{(h)}(\phi)\Big) (\psi\cdot \psi)^{n} - \varepsilon \phi^{\varepsilon}_{\sigma} E_{n}^{(h)}(\phi) \Big(\frac{\partial}{\partial \psi^{\varepsilon}_{\sigma}} (\psi\cdot \psi)^{n}\Big) \Big] \equiv\nonumber\\
&&\sum_{\substack{\sigma, \varepsilon \\ n = 0,1,2}} \Big[ \psi^{\varepsilon}_{\sigma} \Big(\frac{\partial}{\partial \phi^{\varepsilon}_{\sigma}} E_{n}^{(h)}(\phi)\Big) (\psi\cdot \psi)^{n} - \phi^{\varepsilon}_{\sigma} E_{n}^{(h)}(\phi) \psi^{-\varepsilon}_{\sigma} n (\psi\cdot \psi)^{n-1}\chi(n\geq 1) \Big]\;.\nonumber
\end{eqnarray}
Let us take $\phi\in \mathbb{S}^{(h)} \cap \mathbb{R}^{4}$, and let us expand the right-hand side to second order in the fields. Combining Eqs. (\ref{eq:QE}), (\ref{eq:QE2}) we get, using that $\partial_{\phi} E^{(h)}_{n}(0) = 0$:
\begin{eqnarray}
0 &=& \sum_{\sigma, \varepsilon} \Big[ \psi^{\varepsilon}_{\sigma} \phi^{-\varepsilon}_{\sigma} \frac{\partial^{2}}{\partial \phi^{\varepsilon}_{\sigma} \partial \phi^{-\varepsilon}_{\sigma}} E^{(h)}_{0}(0) - \phi^{\varepsilon}_{\sigma} \psi^{-\varepsilon}_{\sigma} E^{(h)}_{1}(0)\Big]\nonumber\\
&=& \Big[ \sum_{\sigma, \varepsilon} \psi^{\varepsilon}_{\sigma} \phi^{-\varepsilon}_{\sigma}\Big] \Big( \frac{\partial^{2}}{\partial \phi^{\varepsilon}_{\sigma} \partial \phi^{-\varepsilon}_{\sigma}} E^{(h)}_{0}(0) -  E^{(h)}_{1}(0)\Big)\nonumber\\
&\equiv&\Big[ \sum_{\sigma, \varepsilon} \psi^{\varepsilon}_{\sigma} \phi^{-\varepsilon}_{\sigma}\Big] \Big( \frac{1}{2}\frac{\partial^{2}}{\partial \|\phi\|^{2}} E^{(h)}_{0}(0) -  E^{(h)}_{1}(0)\Big)
\end{eqnarray}
where we used that, by symmetry, $\frac{\partial^{2}}{\partial \phi^{\varepsilon}_{\sigma} \partial \phi^{-\varepsilon}_{\sigma}} E^{(h)}_{0}(0)$ does not depend on $\varepsilon$, $\sigma$. Therefore, we conclude that:
\begin{equation}
\frac{1}{2}\frac{\partial^{2}}{\partial \|\phi\|^{2}} E^{(h)}_{0}(0) -  E^{(h)}_{1}(0) = 0\;,
\end{equation}
which implies that $\gamma^{(h)}_{\phi, 2} = \gamma^{(h)}_{\psi, 2}$ as claimed. Let us now prove the last of Eq. (\ref{eq:gammaeq}). To do so, we expand Eq. (\ref{eq:QE2}) to fourth order in the fields. We get:
\begin{eqnarray}
0 &=& \sum_{\substack{\sigma,\varepsilon \\ \sigma', \varepsilon'}} \psi^{\varepsilon}_{\sigma} \phi^{-\varepsilon}_{\sigma} \phi^{\varepsilon'}\phi^{-\varepsilon'}_{\sigma'} \frac{2}{4!}\frac{\partial^{4} E_{0}^{(h)}(0)}{\partial \|\phi\|^{4}}\nonumber\\
&& + \sum_{\sigma,\varepsilon} \psi^{\varepsilon}_{\sigma} \phi^{-\varepsilon}_{\sigma} (\psi\cdot \psi) \frac{1}{2}\frac{\partial^{2} E^{(h)}_{1}(0)}{\partial \|\phi\|^{2}}\nonumber\\
&& - \sum_{\substack{\sigma, \varepsilon \\ \sigma', \varepsilon'}} \phi^{\varepsilon}_{\sigma} \psi^{-\varepsilon}_{\sigma} \phi^{\varepsilon'}_{\sigma'} \phi^{-\varepsilon'}_{\sigma'} \frac{1}{2}\frac{\partial^{2} E^{(h)}_{1}(0)}{\partial \|\phi\|^{2}}\nonumber\\
&& - 2\sum_{\sigma, \varepsilon} \phi^{\varepsilon}_{\sigma} \psi^{-\varepsilon}_{\sigma}(\psi\cdot \psi) E^{(h)}_{2}(0)\;.
\end{eqnarray}
Therefore, from this equation we infer:
\begin{equation}
\frac{2}{4!}\frac{\partial^{4} E^{(h)}_{0}(0)}{\partial \|\phi\|^{4}} = \frac{1}{2} \frac{\partial^{2} E^{(h)}_{1}(0)}{\partial \|\phi\|^{2}} = 2 E^{(h)}_{2}(0)\;,
\end{equation}
that is:
\begin{equation}
2\gamma^{(h)}_{\phi\phi,4} = \gamma^{(h)}_{\phi\psi,4} = 2\gamma^{(h)}_{\psi\psi,4}\;,
\end{equation}
as claimed.
\end{proof}
Finally, we conclude the appendix by mentioning a well-known result on supersymmetric functions \cite{BT, SZ} (see also \cite{BBS3}, Theorem 11.4.5).
%
\begin{theorem}[Localization theorem]\label{prp:SUSY} Let the function $f(\zeta)$ be supersymmetric in the sense of Eq.~(\ref{eq:SUSYdef}) and of Schwartz type in the sense of Definition~\ref{def:schwartz_superfunction}.
Then:
%
\begin{equation}
\int d\zeta\, f(\zeta) = f(0)\;.
\end{equation} 
\end{theorem}
\begin{remark}
We are not formulating the localization theorem in its most general form. Furthermore, in the present context the theorem is a simple and elegant application of integration by parts \cite{BT}.
\end{remark}
\begin{remark}\label{rem:Nh} 
\begin{itemize}
\item[(i)] Thus, being $U^{(h)}(\zeta)$ of Schwartz type and supersymmetric, and $\frac{L^{3}}{2} \in \mathbb{N}$:
\begin{equation}
E^{(h)}_{0}(0) = \int d\mu(\zeta) [U^{(h)}(\zeta)U^{(h)}(-\zeta)]^{\frac{L^{3}}{2}} = 1\;.
\end{equation}
%
%
%
\item[(ii)] We shall also prove that $\mathcal{Z}_{N} = 1$, recall (\ref{eq:ZN}), and that $\langle \phi^{+}_{x} \phi^{-}_{y} \rangle_{N} = - \langle \psi^{+}_{x} \psi^{-}_{y} \rangle_{N}$, Eq. (\ref{eq:SUSYcorr}), as a consequence of the localization theorem. Here, we shall rely on supersymmetry for functions of the full hierarchical superfield. Some care is needed, since in the hierarchical model the superfield integration is only definined in terms of the integrations of the single-scale fields $\zeta_{x,\sigma}^{(h)\pm}$, $h = 0, ..., N-1$. The full superfield $\Phi^{\pm}_{x,\sigma}$ is a linear combination of the single scale superfields, see Eq.~\eqref{eq:gausshier}.
Let us introduce the global differential operator:
\begin{equation}
Q : = \sum_{h = 0}^{N-1} \sum _{x \in \Lambda^{(h+1)}} Q^{(h)}_{x} \;,
\qquad
Q^{(h)}_{x} : = \sum_{\sigma, \varepsilon} \Big[ \zeta_{\psi,x,\sigma}^{(h)\varepsilon} \frac{\partial}{\partial \zeta_{\phi,x,\sigma}^{(h)\varepsilon}} - \varepsilon
\zeta_{\phi,x,\sigma}^{(h)\varepsilon} \frac{\partial}{\partial \zeta_{\psi,x,\sigma}^{(h)\varepsilon}}\Big] \;.
\end{equation}
Recall the definition \eqref{eq:Gibbs_state_definition}:
\begin{equation}\label{eq:append_expectation_P}
\langle P \rangle_{N} = \int \big[\prod_{h=0}^{N-1} d\mu(\zeta^{(h)})\big]\,  e^{-V(\Phi)} P(\Phi) \;.
\end{equation}
%
To begin, the identity 
\begin{equation}
Q^{(h')}_{x'} e^{-i \sum_{\sharp = \phi, \psi} \sum_{\sigma = \uparrow\downarrow} \zeta^{(h)+}_{\sharp, x,\sigma} \zeta^{(h)-}_{\sharp, x,\sigma}} = 0
\end{equation}
implies that:
\begin{equation}
Q\left ( \prod_{h = 0}^{N-1} \prod_{x \in \Lambda^{(h+1)}}\, e^{-i \sum_{\sharp = \phi, \psi} \sum_{\sigma = \uparrow\downarrow} \zeta^{(h)+}_{\sharp, x,\sigma} \zeta^{(h)-}_{\sharp, x,\sigma}}\right ) = 0 \;.
\end{equation}
Furthermore, we have also have that $Q \, e^{-V(\Phi)} = 0$. Indeed, since $e^{-V(\Phi)} = \prod _{x \in \Lambda^{(0)}} f(\Phi_{x})$, where $f(\Phi_{x}) = e^{-\lambda (\Phi_{x} \cdot \Phi_{x})^{2} - i \mu (\Phi_{x} \cdot \Phi_{x})}$:
\begin{equation}\label{eq:Q_supersymmetry_Potential}
\begin{split}
Q f(\Phi_{x}) & = \left (\sum_{h = 0}^{N-1} Q_{\fl{L^{-h-1}x}}^{(h)} \right ) \, f\Big(\sum_{h=0}^{N-1} L^{-h} A_{\fl{L^{-h}x}}\zeta_{\fl{L^{-h-1}x}}^{(h)}\Big)
\\
& \equiv Q^{\Phi}_{x} f(\Phi_{x}) = 0 \;,
\end{split}
\end{equation}
where $Q^{\Phi}_{x} := \sum_{\sigma, \varepsilon} [ \psi_{x,\sigma}^{\varepsilon} \frac{\partial}{\partial \phi_{x,\sigma}^{\varepsilon}} - \varepsilon
\phi_{x,\sigma}^{\varepsilon} \frac{\partial}{\partial \psi_{x,\sigma}^{\varepsilon}}]$. The first equality follows by the fact that $\Phi_{x}$ depends only on $\zeta^{(h)}_{\fl{L^{-h-1}x}}$ for $h=0,...,N-1$; the second equality is obtained after repeated application of the identity \eqref{eq:claimSUSY}. 
Finally, $Q^{\Phi}_{x} f(\Phi_{x}) = 0$ by direct computation (recall that the function $(\Phi_{x} \cdot \Phi_{x})$ is $Q_{x}^{\Phi}$-supersymmetric) Thus, the localization theorem \cite{BT, SZ} implies that:
\begin{equation}\label{eq:P0}
\langle P \rangle_{N} = P(0)\;,
\end{equation}
provided that $e^{-V(\Phi)}P(\Phi)$ is of Schwartz type and that $Q \, P(\Phi)  = 0$. Taking $P(\Phi) = 1$, Eq. (\ref{eq:P0}) immediately implies $\mathcal{Z}_{N} = 1$. To conclude, let us consider ($\phi^{+} = \overline{\phi^{-}}$ here) 
\begin{equation}
P(\Phi) =  (\psi^{(\geq k)+}_{x,\sigma} \psi^{(\geq k) -}_{x,\sigma} + \phi^{(\geq k) +}_{x,\sigma} \phi^{(\geq k)-}_{x,\sigma})\;.
\end{equation}
%
The function $e^{-V(\Phi)}P(\Phi_{x}^{(\geq k)})$ is of Schwartz type. 
Moreover, by Eq. (\ref{eq:Q_supersymmetry_Potential}):
\begin{equation}
Q P(\Phi) = Q_{x}^{\Phi^{(\geq k)}} P(\Phi_{x}^{(\geq k)}) = 0\;,
\end{equation}
with $Q^{\Phi^{(\geq k)}}_{x} := \sum_{\sigma, \varepsilon} [ \psi_{x,\sigma}^{(\geq k)\varepsilon} \frac{\partial}{\partial \phi_{x,\sigma}^{(\geq k)\varepsilon}} - \varepsilon
\phi_{x,\sigma}^{(\geq k)\varepsilon} \frac{\partial}{\partial \psi_{x,\sigma}^{(\geq k)\varepsilon}}] $. Hence, by (\ref{eq:P0}), for $\fl{L^{-k+1}x} = \fl{L^{-k+1}y}$:
\begin{equation}
\langle (\phi^{(\geq k-1)+}_{\fl{L^{-k+1}x}} \phi^{(\geq k-1)-}_{\fl{L^{-k+1}y}} + \psi^{(\geq k-1)+}_{\fl{L^{-k+1}x}} \psi^{(\geq k-1)-}_{\fl{L^{-k+1}y}})\rangle_{N} = 0\;,
\end{equation}
which proves Eq. (\ref{eq:SUSYcorr}).
\end{itemize}
\end{remark}

\end{document}